\documentclass[11pt, letterpaper]{article}
\usepackage[margin=1in]{geometry}

\newcommand*\samethanks[1][\value{footnote}]{\footnotemark[#1]}

\usepackage{amsmath}
\usepackage{amssymb}
\usepackage{amsthm}

\usepackage{url}
\usepackage{xcolor}
\usepackage{paralist}
\usepackage{hyperref}
\usepackage{graphicx}
\usepackage{subcaption}

\usepackage[boxruled,linesnumbered]{algorithm2e}

\newtheorem{fact}{Fact}[section]
\newtheorem{claim}{Claim}[section]
\newtheorem{remark}{Remark}[section]

\newtheorem{theorem}{Theorem}[section]
\newtheorem{lemma}{Lemma}[section]
\newtheorem{corollary}{Corollary}[section]
\newtheorem{definition}{Definition}[section]
\newtheorem{proposition}{Proposition}[section]

\newcommand{\eps}{\varepsilon}
\newcommand{\e}{\eps}
\newcommand{\eat}[1]{}

\newcommand{\R}{\mathbb{R}}
\newcommand{\E}{\ensuremath{\mathbb{E}}}
\newcommand{\Z}{\mathbb{Z}}

\newcommand{\poly}{\operatorname{poly}}
\newcommand{\diam}{\ensuremath{\mathsf{diam}}\xspace}

\newcommand{\Exp}{\ensuremath{\mathsf{Exp}}\xspace}

\newcommand{\Dim}{\ensuremath{\mathsf{dim}}\xspace}
\newcommand{\DDim}{\ensuremath{\mathsf{ddim}}\xspace}
\newcommand{\PDim}{\ensuremath{\mathsf{pdim}}\xspace}

\newcommand{\Par}{\ensuremath{\mathsf{par}}\xspace}
\newcommand{\Des}{\ensuremath{\mathsf{des}}\xspace}
\newcommand{\cost}{\ensuremath{\mathsf{cost}}\xspace}
\newcommand{\range}{\ensuremath{\mathsf{range}}\xspace}
\newcommand{\ranges}{\ensuremath{\mathsf{ranges}}\xspace}

\newcommand{\Bdt}{\ensuremath{B^{\delta}}\xspace}
\newcommand{\Bd}{\ensuremath{B^{d}}\xspace}
\newcommand{\Bdn}{\ensuremath{B^{d_n}}\xspace}
\newcommand{\FL}{Feldman-Langberg\xspace}

\newcommand{\ignore}[1]{}

\ignore{

\newcommand{\shaofeng}[1]{{\footnotesize\color{purple}[Shaofeng: #1]}}
\newcommand{\xuan}[1]{{\footnotesize\color{magenta}[Xuan: #1]}}
\newcommand{\jian}[1]{{\footnotesize\color{blue}[Jian: #1]}}
\newcommand{\lingxiao}[1]{{\footnotesize\color{red}[LH: #1]}}

}

\newcommand{\shaofeng}[1]{}
\newcommand{\xuan}[1]{}
\newcommand{\jian}[1]{}
\newcommand{\lingxiao}[1]{}

\newcommand{\kdist}{{\mathcal{K}}}
\newcommand{\calA}{{\mathcal{A}}}
\newcommand{\calD}{{\mathcal{D}}}
\newcommand{\calF}{{\mathcal{F}}}
\newcommand{\calG}{{\mathcal{G}}}
\newcommand{\calH}{{\mathcal{H}}}

\newcommand{\calO}{{[X]^k}}
\newcommand{\calP}{{\mathcal{P}}}
\newcommand{\calR}{{\mathcal{R}}}
\newcommand{\calS}{{\mathcal{S}}}

\newcommand{\calZ}{{\mathcal{Z}}}

\title{	$\varepsilon$-Coresets for Clustering (with Outliers)
	in Doubling Metrics}
\date{}

\author{
	Lingxiao Huang\thanks{\'{E}cole polytechnique f\'{e}d\'{e}rale de Lausanne.~\texttt{lingxiao.huang@epfl.ch}}
		\and Shaofeng H.-C. Jiang\thanks{Weizmann Institute of Science.~\texttt{shaofeng.jiang@weizmann.ac.il}} 
		\and Jian Li\thanks{
		Institute for Interdisciplinary Information Sciences, Tsinghua University. \texttt{lijian83@mail.tsinghua.edu.cn, wu3412790@gmail.com}}
		\and Xuan Wu\samethanks
}

\ignore{
{
\author{Lingxiao Huang \footnotemark[1] \and Shaofeng H.-C. Jiang\footnotemark[2] \and Jian Li \footnotemark[3]  \and Xuan Wu \footnotemark[3] }

\footnotetext[1]{\'{E}cole polytechnique f\'{e}d\'{e}rale de Lausanne.~{\texttt{lingxiao.huang@epfl.ch}}}

\footnotetext[2]{The Weizmann Institute of Science.~{\texttt{shaofeng.jiang@weizmann.ac.il}}}

\footnotetext[3]{\mbox{Institute for Interdisciplinary Information Sciences, Tsinghua University.~{\texttt{lijian83@mail.tsinghua.edu.cn}}}}

\footnotetext[3]{\mbox{Institute for Interdisciplinary Information Sciences, Tsinghua University.~{\texttt{wu3412790@gmail.com}}}}
}
}

\begin{document}

\begin{titlepage}

\maketitle

\begin{abstract}
	\eat{
	\jian{Other possible titles? e.g.,
		
		(1) remove succinct?
		
		(2) (for combined paper)
		$\epsilon$-Coresets for clustering (with outliers)
		in doubling metrics
		
		(3) $\epsilon$-Coresets for Clustering Problems: Robustness and Doubling Metrics
		
		(4) Robust Coresets for Clustering Problem in Doubling Metrics.
		}
	}
	We study the problem of constructing $\eps$-coresets for the $(k, z)$-clustering problem in a doubling metric $M(X, d)$. An $\eps$-coreset is a weighted subset $S\subseteq X$ with weight function $w : S \rightarrow \mathbb{R}_{\geq 0}$, such that for any
	$k$-subset $C \in [X]^k$, it holds that $\sum_{x \in S}{w(x) \cdot d^z(x, C)} \in (1 \pm \eps) \cdot \sum_{x \in X}{d^z(x, C)}$.
	
	We present an efficient algorithm that constructs an $\eps$-coreset for the $(k, z)$-clustering problem in $M(X, d)$, where the size of the coreset only depends on the parameters $k, z, \epsilon$ and the doubling dimension $\DDim(M)$.
	To the best of our knowledge, this is the first efficient $\epsilon$-coreset construction of size \emph{independent} of $|X|$ for general clustering problems in doubling metrics.

	To this end, we establish the first relation between the doubling dimension
	of $M(X, d)$ and the shattering dimension (or VC-dimension) of the range space induced by the distance $d$.
	Such a relation is not known before, since one can easily construct instances in which
	neither one can be bounded by (some function of) the other.
	Surprisingly, we show that if we allow a small $(1\pm\epsilon)$-distortion
	of the distance function $d$ (the distorted distance is called the smoothed distance function),  the shattering dimension can be upper bounded by ${O}(\epsilon^{-O(\DDim(M))})$.
	For the purpose of coreset construction,
	the above bound does not suffice as it only works for unweighted spaces.
	Therefore, we introduce the notion of $\tau$-error {\em probabilistic shattering dimension},
    and prove a (drastically better) upper bound of
	$O( \DDim(M)\cdot \log(1/\eps) +\log\log{\frac{1}{\tau}} )$ for the probabilistic shattering dimension
	for weighted doubling metrics.
	As it turns out, an upper bound for the probabilistic shattering dimension
	is enough for constructing a small coreset.
	We believe the new relation between doubling and shattering dimensions
	is of independent interest and may find other applications.
	
	
	
	Furthermore, we study {\em robust coresets}
	for $(k,z)$-clustering with outliers
	in a doubling metric. We show an improved connection between $\alpha$-approximation and robust coresets. 
	This also leads to improvement upon the previous best known  bound of the size of 
	robust coreset for Euclidean space [Feldman and Langberg, STOC 11].
	The new bound entails a few new results in
	clustering and property testing.
	
	As another application, we show constant-sized $(\eps, k, z)$-centroid sets in doubling metrics can be constructed by extending our coreset construction.
	Prior to our result, constant-sized centroid sets for general clustering problems were only known for Euclidean spaces.
	We can apply our centroid set to accelerate the local search algorithm (studied in [Friggstad et al., FOCS 2016]) for the $(k, z)$-clustering problem in doubling metrics.
	
\end{abstract}

\thispagestyle{empty}
 \end{titlepage}

\section{Introduction}\label{section:intro}


We study the $(k, z)$-clustering problem in a metric space $M(X, d)$.
In the $(k, z)$-clustering problem, the objective is to find a $k$-subset $C \in [X]^k$
(which we call the set of centers), such that the objective function $\kdist_z(X, C) := \sum_{x \in X}{d^{z}(x, C)}$ is minimized, where $d(x, C) := \min_{y \in C}{d(x, y)}$.
The $(k,z)$-clustering problem is a general and fundamental problem
in many areas including approximation algorithms, unsupervised learning and computational geometry~\cite{lloyd1982least,tan2006cluster,arthur2007k,coates2012learning}.
In particular, $(k,1)$-clustering is the well known $k$-median problem,
$(k,2)$-clustering the $k$-means problem, and $(k,\infty)$-clustering the $k$-center problem.

\vspace{0.1cm}

\noindent\textbf{Coresets.}
A powerful technique for solving the $(k, z)$-clustering problem is to construct coresets~\cite{harpeled2004on,chen2006k,FL11,feldman2013turning}.
A coreset is a weighted subset of the point set, such that
for any set of $k$ centers,
the objective function computed from the coreset is approximately the same as
that computed from all points in $X$.
Hence, a coreset can be used as proxy for the full data
set: one can apply the same algorithm on the coreset, and the
result on the coreset approximates that on the full data set.
\begin{definition}
An $\eps$-coreset for the $(k, z)$-clustering problem
in metric space $M(X,d)$ is a weighted subset $S$ of $X$ with weight $w : S \rightarrow \mathbb{R}_{\geq 0}$
\footnote{
	Some previous work needs negative weights,
	but we only need nonnegative weights.
	}, such that for any $k$-subset $C \in [X]^k$,
	$$
	\sum_{x \in S}{w(x) \cdot d^z(x, C)} \in (1 \pm \eps) \cdot \kdist_z(X, C).
	$$
\end{definition}

Typically, we require that the size of the coreset depends on
$1/\eps$, $k$ and $z$ (independent of $|X|$).
Apparently, a small coreset is much cheaper to store and can be used to estimate the objective function more efficiently. In fact, constructing coresets can be useful in
designing more efficient approximation algorithms for many clustering problems, with various
constraints and outliers~\cite{FL11,feldman2012data,feldman2013turning,braverman2016new,DBLP:conf/focs/FriggstadRS16,lucic2017training}.

\vspace{0.1cm}
\noindent\textbf{Doubling Metrics.}
In this paper, we mainly consider
metric spaces with bounded \emph{doubling dimension}~\cite{Assouad83,DBLP:conf/focs/GuptaKL03}.
The doubling dimension of a metric space $M$, denoted as $\DDim(M)$, is the smallest integer $t$ such that any ball can be covered by at most $2^t$ balls of half the radius.
A doubling metric is a metric space of bounded doubling dimension.
The doubling dimension measures the intrinsic dimensionality of a general metric space, and it generalizes the dimension of normed vector spaces, where $t$-dimensional $\ell_p$ space has doubling dimension $O(t)$~\cite{Assouad83}.

Many problems have been studied in doubling metrics, such as spanners~\cite{DBLP:conf/compgeom/GaoGN04,CG06,DBLP:conf/soda/GottliebR08,DBLP:conf/esa/GottliebR08,DBLP:journals/dcg/ChanG09,DBLP:journals/siamcomp/ChanLNS15,DBLP:journals/algorithmica/ChanLN15,DBLP:conf/stoc/Solomon14,DBLP:journals/talg/ChanGMZ16},
metric embedding~\cite{DBLP:conf/focs/GuptaKL03,DBLP:conf/stoc/AbrahamBN06,DBLP:journals/jacm/ChanGT10},
nearest neighbor search~\cite{Cla99,IndykNaor,DBLP:conf/compgeom/Har-PeledM05},
and approximation algorithms~\cite{DBLP:conf/stoc/Talwar04,DBLP:journals/siamcomp/BartalGK16,DBLP:journals/dcg/ChanE11,DBLP:journals/talg/ChanJ18,DBLP:conf/focs/ChanHJ16,DBLP:conf/focs/FriggstadRS16}.
Apart from the above work,  
some machine learning problems have also been studied in the context of doubling metrics~\cite{DBLP:journals/jcss/BshoutyLL09,DBLP:journals/tit/GottliebKK14}.
However, to the best of our knowledge, no previous work has studied constructing coresets in doubling metrics.

\subsection{Our Results}
Our main result is an efficient construction of $\eps$-coresets for the $(k, z)$-clustering problem in doubling metrics. The size of our coreset does not depend on the number of input points. Moreover, both the running time and the size
of the coreset depend {\em polynomially} on the doubling dimension and $k$.
The result is stated in the following theorem.
\begin{theorem} (informal version of Theorem~\ref{thm:coreset})
	\label{theorem:main}
	Consider a metric space $M(X,d)$ with $n$ points.
	Let real numbers $0<\e,\tau<1/100$, $z > 0$, and integer $k \geq 1$.
	There exists an algorithm running in $\poly(n)$ time (assuming oracle access to the distance function),
	that constructs an $\eps$-coreset of size $\tilde{O}(2^{O(z\log z)}\cdot k^3\cdot \DDim(M)/\eps^2)$ for the $(k,z)$-clustering problem with probability at least $1-\tau$.
\end{theorem}

A first natural attempt is to
embed the doubling space to the Euclidean space and
use the existing Euclidean construction.
As shown in~\cite[Theorem 4.5]{DBLP:conf/focs/GuptaKL03}, for a doubling metric $M(X, d)$, it is possible to embed $d^{\frac{1}{2}}$ to an $O(\DDim(M)\cdot\log {\DDim(M)})$-dimensional $\ell_2$ space with $O(\DDim(M))$-distortion. Then an $\eps$-coreset for $(k, 2z)$-clustering problem in $\ell_2$ would imply an $O(\DDim(M)^z)$-coreset for the $(k, z)$-clustering problem in $M$.
However, it is generally not possible to embed $(X, d^{\frac{1}{2}})$ into $\ell_2$ with $(1+\eps)$-distortion
for an arbitrarily small constant $\eps>0$ and doubling metric $M(X, d)$ (where an example can be found in Proposition~\ref{prop:expander}).
Hence, in order to construct an $\eps$-coreset
in a doubling metric, we need new ideas.

A by now standard technique for constructing small coresets for clustering problems is importance
sampling, developed in a series of work \cite{langberg2010universal, FL11,varadarajan2012sensitivity}.
In particular, by the framework in~\cite{FL11,varadarajan2012sensitivity},  one can obtain an $\eps$-coreset by taking $\tilde{O}(2^{O(z\log z)}\cdot k^3\cdot \Dim/\eps^2)$ samples (see Section~\ref{sec:coreset} for more details).
Here $\Dim$ is the (shattering) dimension of the range space induced
by the distance function.
(i.e., the range space consists of all balls of different radii
\footnote{
In fact, we will deal with the range space in a certain
function space. See Section~\ref{section:dim_range_doubling} for the precise definition.}
).
Hence, if one can show that
$\Dim$ is bounded by some function of $\DDim(M)$,
the construction of an $\eps$-coreset would be finished.

\vspace{0.1cm}
\noindent\textbf{Doubling Dimension and Shattering Dimension.}
Now, we discuss the relation between the doubling dimension
$\DDim(M)$ and the (shattering) dimension $\Dim$ of the range space.
While the dimension $\Dim$ measures
the combinatorial complexity of the metric space, doubling dimension $\DDim(M)$ is the intrinsic geometric dimension of the metric space.
They both generalize the ordinary Euclidean dimension, but from different perspectives. In particular, for $\mathbb{R}^d$, both the (shattering) dimension
and the doubling dimension are $O(d)$.
Although both dimensions are subjects of extensive research,
to the best of our knowledge, there is no nontrivial relation known between the two.
This may not be a surprise, as we can easily construct a doubling metric, which has unbounded shattering dimension on the corresponding induced range space (see Theorem~\ref{theorem:loose_bound}).
The other direction cannot be bounded neither
\footnote{Consider a star with $n$ leaves.
It is immediate that the metric induced by the star has doubling dimension $\Omega(\log{n})$.
However, the shattering dimension of the range space induced by the star metric is $O(1)$.
}.
Hence, studying their relation may appear to be hopeless.
However, we observe that in the bad instance in Theorem~\ref{theorem:loose_bound},
if we allow a $(1\pm \eps)$-distortion to the distance function $d$,
then the instance actually has a small shattering dimension.

Inspired by this observation, we introduce the {\em smoothed distance function}.
A $\eps$-smoothed distance function $\delta: X\times X\rightarrow \R_{\geq 0}$ satisfies $\delta(x, y) \in (1 \pm O(\eps))\cdot d(x, y)$ for all $x, y \in X$. Basically, it is a small perturbation of the original distance function $d$.
We show, somewhat surprisingly, that if we use a certain smoothed distance function $\delta$,
defined by a hierarchical net of the doubling metric,
the shattering dimension of the range space (induced by $\delta$, instead of $d$)
can be upper bounded by
some function of the doubling dimension $O(\DDim(M))$, as in the following theorem.

\begin{theorem} (informal, unweighted case)
	\label{theorem:informal_loose_bound}
	Suppose $M(X,d)$ is a metric space.
	Let $0< \e\leq \frac{1}{8}$ be a constant. There is some
	$\eps$-smoothed distance function such that
    $\Dim(\calF) \leq O(1/\eps)^{\DDim(M)}$,
    where $\calF := \left\{ \delta(x, \cdot)  \mid x\in X\right\}$ is the
    set of $\eps$-smoothed distance functions.	
\end{theorem}


While the above theorem is encouraging,
there are still some drawbacks.
First, the dimension bound is exponential in $\DDim(M)$
(in contrast to the linear dependency in Euclidean case).
It is a natural question whether one can obtain a better bound in general.
More importantly, the above bound is not sufficient for the purpose
of constructing small coresets, for which we need a dimension bound for weighted spaces.
Unfortunately, it seems difficult to extend the proof of Theorem~\ref{theorem:loose_bound} to the weighted case.

\noindent
{\bf Weighted Space and Probabilistic Shattering Dimension.}
Recall that our goal is to construct $\eps$-coresets for the $(k,z)$-clustering problem.
According to the framework~\cite{FL11}, we shall consider the
set of \emph{weighted} distance functions $g_x : [X]^k \rightarrow \mathbb{R}_{\geq 0}$ for each point $x \in X$, defined as $g_x(C) := w(x) \cdot \delta^z(x, C)$ for $C\in [X]^k$, where $w : X \rightarrow \mathbb{R}_{\geq 0}$ is a weight function and $\delta$ is an $O(\eps/z)$-smoothed distance function.
We consider the function set $\calG:=\{ g_x \mid x \in X \}$,
and would like to
show that there is a subset $\calS\subseteq \calG$ such that $\calS$ is an $\alpha$-approximation for the range space of $\calG$ for a certain constant
$\alpha$.
Then we can apply~\cite[Theorem 4.1]{FL11} (restated in Theorem~\ref{thm:fl}) to show that one can efficiently find an $\eps(\alpha)$-coreset of size $|\calS|$.

In order to prove an $\alpha$-approximation result, it suffices to bound
the shattering dimension of $\calG$.
We recall that
$\calG = \{g_x \mid x\in X\}$ with $g_x(C) = w(x) \cdot \delta^z(x, C)$ for $C \in [X]^k$.
%
Let $\calF:=\left\{ f_x\mid x\in X \right\}$ be a collection of functions $f_x(y):=w(x)\cdot \delta^z(x, y)$ for $y\in X$. Note that the difference between $\calF$ and $\calG$ is that, the ground set of $\calF$
consists of singletons and that of $\calG$
contains $k$-subsets.
By a simple argument (see Claim~\ref{claim:min_equal}),
one can show that roughly the shattering dimension of $\calG$
is at most $k$ times of the shattering dimension of $\calF$.
%
Hence, the key is to bound the shattering dimension of $\calF$ (the set of weighted smoothed distance function with ground set $X$).
It turns out the proof for the weighted case is much more involved than the unweighted case.
Instead of using a deterministic $\delta$ defined with respect to $d$, we introduce a random smoothed distance function,
defined on top of a randomized hierarchical decomposition introduced by~\cite{DBLP:conf/stoc/AbrahamBN06} in the doubling metric.
A key property of the randomized hierarchical decomposition is that a set with small diameter is cut by a large cluster in the decomposition with very small probability.
Intuitively, the property enhances the smooth property, so that we can still hang the balls centered at any $x\in X$ to a net point of higher layer in the weighted space.

Consider an arbitrary fixed $H\subseteq X$ and $\calF_H := \{f_x \in \calF \mid x \in H\}$.
Due to the randomness in the randomized hierarchical decomposition,
we can only show $|\ranges(\calF_H)|$ is bounded with \emph{constant} (close to 1) probability.
Hence, we introduce the notion of {\em probabilistic shattering dimension}
for the range space induced by a family of random functions (formally in Definition~\ref{def:pdim}):
for any subset of (random) functions $\calF_H$, if the probability that
$|\ranges(\calF_H)|\leq O(|\calF_H|^t)$ with probability $1-\tau$
(note that
$|\ranges(\calF_H)|$ is a random variable), we say that
the probabilistic shattering dimension $\PDim_\tau(\calF)$ of the range space is $t$.
Our main technical result is the following theorem.

\begin{theorem} (informal, weighted case)
	\label{theorem:informal_weighted}
	Suppose $M(X, d)$ is a metric space together with a gap-2 weight function (Definition~\ref{definition:weight}) $w: X\rightarrow \R_{\geq 0}$.
	Let $0<\eps\leq 1/100 z$ and $0 < \tau < 1$ be constants.
	There exists a random $\eps$-smoothed distance function $\delta$ and a collection $\calF:=\left\{ f_x=w(x)\cdot \delta^z(x, \cdot)\mid x\in X \right\}$, such that
	the following holds:
	for any fixed $H\subseteq X$ and $\calF_H := \left\{ f_x:x\in H \right\}$,
	we have
	\begin{align*}
	\Pr_{\delta}\left[|\ranges(\calF_H)| \leq O\left(\eps^{O(-\DDim(M))} \cdot \log\frac{1}{\tau}\cdot \poly(|H|)\right) \right] \geq 1 - \tau.
	\end{align*}
	In other words, $\PDim_\tau(\calF)= {O}(\DDim(M)\cdot \log (1/\eps) + \log\log{1/\tau}) $.
\end{theorem}

The above theorem drastically improves the dimension from $O(1/\eps)^{\DDim(M)}$ (in Theorem~\ref{theorem:loose_bound})
to $\tilde{O}\left(\DDim(M)\cdot \log (1/\eps)\right)$ (albeit with a weaker probabilistic guarantee).
Note that one cannot afford to apply a union bound over all different $H$'s to show that $\dim(\calF)$ is bounded.
Hence, the bound on the probabilistic shattering dimension does not directly lead to an $\alpha$-approximation by the standard PAC learning theory.
However, we prove in Lemma~\ref{lemma:restate_weak_app}
a probabilistic analogue of the $\alpha$-approximation lemma,
which only requires a bounded probabilistic shattering dimension.
%
%
%



\vspace{0.1cm}
\noindent
{\bf Robust Coreset.}
We also consider robust coresets which are coresets for $(k,z)$-clustering problems with outliers.
The notion of robust coreset was first introduced in \cite{FL11}.
In the following, we give the definition of robust coreset for the $(k,z)$-clustering problem.

\begin{definition}[robust coresets]
	\label{def:robust_coreset}
	Let $M(X,d)$ be a metric space.
	Let $0<\gamma\leq 1$, $0\leq \eps, \alpha\leq \frac{1}{4}$, $k \geq 1$ and $z>0$.
	For any $H\subseteq X$ and $C\in [X]^k$, let
	\[
	\kdist_z^{-\gamma}(H, C):= \min_{H'\subseteq H: |H'|= \lceil(1-\gamma)|H|\rceil} \sum_{x\in H'} d^z(x,C)
	\]
	denote the sum of the smallest $\lceil (1-\gamma) |H| \rceil$ values $d^z(x,C)$ over $x\in H$ (i.e., we exclude the largest $\gamma|H|$ values as outliers).
	An $(\alpha,\eps)$-robust coreset for the $(k,z)$-clustering problem with outliers is a subset $S\subseteq X$ such that for any $k$-subset $C\in [X]^k$ and any $\alpha< \gamma< 1-\alpha$,
	\[
	(1-\eps)\cdot \frac{\kdist_z^{-(\gamma+\alpha)}(X, C)}{|X|} \leq \frac{\kdist_z^{-\gamma}(S, C)}{|S|} \leq (1+\eps)\cdot \frac{\kdist_z^{-(\gamma-\alpha)}(X, C)}{|X|}.
	\]
	
\end{definition}
%

\eat{
For any $C\in [X]^k$ and $\gamma\in [0,1]$, let $\mathcal{G}_{C}^{-\gamma}$ denote the set of $(1-\gamma)|\calG|$ functions $f_x\in \calG$ with the smallest function value $f_x(C)$. $\mathcal{D}\subset \calG$ is called an $(\eps,z)$-robust coreset of $\calG$ if for every $\gamma\in (\eps,1-\eps)$ and $C\in [X]^k$, $$
$$

If $\mathcal{D}$ is an $(\eps,z)$-robust coreset of $\calG$, we call $S$ an $\eps$-robust coreset of $X$ for $(k,z)$-clustering.
}


Our result for robust coreset for $(k,z)$-clustering is presented in the following
theorem, which
generalizes and improves the prior result in~\cite{FL11} for Euclidean space.
%

\begin{theorem} [informal, robust coreset]
	Let $M(X,d)$ be a doubling metric (a $d$-dimensional Euclidean space resp.). Let $S$ be a uniform sample of size
	$\tilde{O}(k\cdot\mathrm{ddim}(M)/\alpha^2)$ ($\tilde{O}(kd/\alpha^2)$ resp.) from $X$.
	Then with constant probability, $S$ is an $(\alpha,\eps)$-robust coreset ($(\alpha,0)$-robust coreset resp.)
	for the $(k,z)$-clustering problem with outliers.
\end{theorem}

The definition of robust coreset in~\cite{FL11} is slightly different from ours.
\footnote{
	In ~\cite[Definition 8.1]{FL11},
	$S\subset X$ is called a $(\gamma,\eps)$-coreset if
	for every $C\in [X]^k$, $\gamma_1\geq \gamma$ and
	$\eps_1\geq \eps$,	$
	(1-\eps_1)\cdot \frac{1}{|X|}\kdist_1^{-(1-\gamma_1+\eps_1\gamma_1)}(X, C) \leq \frac{1}{|S|}\kdist_1^{-(1-\gamma_1)}(S, C) \leq (1+\eps_1)\cdot \frac{1}{|X|}\kdist_1^{-(1-\gamma_1-\eps_1\gamma_1)}(X, C).
	$}
One can directly check that in Euclidean space, an $(\gamma \eps/4,0)$-robust coreset in Definition \ref{def:robust_coreset} is an $(\gamma,\eps)$-coreset in ~\cite[Definition 8.1]{FL11}. Thus the above theorem improves the size of $(\gamma,\eps)$-coreset in \cite[Corollary 8.4]{FL11} from $\tilde{O}(kd\gamma^{-2}\eps^{-4})$ to $\tilde{O}(kd\gamma^{-2}\eps^{-2})$.


\eat {
The technique of~\cite{FL11} can only construct an $(\alpha,\eps)$-robust coreset of size $\tilde{O}(k\cdot\mathrm{ddim}(M)\cdot \log (z/\eps)/\eps^2\alpha^2)$.
Compared to~\cite{FL11}, we improve the size by saving an $\eps^{-2}$ factor.
\jian{I can't understand this sentence. please state the result of FL11.
	where is the $\eps^{-2}$ factor? such a statement is not good enough.}
\lingxiao{I try to change the statement. @Xuan, please check whether it makes sense.} \xuan{I think we should change the statement. \cite{FL} cannot apply in doubling metrics and they require a lower bound for $\gamma$. In doubling metrics, their technique cannot imply even weaker bound for robust coreset defined here } }

Furthermore, we demonstrate an application of robust coresets in property testing (Section \ref{PT}). 
Our testing for $(k,z)$-clustering
problem is in the same spirit as the testing for 
$k$-center problem proposed in Alon et al.~\cite{alon2003testing}.
We design a simple testing algorithm for $(k,z)$-clustering.

Constructing robust coresets is also a useful subroutine in several other problems,
such as robust median and bi-criteria approximation for projective clustering (see \cite{FL11}). Hence, our improvement may lead to certain improvements of
these problems as well. Since this is not the focus of the this paper, we do not
go into the details.


\vspace{0.1cm}
\noindent
{\bf Centroid Set.}
We also consider a notion closely related to coreset, called centroid set.
Roughly speaking, a centroid set can be viewed as a coreset that
contains an $(1+ \eps)$-approximate solution (which is a $k$-subset) to the clustering objective (see Definition~\ref{definition:centroid_set}).
Applying our coreset result, we show the existence of succinct centroid sets in doubling metrics, which is presented in the following theorem.
To the best of our knowledge, this is the first result
on centroid sets beyond Euclidean spaces.
%
%
	\begin{theorem}[informal, centroid set]
	Let $M(X,d)$ be a metric space of $n$ discrete points.
	Let $S$ be an $\frac{\eps}{2}$-coreset for the $(k,z)$-clustering problem on $X$.
	There is an algorithm running in $\poly(n)$ time, that finds a centroid set of size at most $ (\frac{z}{\eps})^{O(\DDim(M))}\cdot |S|^2$.
\end{theorem}

Applying the above theorem, we can also accelerate the local search algorithm~\cite{DBLP:conf/focs/FriggstadRS16} for $(k,z)$-clustering in doubling metrics, from $n^{O(\rho)}$ to $(2^{O(z\log z)}\cdot \frac{k}{\eps})^{O(\rho)}$ running time per iteration, where $n=|X|$ is the number of points and $\rho := \rho(\epsilon, \DDim(M), z)$ is a large constant (depending only on $\epsilon, \DDim(M), z$).




\subsection{Overview of Our Techniques}
\vspace{0.1cm}
\noindent\textbf{The \FL Framework~\cite{FL11}.}
Our coreset construction makes use of the the framework of~\cite{FL11},
which we briefly discuss below.
Let $[X]^k$ be the ground set (the set of $k$-tuples) and $\delta$ be an $O(\eps/z)$-smoothed distance function.
For the $(k, z)$-clustering problem, assign a \emph{weighted} distance function $g_x : [X]^k \rightarrow \mathbb{R}_{\geq 0}$ to each point $x \in X$, such that $g_x(C) := w(x) \cdot \delta^z(x, C)$ for $C\in [X]^k$,
where $w : X \rightarrow \mathbb{R}_{\geq 0}$ is a weight function.
Consider the function set $\calG := \{ g_x \mid x \in X \}$.
The range spaces of $\calG$ is defined as $(\calG, \ranges(\calG))$, where $\ranges(\calG) := \{ \range(\calG, C, r) \mid C \in [X]^k, r \geq 0 \}$ and $\range(\calG, C, r) := \{ g_x \in \calG \mid g_x(C) \leq r \}$.
To interpret the definition, one can think of $\{ g_x\in \calG \mid g_x(C) \leq r \}$ as a ball of functions in $\calG$ that is centered at $C$ with radius $r$, and the distance from $C$ to $g_x \in \calG$ is measured as $g_x(C)$.
In the unweighted case, $\range(\calG, x, r)$ indeed corresponds to a ball in the metric space.
So $|\ranges(\calG)|$ counts the number of distinct balls (of functions in $\calG$) that may be formed by any center in the ground set and radii.

Recall that a subset $\calS\subseteq \calG$ is an $\alpha$-approximation for
the range space $(\calG, \ranges(\calG))$ if for any $\calR \in \ranges(\calG)$,
$
\bigl| |\calR|/|\calG| - |\calR \cap \calS|/|\calS| \bigr| \leq \alpha.
$
In other words, $\calS$ can be used as a good estimator for the density of $\calR$ relative to $\calG$.
It is shown in \cite[Theorem 4.1]{FL11} (restated in Theorem~\ref{thm:fl}) that,
if there is a subset $\calS\subseteq \calG$ such that $\calS$ is an $\alpha$-approximation for the range space of $\calG$,
then we can efficiently construct an $\eps(\alpha)$-coreset $S\subseteq X$ of size $|\calS|$ in $M(X,\delta)$.
Since $\delta$ is a small perturbation of the original distance function $d$, $S$ is also an $\eps(\alpha)$-coreset in $M(X,d)$.
Constructing an $\alpha$-approximation of small size is extensively studied in the PAC learning theory.
In particular, if a range space has bounded shattering (or VC) dimension, then a small sample (whose size depends on $\alpha$ and shattering dimension) from the set of functions would be an $\alpha$-approximation with constant probability
(see e.g.,~\cite{DBLP:journals/jcss/LiLS01}).
Hence, if $\calG$ has bounded shattering dimension, we can apply the existing $\alpha$-approximation construction.
This is also the approach taken in~\cite{FL11}.

\vspace{0.1cm}
\noindent
{\bf $\alpha$-Approximation.}
As one can imagine, in order to obtain an $\alpha$-approximation, we would like to apply Theorem~\ref{theorem:informal_weighted} (to bound the shattering dimension of
$(\calG, \ranges(\calG))$).
More precisely, for any $H\subseteq X$, we want to bound $|\ranges(\calG_H)|$ where $\calG_H:=\left\{g_x\in \calG\mid x\in H \right\}$.
However, the ground set of $\calG_H$ is the set of $k$-subsets $[X]^k$ (with the distance function in $\calG_H$ defined as $g_x(C) := w(x) \cdot \delta^z(x, C)$),
but the ground set of $\calF_H$ in Theorem~\ref{theorem:informal_weighted} is the point set $X$ of the metric space $M(X,d)$.
This is easy to handle: one can show that
$|\ranges(\calG_H)|\leq |\ranges(\calF_H)|^k$ (see Claim~\ref{claim:min_equal}).
Hence, we only need to bound $|\ranges(\calF_H)|$.
%
Another problem is that the bound for $|\ranges(\calF_H)|$ only holds with constant probability.
As a result, we cannot directly use the standard $\alpha$-approximation result.
In Lemma~\ref{lemma:restate_weak_app}, we introduce
a probabilistic analogue of the $\alpha$-approximation lemma from the PAC learning theory,
which only requires a bounded probabilistic shattering dimension.

Our proof borrows the classical double sampling idea from the construction of $\alpha$-nets in the PAC learning theory (see for example~\cite{kearns1994introduction}).
An obvious challenge is that we cannot afford to guarantee $|\ranges(\calF_H)|$ is small for many $H$ simultaneously by the union bound, which is required in the original proof.
Note that our guarantee has an additional randomness from $\delta$, and it is
important to take advantage the additional randomness.
We crucially use the fact that $\calF$ is actually indexed by $X$, that is, for each $x\in X$, a function $f_x \in \calF$ is generated by applying a \emph{random} map from $x$ to $w(x)\cdot\delta(x, \cdot)$.
This enables us to separate the two randomness, in a way that we view a sample from $\calF$, as firstly sampling from $X$ then applying a \emph{random} map on the sample.
The randomness of sampling from $X$ is used in the similar way as in the original proof, but the randomness of $\delta$ is used in another conditional probability argument to avoid the overlarge union bound. Full details of the proof are provided in Section~\ref{sec:approximation}.



\vspace{0.1cm}
\noindent
{\bf Doubling Dimension and Shattering Dimension.}
Now, we highlight some technical aspects of Theorems~\ref{theorem:informal_loose_bound}
and~\ref{theorem:informal_weighted},
which are the key technical contributions of this paper.
Our smoothed distance function $\delta$ is defined over the hierarchical net tree (see Section~\ref{subsec:doubling} for the definition)
of the doubling metric $M(X,d)$, i.e., for any $x,y\in X$, we define $\delta(x,y)$ to be the distance between their ancestors of a proper height in the hierarchical net tree (see Definition~\ref{section:smoothed_dis}).
We also define $B^{\delta}(x, r) := \left\{ y \in X \mid \delta(x, y) \leq r \right\}$ to be the ball of radius $r$ centered at $x\in X$ with respect to $\delta$.
%
We can show the smoothed distance function $\delta$ satisfies several
useful properties.
One is the \emph{smooth} property (Lemma~\ref{lemma:ball_equal}): roughly speaking,
for any radius $r$, $B^\delta(x, r) = B^\delta(u, r)$ for any point $x\in X$ and a nearby net point $u$ with higher height (relative to $r$). This intuitively means we can ``hang'' the center $x$ to the net point $u$.
Since the number of net points with higher height
is smaller, the smooth property greatly reduces the number of possible balls we need to consider.
Another important property is the cross-free property (Lemma~\ref{lemma:ball_laminar}), which implies that,
if we let $\calF$ be the range space induced by $\delta$,
then for any fixed $\calD \subseteq \calF$, for any ball $\range(\calF,x,r)$ in the range space, $\range(\calF,x,r) \cap \calD$ can be represented by a union of at most $\eps^{-O(\DDim(M))}$ subsets, and all these subsets are from a support of size $O(|\calD|)$.
This implies that at most $ |\calD|^{\eps^{-O(\DDim(M))}} $ possible subsets of $\calD$ can be formed by intersecting with balls of a fixed radius 
in the metric, which is the main observation used in our proof for Theorem~\ref{theorem:informal_loose_bound}.

Unfortunately, it seems difficult to
extend the above idea to the weighted case,
and our proof for Theorem~\ref{theorem:informal_weighted} is much more involved.
We restrict our attention to the weight function
such that the set of distinct weights $\{w_1, w_2, \ldots, w_l\}$ satisfies $w_1\geq 2w_2\geq 4w_3\geq\ldots\geq 2^{l-1}w_l$
(this suffices for the coreset construction).
Fix a set $H\subseteq X$  and let
$H_i=\{x\in H\mid w(x)=w_i\}$.
Essentially, we need to bound the number of different ranges
$
\bigcup_{i\in [l]} B^{\delta}(x, r/w_i)\cap H_i\,\,
$
$(r> 0, x\in X)$.
For this purpose, we divide $[0,+\infty)$ (the range of $r$) into
at most $O(|H|^4)$ {\em critical intervals}. Inside each critical interval,
we enforce some invariance properties. For a critical interval
$[a,b)$ with $b\leq 2a$, we simply apply the packing property to bound the number of different ranges for $r\in [a,b)$.
%
For a critical interval $[a,b)$ with $b \gg a$,
we need to use the
randomized hierarchical decomposition
developed in~\cite{DBLP:conf/stoc/AbrahamBN06} to enhance the smooth property.
We provide a more detailed overview of the proof in
Section~\ref{section:weighted_overview} when all necessary notations are
available.

\vspace{0.1cm}
\noindent
{\bf Robust Coreset.} We prove an improved
connection between $\alpha$-approximation and robust coreset,
which improves the one in \cite[Theorem 8.3]{FL11}.
Our proof is much simpler.
Combining with the $\alpha$-approximation result,
we can construct an $(\alpha,\eps)$-robust coreset in Euclidean space or doubling metrics.
The algorithm is extremely simple:
to take a uniform sample of size $\tilde{O}(kd/\alpha^2) $ or $\tilde{O}\big(k\cdot \DDim(M)/\alpha^2\big)$.

\section{Related Work}
\label{sec:related_work}

In the seminal paper \cite{agarwal2004approximating},
Agarwal et al. proposed
the notion of coresets for
the directional width problem (in which a coreset is called an $\epsilon$-kernel)
and several other geometric shape-fitting problems.
Since then, coresets have become increasingly more relevant
in the era of big data as they can reduce the size of a dataset with provable guarantee that
the answer on the coreset is a close approximation of the one on the whole dataset.
Many efficient algorithms for constructing small coresets
for clustering problems in Euclidean spaces are known (see e.g., ~\cite{agarwal2002exact,har2004clustering,chen2006k,harpeled2007smaller,langberg2010universal,FL11,feldman2013turning,braverman2016new}).
In particular, Feldman and Langberg~\cite{FL11}
(see their latest full version)
showed a construction for $\eps$-coresets of size $\tilde{O}(dk/\e^{2z})$ for general $(k,z)$-clustering problems with arbitrary $k$ and $z$, in $\tilde{O}(nk)$ time.
For the special case that $z=2$ which is the $k$-means clustering,
Braverman et al.~\cite{braverman2016new} improved the size to $\tilde{O}(k\min\left\{k/\eps, d \right\}/\eps^2)$,
which is \emph{independent} of the dimensionality $d$.
For another special case $z=\infty$, which is the $k$-center clustering, an $\eps$-coreset of size $O(k/\eps^d)$ can be constructed in $O(n+k/\eps^d)$ time, for $\mathbb{R}^d$~\cite{agarwal2002exact,har2004clustering}.
%
%
%
For general metrics, an $\eps$-coreset for the $(k, z)$-clustering problem of size $O(k\log n/\eps^{2z})$ can be constructed in time $\tilde{O}(nk)$~\cite{FL11}, and
for $k$-means clustering, Braverman et al.~\cite{braverman2016new} showed a construction of size $O(k\log k\log n/\eps^{2})$.
%
%
%
We also refer interested readers to Phillips's survey~\cite{phillips2016coresets}
for more construction algorithms as well as the applications of coresets in many other
areas.
%

%
%
%
%
%
%

Feldman and Langberg \cite{FL11} first studied the notion of robust coreset to handle the clustering problems with outliers. In $\R^d$, they showed how to construct a $(\gamma,\eps)$-coreset
\footnote{Note that their definition~\cite[Definition 8.1]{FL11} is similar but slightly different to ours.
However, considering the $(k,z)$-clustering problem with outliers, one can check that an $(\eps \gamma/4,\eps)$-robust coreset in our Definition \ref{def:robust_coreset} is a $(\gamma,\eps)$-coreset in~\cite[Definition 8.1]{FL11}. 
In fact, our defintion is more general. It is unclear whether their result applies to our definition.}
of size $\tilde{O}(kd\eps^{-4}\gamma^{-2})$ by uniform sampling.
We improve the bound to $\tilde{O}(kd\eps^{-2}\gamma^{-2})$.
Later, Feldman et al.~\cite{feldman2012data} developed another notion called \emph{weighted coreset} to handle outliers.
They used such coresets to design an $(1+\eps)$-approximation algorithm for the $k$-median problem with outliers.

Constructing coresets for clustering problems in Euclidean spaces has been also investigated 
in the streaming and distributed settings in the literature 
{e.g., \cite{FL11,feldman2013turning, DBLP:conf/nips/BalcanEL13,braverman2016new,DBLP:conf/icml/BravermanFLSY17}).
However, it is unclear how to define the streaming or distributed model 
in a general doubling metric, since 
there is no coordinate representation for each point
and we need all distances between the new coming point and the prior points.
Hence, in this paper, we focus on the centralized setting.


Besides unsupervised clustering problem, some supervised learning problems are 
also studied in the context of doubling metrics,
and the connections between doubling dimension and VC dimension (and closely related
notions) have been investigated in a variety of settings. 
Bshouty, Li and Long~\cite{DBLP:journals/jcss/BshoutyLL09} provided a generalization bound in terms of
the maximum of the doubling dimension and the VC-dimension of the hypothesis class $F$.
They also showed that the doubling dimension of metric $(F,d)$,
where the distance $d$ is defined as $d(f,g)=\Pr_x[f(x)\ne g(x)]$ for any two classifiers $f$ and $g$,
cannot be bounded by the VC-dimension of $F$ in general.
Gottlieb et al.~\cite{DBLP:journals/tit/GottliebKK14} studies the classification problem of points in a metric space, and obtained a generalization bound with respect to the doubling dimension.
Abraham et al.~\cite{DBLP:journals/jcss/BshoutyLL09}
introduced the concept of highway dimension (which is closely related to doubling dimension)
in the context of designing efficient shortest path algorithm,
and they showed that VC dimension and learning theory are also useful in this context.



Alon et al.~\cite{alon2003testing} first considered the property testing problem in the context of clustering.
In particular, they studied the testing algorithm for $k$-center clustering.
In this paper, we use robust coreset to develop a unified testing algorithm for $(k,z)$-clustering (for constant $k$ and $z$).
As pointed out in \cite{alon2003testing}, the testing algorithms can be converted into a sublinear time approximation algorithms for clustering with outliers.
One interesting benefit of such algorithms is that they can answer the query ``which cluster does a data point belong to'', without actually having to partition all the data points.

\section{Preliminaries}
\label{section:prelim}
Let $[m] := \{1, 2, \ldots, m\}$ for an integer $m \geq 1$.
For a function $f$ defined on some ground set $\mathcal{U}$ and $S\subseteq \mathcal{U}$, let $f(S) := \{ f(x) \mid x\in S \}$.
For a set $S$ and integer $k\geq 1$, let $[S]^k := \{ P \mid P \subseteq S , |P| = k\}$.
%
Consider a metric space $M(X, d)$.
Define $B^{\xi}(x, r) := \left\{ y \in X \mid \xi(x, y) \leq r \right\}$ to be the ball of radius $r$ centered at $x\in X$, with respect to some function $\xi : X\times X \rightarrow \mathbb{R}_{\geq 0}$.
For $S\subseteq X$ define the diameter of $S$ as $\diam(S) := \max_{x, y \in S}\left\{ d(x, y) \right\}$. 
For $S, T\subseteq X$, define $d(S, T) := \min_{x\in S,y \in T}{d(x, y)}$.

\eat{
\vspace{0.1cm}
\noindent\textbf{Coreset.} We consider the problem of constructing an $\epsilon$-coreset for the $(k, z)$-clustering problems in a metric space $M(X, d)$.
Let $\kdist_z(X, C) := \sum_{x \in X}{d^z(x, C)}$ denote the objective function for the $(k, z)$-clustering problem.
An $\eps$-coreset for the $(k, z)$ clustering problem is a weighted subset $S$ of $X$ with weight $w : S \rightarrow \mathbb{R}_{\geq 0}$, such that for any center $C \in [X]^k$,
\begin{align*}
	\sum_{x \in S}{w(x) \cdot d^z(x, C)} \in (1 \pm \eps) \cdot \kdist_z(X, C).
\end{align*}

\noindent\textbf{Robust Coreset.}
We also consider the problem of constructing small robust coresetS.
For any $H\subseteq X$ and $C\in [X]^k$, let 
$
\kdist_z^{-\tau}(H, C)
$
denote the sum of the smallest $\lceil (1-\tau) |H| \rceil$ values $d^z(x,O)$ for $x\in H$. 
An $(\alpha,\eps)$-robust coreset for the $(k,z)$-clustering problem with outliers is a subset $S\subseteq X$ such that for any $k$-subset $C\in [X]^k$ and any $\alpha< \tau< 1-\alpha$,
\[
(1-\eps)\cdot \frac{\kdist_z^{-(\tau+\alpha)}(X, C)}{|X|} \leq \frac{\kdist_z^{-\tau}(S, C)}{|S|} \leq (1+\eps)\cdot \frac{\kdist_z^{-(\tau-\alpha)}(X, C)}{|X|}.
\]
\jian{it doesn't make much sense to have one formal def in intro and another one here. for both coreset and robust coreset.
maybe a less formal ones in intro and formal def here.}
}

\subsection{Doubling Dimension and Hierarchical Nets}
\label{subsec:doubling}
\begin{definition}[doubling dimension]
	\label{def:doublingdim}
	A metric space has doubling dimension at most $t$, if any ball can be covered by at most $2^t$ balls of half the radius. The doubling dimension of a metric space $M$ is denoted as $\DDim(M)$. 
\end{definition}

\noindent\textbf{Covering, Packing and Net.}
Consider a subset of points $S\subseteq X$.
$S$ is a $\rho$-covering, if for any $x\in X$, there exists $y \in S$ such that $d(x, y) \leq \rho$.
$S$ is a $\rho$-packing, if for all $x, y \in S$, it holds that $d(x, y) \geq \rho$.
$S$ is a $\rho$-net, if $S$ is both a $\rho$-packing and a $\rho$-covering.
\begin{fact}[packing property. see. e.g.,~\cite{DBLP:conf/focs/GuptaKL03}]
	\label{fact:packing}
	Given a metric space $M(X,d)$, if $S\subseteq X$ is a $\rho$-packing then
	$|S| \leq (\frac{2\cdot \diam(S)}{\rho})^{\DDim(M)}$.
\end{fact}

\noindent\textbf{Hierarchical Nets and Net Trees.}
Now, we introduce some useful concepts that are 
well known in the doubling metric literature (see e.g.~\cite{DBLP:conf/stoc/Talwar04,DBLP:journals/talg/ChanGMZ16}).
Rescale the metric such that the minimum intra-point distance is $1$.
Suppose the diameter of the space is between $[2^{L-1}, 2^L)$.
Construct nets $N_L \subseteq N_{L-1} \subseteq \ldots \subseteq N_1 \subseteq N_0 = N_{-1} = \ldots = N_{-\infty} =  X$, where $N_i$ is a $2^{i}$-net of $N_{i-1}$.
The set of nets $\{N_i \mid i \leq L\}$ is called a hierarchical net.

We identify a point $u\in N_i$ in the tree by $u^{(i)}$ for $i \leq L$.
Note that the same point may belong to several $N_i$'s, but they have different identities.
A net tree is a rooted tree with node set $\left\{ u^{(i)} \mid i  \leq L , u \in N_i \right\}$, and the root is defined as the only node in $N_L$ (observing that $|N_L|=1$).
For each $u\in N_i$, $u^{(i)}$ has a unique parent node $v^{(i+1)}$ such that $ v\in N_{i+1}$, and we denote $\Par(u^{(i)}) = v^{(i+1)}$. 

For a net tree, define $\Des(u^{(i)}) \subseteq X$ to be the set of points in the metric space corresponding to descendants of $u^{(i)} \in N_i$ in the net tree.
For a leaf node $x\in X$, define $\Par^{(i)}(x)$ to be the ancestor of $x$ in $N_i$.

\begin{definition}[$c$-covering net trees]
	\label{def:coveringnettree}
	A net tree is $c$-covering ($c\geq 1$), if for each height $i$ and each $u \in N_i$, it holds that $d(u, \Par(u^{(i)})) \leq c \cdot 2^{i+1}$.
\end{definition}

\noindent
The following fact is immediate from Definition~\ref{def:coveringnettree}.

\begin{fact}
	\label{fact:des_dis}
	In a $c$-covering net tree, for each $x\in X$ it holds that $d(x, \Par^{(i)}(x)) \leq c \cdot 2^{i+1}$.
\end{fact}

\subsection{Range Space, Shattering Dimension and $\alpha$-Approximation}

\ignore{

As illustrated in Section~\ref{section:intro}, for each point $x\in X$, we assign it a weight $w(x) \geq 0$, and consider a function $f_x : X \rightarrow \mathbb{R}_+$ as $f_x(y) := w(x) \cdot d(x, y)$ for $y\in X$. In other words, $f_x$ is the weighted distance from $x$ to other points.

}
We adopt the function representation used in~\cite[Definition 7.2]{FL11}, but 
specifically tailored to our own needs.
In particular, since we focus on the clustering problems
in a doubling metric $M(X,d)$, the ground set is $[X]^k$ (the set of $k$-subsets) 
throughout the paper.
When $k=1$, we use $X$ to represent $[X]^1$ for simplicity.
%

\noindent\textbf{Indexed Function Sets.}
As seen in Section~\ref{section:intro}, we mainly focus on range spaces induced by a metric space. 
Hence we always consider \emph{indexed} function sets.
A set of functions $\calF$ is called {\em indexed}, 
if there exists an index set $V$ such that $\calF = \{ f_x \mid x \in V \}$.
In most cases, we simply use $V=X$ as the index set. 
We will make necessary clarification when we use other index set.
For an indexed function set $\calF$, define $\calF_H := \{ f_x \mid x\in H \}$ for a subset $H\subseteq V$ of the index set.
There are technical reasons to consider the indexed function set (rather than a general set of functions). 
See Remark~\ref{remark:indexed_func}.

\noindent\textbf{Range Space.}
Let $\calF$ be an indexed function set.
Define $\range(\calF, C, r) := \left\{ f_x \in \calF \mid f_x(C) \leq r \right\}$ for $C\in [X]^k, r \geq 0$.
Define $\ranges(\calF):=\left\{\range(\calF, C, r) \mid C \in [X]^k, r\geq 0 \right\}$ to be the collection of all the range sets. 
The range space of $\calF$ is defined as the pair $(\calF, \ranges(\calF))$.
%
%

Now, We define the dimension of a range space,
following~\cite{FL11}.
\begin{definition}[(shattering) dimension of a range space]
Suppose $\calF$ is an indexed function set with ground set $[X]^k$.
The (shattering) dimension of the range space $(\calF, \ranges(\calF))$, or simply the (shattering) dimension of $\calF$, denoted as $\Dim(\calF)$, is the smallest integer $t$, such that for any $\calD\subseteq \calF$ with $|\calD|\geq 2$, $|\ranges(\calD)| \leq |\calD|^t$. We note that in $\ranges(\calD)$, the same ground set $[X]^k$ is implicit.
\end{definition}

However, as discussed in Section~\ref{section:intro}, our guarantee of the dimension for the weighted doubling distance functions only holds in a probabilistic sense.
We capture this formally in the following.

\begin{definition}[probabilistic (shattering) dimension of a range space]
	\label{def:pdim}
	Suppose $\calF$ is a random indexed function set with a deterministic index set denoted as $V$.
	The $\tau$-error probabilistic (shattering) dimension of $(\calF, \ranges(\calF))$, or simply the $\tau$-error probabilistic dimension of $\calF$, denoted as $\PDim_{\tau}(\calF)$,
	is the smallest integer $t$ such that for any fixed $H \subseteq V$ with $|H|\geq 2$, $|\ranges(\calF_H)| \leq |H|^t$ with probability at least $1 - \tau$.
\end{definition}

\ignore{

\begin{definition}[VC-dimension]
	\label{definition:vc_dim}
	Suppose $\calF$ is a set of functions $[X]^k \rightarrow \mathbb{R}_{\geq 0}$. The VC-dimension of the range space $(\calF, \ranges(\calF))$ is the largest integer $t$, such that there exists a subset $\calD\subseteq \calF$ of size $t$ satisfying $|\ranges(\calD)|=2^t$.
\end{definition}

\begin{fact} (see e.g., \cite[Lemma 5.14]{har2011geometric})
	\label{lm:balltovc}
	Suppose $\calF$ is a set of functions $[X]^k \rightarrow \mathbb{R}_+$.
	If $\Dim(\calF)$ is $t$, then the VC-dimension of $\calF$ is bounded by $O(t\log t)$.
\end{fact}

}

We need a well studied notion in the PAC learning theory, called $\alpha$-approximation.

\begin{definition} [$\alpha$-approximation of a range space]
	\label{def:epsilonapp}
	Given a range space $(\calF,\ranges(\calF))$ (with ground set $[X]^k$), a set $\calS\subseteq \calF$ is an $\alpha$-approximation of the range space, if for every $\range(\calF, C, r) \in \mathsf{ranges}(\calF)$ ($C \in [X]^k, r\geq 0$)
	\begin{align*}
	\left|\frac{|\range(\calF, C, r)|}{|\calF|}-\frac{|\calS\cap \range(\calF, C, r)|}{|\calS|}\right|
	\leq \alpha.
	\end{align*}
\end{definition}
In particular, it was shown that a small sized (depending on $\alpha$ and the VC dimension\footnote{Our definition of the dimension is the shattering dimension of a range space, which tightly relates to the VC-dimension (see for example~\cite{kearns1994introduction}). In particular,
	if $\Dim(\calF)$ is $t$, then the VC-dimension of $\calF$ is bounded by $O(t\log t)$.}) independent sample from the function set is an $\alpha$-approximation with constant probability (see for example~\cite{DBLP:journals/jcss/LiLS01}).
However, the traditional results are for range spaces with bounded VC dimension only, and our probabilistic dimension is very different in nature.
We prove the following version of the sampling bound that only requires a bounded probabilistic dimension.
The proof can be found in Section~\ref{sec:approximation}.

\begin{lemma}
	\label{lemma:restate_weak_app}
	Suppose $\calF$ is a random indexed function set with fixed index set $V$. In addition, suppose $T : \mathbb{N} \times \mathbb{R}_{\geq 0}$ satisfies for any $H \subseteq V$ and $0<\gamma<1$,
	\begin{align*}
		\Pr[|\ranges(\calF_H)| \leq T(|H|, \gamma) ] \geq 1 - \gamma.
	\end{align*}
	%
	%
	Let $\calS$ be a collection of $m$ uniformly independent samples from $\mathcal{\calF}$. 
	Then with probability at least $1-\tau$, $\calS$ is an $\alpha$-approximation of the range space $(\calF, \ranges(\calF))$, where the randomness is taken over $\calS, \calF$ and
	\begin{align*}
		\alpha := \sqrt{\frac{48\left(\log(T(2m, \frac{\tau}{4})) + \log{\frac{8}{\tau}}\right)}{m}}.
	\end{align*}
\end{lemma}
In this bound, we directly use the size $|\ranges(\calF_H)|$ (rather than $\PDim$), 
which can provide a slightly more precise bound.
\footnote{
In Corollary~\ref{corollary:weighted_z_ball}, we can actually show $|\ranges(\calF_H)| \leq \eps^{-O(\DDim(M))} \cdot \log{\frac{1}{\tau}}\cdot \poly(|H|)$ with probability $1-\tau$, for a set $\calF$ of weighted doubling distance functions.
Of course we can also say $\PDim_\tau(\calF) \leq O\left(\DDim(M)\cdot \log(1/\eps) + \log\log1/\tau \right)$, but this would lead to a slightly looser bound.
}

\begin{remark}
	\label{remark:indexed_func}
	There are also technical reasons for considering indexed function sets.
	As discussed in Section~\ref{section:intro}, regarding the $\alpha$-approximation, we crucially use the fact that the function set is indexed (in particular the index set is fixed), 
	and we do not manage to prove the $\alpha$-approximation lemma (Lemma~\ref{lemma:restate_weak_app}) for more general function sets.
\end{remark}

\eat{
\begin{lemma}
	\label{lemma:restate_weak_app}
	Let $[X]^k$ be a ground set, and let $\mathcal{F}$ be a family of functions $[X]^k \rightarrow \mathbb{R}_{\geq 0}$.
	Suppose for any $\calD\subseteq \mathcal{F}$, there exists a function $T:\Z\rightarrow \Z$ such that
	\begin{align*}
	|\ranges(\calD)| \leq T(|\calD|).
	\end{align*}
	Let $\calS$ be a collection of $m$ uniformly independent samples from $\mathcal{F}$. Then with probability at least $1-\tau$, $\calS$ is an $\alpha$-approximation of the range space $(\mathcal{F}, \ranges(\calF))$ with
	\begin{align*}
	\alpha := \sqrt{\frac{32\left(\log{T(2m)} + \log{\frac{2}{\tau}}\right)}{m}}.
	\end{align*}
\end{lemma}

\begin{proof}
	By the proof idea in~\cite{HubertLec}, we have
	$ T(2m) \cdot 2e^{-\frac{\alpha^2 m}{32}} \leq \tau$.
	By calculation, it completes the proof.
\end{proof}
}

\section{Warmup: Unweighted Doubling Metrics}
\label{section:dim_range_doubling}
Let $M(X, d)$ be a doubling metric.
Consider the function set $\calF := \left\{ f_x(\cdot) \mid x \in X \right\}$ indexed by $X$ with $f_x(y) := d(x, y)$ for $y\in X$. 
It is well known that a bounded dimensional Euclidean space is a special case of doubling metrics,
and $\Dim(\calF) \leq O(t)$ if $M$ is the $t$-dimensional Euclidean space.
However, for a general doubling metric $M$, $\Dim(\calF)$ may not be bounded, as stated in the following theorem.



\begin{theorem}
	\label{theorem:doubling_high_ball_dim}
	For any integer $n\geq 1$, there is a metric space $M_n(X_n, d_n)$ with $2^n + n$ unweighted points such that $\DDim(M_n) \leq 2$ and $\Dim(\calF^{M_n}) \geq n/\log n$, where $\calF^{M_n} := \{ d_n(x, \cdot) \mid x\in X_n \}$.
\end{theorem}

\begin{proof}
	We start with the definition of $M_n(X_n, d_n)$.
	
	Define $L_n := \left\{ u_1, u_2, \ldots, u_n \right\}$,
	$R_n := \left\{ v_0, v_2, \ldots, v_{2^n-1} \right\}$.
	Define the point set of $M_n$ to be $X_n:=L_n\cup R_n$.
	For $1\leq i \leq j \leq n$, define $d_n(u_i, u_j) := |j-i|$.
	For $0\leq i \leq j \leq 2^n-1$, define $d_n(v_i, v_j) := |j-i|$.
	For $u_i \in L_n$ and $v_j \in R_n$,
	define $d_n(v_j, u_i) := 2^{n+1} + 1$ if the $i$-th digit in the binary representation of $j$ is $1$, and
	$d_n(v_j, u_i) := 2^{n+1}$ if the $i$-th digit in the binary representation of $j$ is $0$.
	This completes the definition of $M_n$.
	It is immediate that $M_n$ is a metric space.
	
	\noindent\textbf{Doubling Dimension.}
	Consider a ball with center $x\in X_n$ and radius $r$. We distinguish the following two cases.
	\begin{enumerate}
		\item If $r < 2^{n + 1}$, then either $\Bdn(x, r) \subseteq L_n$ or $\Bdn(x, r) \subseteq R_n$.
		Since the distance between points in $L_n$ is induced by a 1-dimensional line, each ball $\Bdn(x, r)\subseteq L_n$ can be covered by at most $3$ balls of radius $\frac{r}{2}$. 
		This argument also holds for each ball $\Bdn(x, r)\subseteq R_n$.
		\item If $r \geq {2^{n+1}}$, $\Bdn(x, r)$ is a union of a subset of $L_n$ and a subset of $R_n$. Then there exists $u \in L_n \cap \Bdn(x, r)$ and $v \in R_n \cap \Bdn(x, r)$. 
		Note that $L_n$ is covered by $\Bdn(u, 2^n)$ and $R_n$ is covered by $\Bdn(v, 2^n)$.
		Hence, each ball $\Bdn(x, r)\subseteq R_n$ can be covered by at most 2 balls $\Bdn(u, \frac{r}{2})$ and $\Bdn(v, \frac{r}{2})$.
	\end{enumerate}
	Therefore, $\DDim(M_n) \leq 2$.
	
	\noindent\textbf{Dimension of the Range Space.}
	Let $\calD$ be the subset of functions $\left\{d_n(u_i,\cdot)\mid i\in [n]\right\}\subseteq \calF^{M_n}$. Consider balls $\Bdn(v_j, 2^{n+1})$ for $v_j \in R_n$.
	By definition, $|\left\{L_n\cap \Bdn(v_j, 2^{n+1}) \mid v_j \in R_n \right\} | = 2^n$.
	Note that $$\left\{ f_{u_i}\in \calD\mid u_i\in L_n\cap \Bdn(v_j, 2^{n+1})\right\}=\left\{f_{u_i}\in \calD \mid f_{u_i}(v_j)=d_n(u_i,v_j)\leq 2^{n+1} \right\} \in \ranges(\calD).$$
	Hence, we have $|\ranges(\calD)|\geq 2^n\geq |\calD|^{n/\log n}$.
	Therefore, $\Dim(\calF^{M_n})$ is at least $n/\log n$.
\end{proof}

In light of Theorem~\ref{theorem:doubling_high_ball_dim}, it is impossible to bound the dimension of $\calF$ for doubling metric $M$.
However, we observe that,
from the hard instance in the proof of Theorem~\ref{theorem:doubling_high_ball_dim}, 
if we allow a small distortion to the distance functions
(i.e., to modify all distances $2^{n+1}+1$ to $2^{n+1}$), 
the dimension of the range space becomes bounded.
Inspired by this observation, we introduce the notion of \emph{smoothed distance functions} for doubling metrics
in the next subsection. 
Then we prove that the range space induced by the smoothed distance functions indeed has bounded dimension in a doubling metric (see Theorem~\ref{theorem:loose_bound}), which is the main result of this section.
 
%

\subsection{Smoothed Distance Functions}
\label{section:smoothed_dis}
The smoothed distance function is defined with respect to a metric space 
$M(X,d)$
and a net tree $T$ of the space
(definition in Section~\ref{subsec:doubling}). The proofs in this section are postponed to Section~\ref{section:smooth_proof}.

\begin{definition}[$\eps$-smoothed distance function]
	\label{definition:eps_smooth}
	Given a net tree $T$ of a metric space $M(X,d)$,
	for $0<\eps<1$,
	define $\delta_{\eps}: X\times X\rightarrow \R_{\geq 0}$ as the $\eps$-smoothed distance function induced by $T$
	as follows.
	For any $x,y\in X$, let $h_\eps(x, y)$ be the largest integer $j$ such that $d(\Par^{(j)}(x), \Par^{(j)}(y)) \geq \frac{2^j}{\eps}$.
	\footnote{Such $j$ must exist, because $j = -\infty$ always satisfies the condition.}
	Define $j = h_\eps(x, y)$ and
	$\delta_\eps(x, y) := d(\Par^{(j)}(x), \Par^{(j)}(y))$.
\end{definition}
We assume that there is an underlying net tree $T$,
and we drop the subscript in $\delta$ and $h$ whenever the context is clear. 
Note that $\delta$ may \emph{not} be a distance function since it may not satisfy the triangle inequality.
But it satisfies the non-negativity and symmetry properties. 
Nonetheless, it is a close approximation of the original distance function $d(\cdot,\cdot)$,
as in the following lemma. 

\begin{lemma}[small distortion]
	\label{lemma:distortion}
	If $T$ is $c$-covering, then for any $x,y \in X$ and any $\eps>0$,
	\begin{align*}
	(1-4c\cdot \eps)\cdot \delta(x,y) \leq d(x,y) \leq (1+4c\cdot \eps)\cdot \delta (x,y).
	\end{align*}
\end{lemma}
\ignore{
\begin{proof}
	Let $j = h(x,y)$. By definition,
	\begin{equation*}
	\delta(x,y) = d(\Par^{(j)}(x), \Par^{(j)}(y))\geq \frac{2^j}{\eps}.
	\end{equation*}
	By Definition~\ref{fact:des_dis}, $d(x, \Par^{(j)}(x)) \leq c\cdot 2^{j+1}$,
	and $d(y, \Par^{(j)}(y)) \leq c \cdot 2^{j+1}$.
	By the triangle inequality,
	\begin{align*}
	d(x,y)
	&\leq d(x, \Par^{(j)}(x)) + d(\Par^{(j)}(x), \Par^{(j)}(y)) + d(\Par^{(j)}(y), y) \\ &\leq (1 + 4c\cdot \eps)\cdot \delta(x, y).
	\end{align*}
	Also,
	\begin{align*}
	d(x, y)
	&\geq d(\Par^{(j)}(x), \Par^{(j)}(y)) - d(x, \Par^{(j)}(x)) - d(y, \Par^{(j)}(y)) \\
	&\geq (1 - 4c \cdot \eps) \cdot \delta(x,y).
	\end{align*}
\end{proof}
}
Next, we show that the $\eps$-smoothed distance function has several useful properties, which 
we will use extensively.
The first is the descendant property, which says that if the smoothed distance of $x$ and $y$
is defined by two nodes $u$ (an ancester of $x$) and $v$ (an ancester of $y$) in layer $j$,
then any descendant of $u^{(j)}$ has the same smoothed distance to any descendant of $v^{(j)}$.

\begin{lemma}[descendant property]
	\label{lemma:delta_equ}
	For any $x, y \in X$, assume that $j = h(x, y)$, $u = \Par^{(j)}(x)$ and $v = \Par^{(j)}(y)$. Then for any $x'\in \Des(u^{(j)})$ and $y'\in \Des(v^{(j)})$, we have $ \delta(x, y) = \delta(x', y') = d(u, v) $.
\end{lemma}
\ignore{
\begin{proof}
	By definition, we immediately have $\delta(x, y) = d(u, v)$.
	Observe that $\Par^{(j)}(x') = \Par^{(j)}(x)$, and $\Par^{(j)}(y) = \Par^{(j)}(y')$.
	By Definition \ref{definition:eps_smooth}, we know $d(\Par^{(j)}(x'),\Par^{(j)}(y'))=d(u,v) \geq \frac{2^j}{\eps}$. Thus, we have $h(x', y') \geq j$.
	On the other hand, for $j' > j = h(x, y)$, we have $d(\Par^{(j')}(x'),\Par^{(j')}(y'))=d(\Par^{(j')}(x),\Par^{(j')}(y)) < \frac{2^{j'}}{\eps}$, where the last inequality is by the definition of $h(x,y)$. Thus, we have $h(x, y) \leq j$. Therefore, $h(x, y) = h(x', y')$ which implies that $\delta(x', y') = d(u, v)$.
\end{proof}
}
The second is the smooth property which says that
at a certain distance scale $r$, if we move the center of the ball (of radius $r$) from $x$ to $x'$
($x'$ is a nearby point in a small subtree),  the ball does not change, when the ball is defined w.r.t. $\delta$.


\begin{lemma}[smooth property, illustrated in Figure~\ref{figs:smooth}]
	\label{lemma:ball_equal}
	Suppose $\{N_i \mid i \leq L\}$ is a hierarchical net and $T$ is a $c$-covering net tree with respect to $\{N_i\}_i$.
	Consider $0<\eps\leq \frac{1}{8c}$ and $r>0$. 
	Let $\lambda := \frac{\eps\cdot (1-5c\eps)}{20 (1+4c\eps)}$.
	Define $j$ to be the integer satisfying that $2^{j-1} \leq \lambda \cdot r $.
	Then for any $x,x'\in X$, if $\Par^{(j)}(x) = \Par^{(j)}(x')$, we have $\Bdt(x, r) = \Bdt(x', r)$.
\end{lemma}

\begin{remark}
	\label{remark:hang}
	An interesting consequence of Lemma~\ref{lemma:ball_equal} is that, for any ball $\Bdt(x, r)$, $\Bdt(x, r) = \Bdt(\Par^{(j)}(x), r)$ ($j$ defined in Lemma~\ref{lemma:ball_equal}). This means we can ``hang'' the center $x$ to $\Par^{(j)}(x)$ which is a net point of higher height. 
\end{remark}


The third is the cross-free property. Consider a ball $\Bdt(x, r)$. The next lemma says that any small subtree (with distance scale less than $\epsilon r$)
is either completely contained in the ball, or does not intersect the ball at all.
A consequence useful later is that each ball can be viewed as the union of some small subtrees.

\begin{lemma}[cross-free property, illustrated in Figure~\ref{figs:cross}]
	\label{lemma:ball_laminar}
Suppose $\{N_i \mid i \leq L\}$ is a hierarchical net and $T$ is a $c$-covering net tree with respect to $\{N_i\}_i$.
Consider $0<\eps\leq \frac{1}{8c}$ and $r>0$. 
Let $\lambda := \frac{\eps\cdot (1-5c\eps)}{20 (1+4c\eps)}$.
Suppose $j$ is an integer such that $2^{j-1} \leq \lambda \cdot r$.
Then for any $x \in X$ and $v \in N_{j}$, either $\Des(v^{(j)}) \subseteq \Bdt(x, r)$ or $\Des(v^{(j)}) \cap \Bdt(x, r) = \emptyset$.
\end{lemma}

\begin{figure}[t]
	\centering
	\begin{subfigure}[t]{0.4\textwidth}
		\centering
		\includegraphics[width=\textwidth]{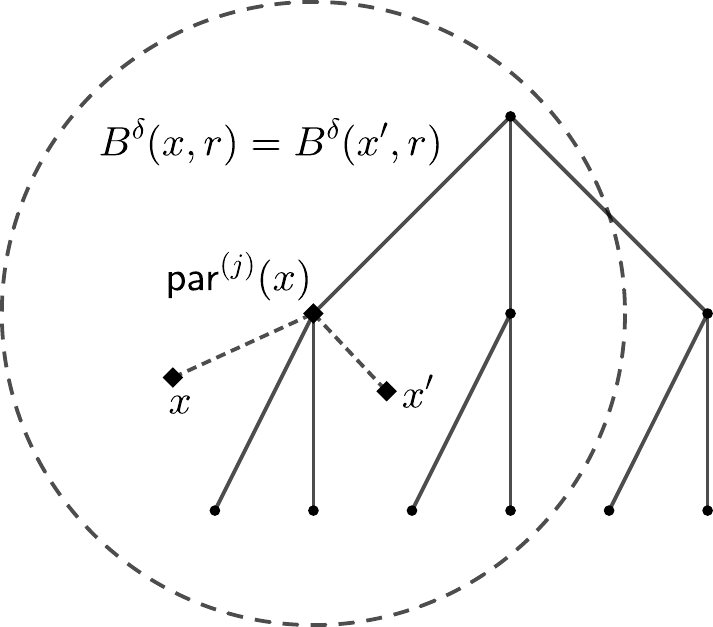}
		\caption{smooth property}
		\label{figs:smooth}
	\end{subfigure}%
	~
	\begin{subfigure}[t]{0.5\textwidth}
		\centering
		\includegraphics[width=\textwidth]{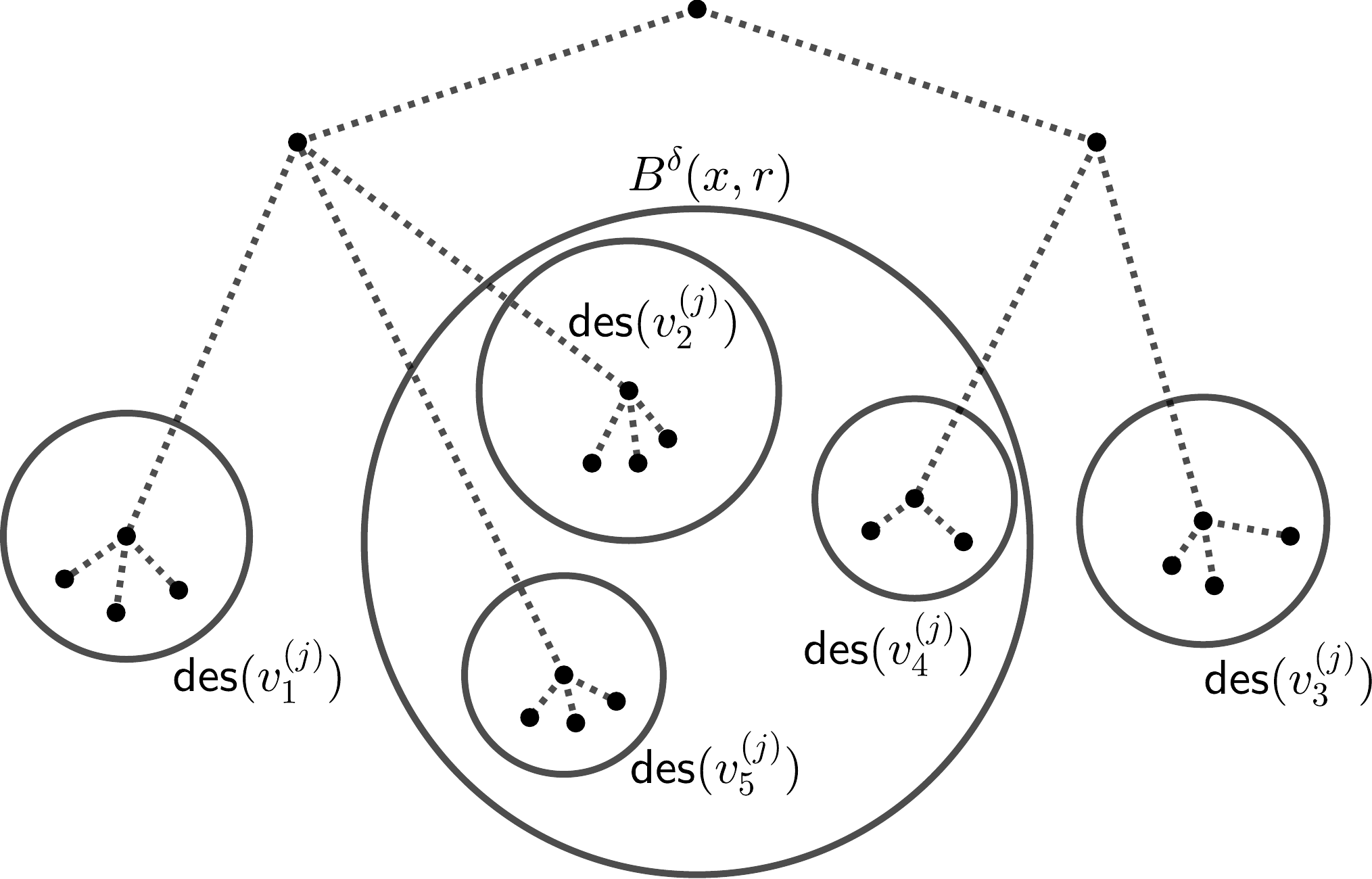}
		\caption{cross-free property}
		\label{figs:cross}
	\end{subfigure}
	\caption{Illustration for the smooth property and the cross-free property}
	\label{figs:smooth_cross}
\end{figure}
\ignore{
\begin{proof}
	We use Lemma~\ref{lemma:ball_equal} to prove this lemma.
	Observe that it suffices to show for the case that $2^{j-1}\leq \lambda\cdot r<2^j$, because there always exists $w \in N_j$ such that $\Des(v^{(j')}) \subseteq \Des(w^{(j)})$ for any $j'\leq j$.\jian{grammar}\shaofeng{fixed.}

	Suppose for some $v^{(j)}\in N_j$, $\Des(v^{(j)}) \cap \Bdt(x, r) \neq \emptyset$. 
	Let $x' \in \Des(v^{(j)}) \cap \Bdt(x, r)$.
	Then we have $\delta(x, x')\leq r$.
	This implies that $x \in \Bdt(x', r)$.
	It suffices to show that for any $y \in \Des(v^{(j)})$, $\delta(y, x) \leq r$.
	By Lemma~\ref{lemma:ball_equal}\jian{fix}\shaofeng{fixed}, for any $y \in \Des(v^{(j)})$, $\Bdt(y, r) = \Bdt(x', r)$. Hence, $x\in \Bdt(y, r)$ by the fact that $x\in \Bdt(x', r)$. This implies $\delta(y, x) \leq r$, which means $y \in \Bdt(x, r)$.
	It completes the proof.
\end{proof}
}

\subsection{Bounded Dimension for Smoothed Doubling Distance Functions}
\label{section:unweighted}
In this section, we showcase the use of the smoothed distance function.
In particular, we show in Theorem~\ref{theorem:loose_bound} that the range space induced by
smoothed doubling distance functions has bounded dimension.

The $\eps$-smoothed distance function in this section is defined with respect to the \emph{simple net tree}, which is the following natural net tree built on a hierarchical net.

\begin{definition}[simple net trees]
In a simple net tree, for each $u\in N_i$,
$\Par(u^{(i)})$ is defined to be the nearest point $v \in N_{i+1}$ to $u$ (ties are broken arbitrarily).
\end{definition}

The following fact follows immediately from the definition of simple net trees.

\begin{fact}
	\label{fact:simple}
	A simple net tree is $1$-covering.
\end{fact}


Now, everything is ready to prove the main theorem of this section.

\begin{theorem}
	\label{theorem:loose_bound}
	Suppose $M(X,d)$ is a metric space and $T$ is a simple net tree on $X$. Let $0< \e\leq \frac{1}{8}$ be a constant. Let $\delta$ be the $\eps$-smoothed distance function induced by $T$. Let $\calF := \left\{ \delta(x, \cdot)\mid x\in X \right\}$ be the function set induced by the $\eps$-smoothed distance functions. Then $\Dim(\calF) \leq O(\frac{1}{\eps})^{O(\DDim(M))}$.
\end{theorem}

\begin{proof}
	Consider any subset $H\subseteq X$ of size $|H|=m\geq 2$.
	It suffices to show
	\begin{align*}
		\left|\left\{H\cap \Bdt(x,r)\mid x\in X,r\geq 0 \right\} \right| \leq m^{O(\eps)^{-O(\DDim(M))}}.
	\end{align*}
	Let $\lambda := \frac{\eps\cdot (1-5c\eps)}{20 (1+4c\eps)}$ as defined in Lemma~\ref{lemma:ball_laminar}.
	Let us first fix some $r \geq 0$ and $x\in X$.
	Define $j$ to be the integer such that $2^{j-1} \leq \lambda \cdot r < 2^j$.
	By Lemma~\ref{lemma:ball_laminar}, $\Bdt(x, r)$ is the union of
	some $\Des(v)$'s for $v$ in a subset of $N_j$. Next, we show the number of such $\Des(v)$'s is a constant (depending on $\DDim(M)$ and $\eps$).
	
	Let $P$ be the set of $v\in N_j$ satisfying that $\Des(v) \subseteq \Bdt(x, r)$, i.e., $\Bdt(x, r) = \bigcup_{v\in P} \Des(v)$. Since $P\subseteq N_j$ is a $2^j$-packing, the distance between any two points in $P$ is at least $2^j$.
	On the other hand, since $P\subseteq B^\delta(x,r)$, we have $\diam(P)\leq 2(1+4\eps)\cdot r< 2^{j+2}/\lambda$.
	Then by packing property (Fact~\ref{fact:packing}), $|P| \leq O(\frac{1}{\lambda})^{\DDim(M)}$. 
	Define $H^{(j)}:=\left\{\Par^{(j)} (x)\mid x\in H \right\}$.
	We have $H\cap \Bdt(x,r) = \bigcup_{v\in P} (H\cap \Des(v)) = \bigcup_{v\in P\cap H^{(j)}} (H\cap \Des(v))$. This implies every ball $\Bdt(x, r)$ is formed by first choosing at most $\Lambda:=O(\frac{1}{\lambda})^{\DDim(M)}$ points $v\in H^{(j)}$, and then letting $\Bdt(x, r)$ be the union of these $\Des(v)$'s.

	Now we turn to general $x$ and $r$. For $r \geq 0$, define
	\begin{align*}
		Q_r:= \left\{ \bigcup_{x\in S} \left(H\cap \Des(\Par^{(j)}(x))\right) \mid S\subseteq H, |S|\leq \Lambda \right\},
	\end{align*}
	where $2^{j-1} \leq \lambda \cdot r < 2^j$.
	By the above argument, we know that $H\cap \Bdt(x,r)\in Q_r$ for any $x\in X$.
	Hence, $\left\{H\cap \Bdt(x,r)\mid x\in X,r\geq 0 \right\} \subseteq \bigcup_{r\geq 0} Q_r$.
	Then to bound $\left|\left\{H\cap \Bdt(x,r)\mid x\in X,r\geq 0 \right\} \right|$, it suffices to bound $|\bigcup_{r\geq 0} Q_r|$.
	Note that for any fixed $r\geq 0$, $|Q_r|\leq O(m^\Lambda)$. 
	We claim that there are at most $m+1$ different collections $Q_r$ for all $r\geq 0$.
	If the claim is true, we can bound $|\bigcup_{r\geq 0} Q_r|$ by  $O((m+1)\cdot m^\Lambda) = O(m^{\Lambda+2})$, and this would conclude the theorem.
	
	It remains to prove the claim that there are at most $m+1$ different collections $Q_r$ for all $r\geq 0$.
	Observe that the cardinality of $H^{(j)}$ is non-increasing as $j$ increases. Assume that $|H^{(i)}|=|H^{(j)}|$ for some $i\leq j$. 
	Then for any $x,y\in H$, we have $\Par^{(i)}(x)=\Par^{(i)}(y)$ if and only if $\Par^{(j)}(x)=\Par^{(j)}(y)$.

	Now fix some $x\in X$. Let $u:= \Par^{(i)} (x) \in N_j$ and $v:=\Par^{(j)} (x) \in N_j$. 
	If $|H^{(i)}|=|H^{(j)}|$, we have the following for any $y\in H$, 
	$$y\in \Des(\Par^{(i)} (x))\,\,\,\Leftrightarrow\,\,\,
	\Par^{(i)}(y) = u \,\,\,\Leftrightarrow\,\,\, 
	\Par^{(j)}(y) = v \,\,\,\Leftrightarrow\,\,\,
	y\in \Des(\Par^{(j)} (x)),$$
	which implies that $H\cap \Des(\Par^{(i)} (x)) = H\cap \Des(\Par^{(j)} (x))$.
	Hence, for any $r', r$, define $i$ to be the integer such that $2^{i-1} \leq \lambda \cdot r' < 2^i$ and $j$ to be the integer such that $2^{j-1} \leq \lambda \cdot r < 2^j$. 
	If $|H^{(i)}|=|H^{(j)}|$, then $Q_r=Q_{r'}$. 
	Since there are at most $m+1$ possible cardinalities for $|H^{(j)}|$, there are at most $m+1$ different $Q_r$'s. 
	This proves the claim and thus concludes the theorem.
\end{proof}
\ignore{
Combining with Fact \ref{lm:balltovc}, we directly obtain the following corollary.

\begin{corollary}
	\label{corollary:doublingtoVC}
	Suppose $M(X,d)$ is a metric space and $T$ is a simple net tree. Let $0< \e\leq \frac{1}{8}$ be a constant. Let $\delta_\eps$ be the $\eps$-smoothed distance function induced by $T$. Then the VC dimension of the ranges space of $\calF_M:= \{\delta(x, \cdot)\}_{x \in X}$ is at most $O(\frac{\DDim(M)}{\e^{\DDim(M)}}\log \frac{1}{\e})$.
\end{corollary}

We note that the $\e$-smoothed process is necessary for achieving the bounded VC dimension, in view of Theorem~\ref{theorem:doubling_high_ball_dim}. 
}

\section{Weighted Doubling Metrics}
\label{section:weighted}
In the last section, we provide a bound of the shattering dimension
by the doubling dimension.
Note that the dimension bound in Theorem~\ref{theorem:loose_bound} is quite large 
in that it is exponential in $\DDim(M)$.
However, considering the Euclidean case, the dependency is only linear.
Moreover,  Theorem~\ref{theorem:loose_bound} is not sufficient for the purpose
of constructing coresets, for which we need a dimension bound for weighted spaces.
In this section, we provide a new proof that can reduce the exponential dependency to a polynomial dependency, in a certain probabilistic sense.
Moreover, the proof also works for weighted doubling metrics, where each point $x\in X$ is associated with a weight $w(x)$.
In particular, we consider the following type of weight functions, which suffices
for coreset construction.

\begin{definition}
	\label{definition:weight}
	We say $w:X\rightarrow \R_{\geq 0}$ is a gap-$c$ weight function if for any $x,y\in X$, we have either $w(x)=w(y)$ or $\max\left\{\frac{w(x)}{w(y)},\frac{w(y)}{w(x)}\right\}\geq c$.
\end{definition}

To achieve the polynomial dependence in $\DDim(M)$, we shall construct a \emph{random} $\eps$-smoothed distance function $\delta$. Let $\calF := 
\left\{ w(x)\cdot \delta(x, \cdot) \mid x\in X \right\}$ be the function set induced by the random $\eps$-smoothed distance function $\delta$.
We will show that $\PDim_{\tau}(\calF) \leq O\left( \DDim(M)  \cdot \log (1/\eps)
+ \log\log 1/\tau \right)$
%
in Theorem~\ref{theorem:weighted}\footnote{Actually, Theorem~\ref{theorem:weighted} provides a better bound than the one with respect to $\PDim$.}, which is the main theorem of this section.

\begin{theorem}
	\label{theorem:weighted}
	Suppose $M(X, d)$ is a metric space together with a gap-2 weight function $w: X\rightarrow \R_{\geq 0}$.
	Let $0<\eps\leq \frac{1}{100}$ and $0 < \tau < 1$ be constant. 
	There exists a random $\eps$-smoothed distance function $\delta$ (defined with respect to some random net tree), such that for $\calF := \left\{ w(x)\cdot \delta(x, \cdot) \mid x\in X \right\}$,
	and any $H\subseteq X$,
	\begin{align*}
		\Pr_{\delta}\left[|\ranges(\calF_H)| \leq O\left(\frac{1}{\eps}\right)^{O(\DDim(M))} \cdot \log{\frac{|H|}{\tau}} \cdot  |H|^{6} \right] \geq 1 - \tau.
	\end{align*}
	In other words, $\PDim_\tau(\calF)= {O}(\DDim(M)\cdot \log (1/\eps) + \log\log{1/\tau}) $.
\end{theorem}



%
The rest of this section is devoted to the proof of Theorem~\ref{theorem:weighted}.

\subsection{Preparations}
We need some notations first.
Suppose we already have a random $\eps$-smoothed distance function $\delta$.
By Definition~\ref{definition:weight}, we can assume there are $l$ distinct weights satisfying that 
$$
w_1\geq 2 w_2 \geq \ldots \geq 2^{l-1} w_l.
$$
Consider a fixed set $H\subseteq X$ of size $m$.
%
%
We divide $H$ into $l$ groups $\{H_i\}_{i\in [l]}$ according to the weight, where $H_i:=\left\{x\in H\mid w(x)=w_i \right\}$ for $i\in [l]$.
%
For any $x\in X$, $i\in [l]$ and $r\geq 0$, define $\Bdt_i(x,r): = H_i\cap \Bdt(x,r)$ to be the intersection of $H_i$ and $\Bdt(x,r)$. Similarly, define $\Bd_i(x, r) := H_i \cap \Bd(x, r)$.
%
%
%
Recalling that $\range(\calF_H,x,r)=\left\{ f_y \in \calF_H \mid f_y(x) \leq r \right\}$,
we
define 
%
\begin{align*}
\calH(a, b): = \Bigl\{ \left\{ y \in H\mid  f_y \in \range(\calF_H,x,r) \right\} \mid x\in X, r\in [a, b) \Bigr\} = \left\{ \bigcup_{i\in [l]} \Bdt_i(x,\frac{r}{w_i}) \mid x\in X, r\in [a, b) \right\}.
\end{align*}
By definition, $|\ranges(\calF_H)| = |\calH(0, \infty)|$. 
Hence, it is sufficient to analyze $|\calH(0, \infty)|$.

In order to bound $|\calH(0, \infty)|$, we break $[0, \infty)$
into a small number of intervals, called {\em critical intervals}.
Then we will bound $\calH(a,b)$ for each critical interval $[a,b)$.
Apart from critical intervals, we also pick some \emph{representatives} for $H$ at each distance $r$, and we call them $r$-representatives.
The $r$-representatives are defined with respect to a hierarchical net on $(H, d)$.
%
%
%

%
%
%
%
%

\begin{definition}[$r$-representatives]
	\label{definition:basic_center}
	Let $\left\{N_i^{(H)} \mid i\leq L\right\}$ be a hierarchical net with respect to $(H, d)$. 
	For $i\in [l]$ and $r\geq 0$, assume $\zeta_i$ is the integer satisfying that $2^{\zeta_i} \leq \frac{r}{w_i} < 2^{\zeta_i + 1}$, we define $R_i(r) := N_{\zeta_i}^{(H)}$ to be the collection of net points in $N_{\zeta_i}^{(H)}$. 
	For any $r\geq 0$, define the $r$-representatives $R(r): = \bigcup_{i\in [l]} R_i(r)$ to be the set formed by the union of $R_i(r)$.
\end{definition}


We emphasize again that the hierarchical net is defined over $H$,
which is different from that defined over the whole point set $X$.
Intuitively, for $r\geq 0$, 
$R_i(r)$ is the set of net points in the height with distance scale
approximately $\frac{r}{w_i}$.
As another clarification, unless otherwise specified, we refer to a hierarchical net $\{N_i\}_i$ defined on the \emph{whole} metric, when we talk about notions relevant to nets (such as ``net point'', the ``hierarchical net''). 
Under rare circumstances when we refer to $\{N_i^{(H)}\}_i$, we will make it explicit.
In the following lemma, we show some useful facts about the $r$-representatives.

\begin{lemma}
	\label{lemma:basic_center_prop}
	$r$-representatives have the following properties.
	\begin{enumerate}
		\item For $r\geq 0$, $|R(r)| \leq m$.
		\item For $r\geq r' \geq 0$ and $i\in [l]$, $R_i(r) \subseteq R_i(r')$.
		\item For $i\in [l]$, $r\geq 0$, and any $x\in H_i$, 
		there exists $y\in R(r)$ such that $x\in B^d_i(y, \frac{r}{w_i})$.
		\item For $r\geq 0$ and $i<j\leq l$, $R_j(r)\subseteq R_i(r)$.
		\item For $r\geq 0$, $R(r)=R_1(r)$.
	\end{enumerate}
\end{lemma}

\begin{proof}
	Item 1 follows from the fact that $R(r) \subseteq H$.
	Item 2 follows from the definition of the hierarchical net $\{N_i^{(H)}\}_i$.
	For item 3, let $\zeta_i$ be the integer such that $2^{\zeta_i}\leq \frac{r}{w_i} <2^{\zeta_i + 1} $.
	Item 3 follows from the fact that $R_i(r)$ is a $2^{\zeta_i}$-covering of $H$ which means there must exist $y\in R_i(r)\subseteq R(r)$ such that $d(x,y)\leq 2^{\zeta_i}\leq \frac{r}{w_i}$.
	%
	Thus, we have $x\in \Bd_i(y,\frac{r}{w_i})$ which proves item 3.
	For item 4, let $\zeta_j$ be the integer such that $2^{\zeta_j}\leq \frac{r}{w_j} <2^{\zeta_j+1} $.
	Since $i<j$, we have $w_i>w_j$ which implies that $\zeta_i\leq \zeta_j$.
	Then item 4 follows from the fact that $N_{\zeta_j}^{(H)}\subseteq N_{\zeta_i}^{(H)}$.
	Item 5 is a direct corollary of item 4.
\end{proof}

Now, we are ready to define critical intervals.

\begin{definition}[critical interval]
	\label{definition:critical_interval}
	We say $I := [a, b)$ is a critical interval, if it is a maximal interval such that
	\begin{enumerate}
		\item for all $r_1, r_2\in I$ and $i\in [l]$, $R_i(r_1) = R_i(r_2)$.
		\item for all $r_1, r_2\in I$, $i \in [l]$ and $x\in R(a)$,
		$\Bd_i(x, \frac{r_1}{w_i}) = \Bd_i(x, \frac{r_2}{w_i})$.
	\end{enumerate}
\end{definition}

We will show in Lemma~\ref{lemma:num_critical_interval} that the number of critical intervals is only $\poly(m)$.
Furthermore, a critical interval $[a,b)$ has some invariance properties, 
which are useful in bounding $|\calH(a,b)|$.
We will provide an overview of how those invariance can help in Section~\ref{section:weighted_overview}.

\subsubsection{Random $\eps$-Smoothed Distance Functions}
Now, we explain how to construct $\delta$. 
Recall that $\delta$ is deterministically defined with respect to a net tree $T$. So the randomness of $\delta$ will come from a random construction of the net tree $T$.
The random net tree $T$ is constructed with respect to the randomized hierarchical decomposition
developed in~\cite{DBLP:conf/stoc/AbrahamBN06}.

\begin{definition}[Hierarchical Decomposition~\cite{DBLP:conf/stoc/AbrahamBN06}]
	\label{defn:single_decomp}
	Suppose the hierarchical net $\left\{N_i\right\}_i$ is given for a metric space $M$.
	At height $i$, an arbitrary ordering $\pi_i$ is
	imposed on the net $N_i$.  Each net-point $u \in N_i$
	corresponds to a \emph{cluster center} and
	samples a random value $h_u$ from a truncated exponential distribution
	$\Exp_i$ having density function $t \mapsto
	\frac{{\chi}}{{\chi}-1} \cdot \frac{\ln \chi}{2^i} \cdot e^{-\frac{t \ln \chi }{2^i}}$ for $t \in [0, 2^i]$, where $\chi = O(1)^{\DDim(M)}$.
	Then, the cluster at $u$ has random radius $r_u := 2^i + h_u$.
	
	The clusters
	induced by $N_i$ and the random radii form a
	decomposition $\Pi_i$,
	where a point $p \in X$ belongs to the cluster
	with center $u \in N_i$ such that $u$ is the first
	point in $\pi_i$ with $p\in \Bd(u,r_u)$.
\end{definition}

We say that the partition $\Pi_i$ {\em cuts} a subset $P\subseteq X$
if $P$ is not totally contained within a single cluster at height $i$.
One useful property about the above decomposition is that
the probability that $P$
is cut by $\Pi_i$ can be bounded as follows.
\begin{theorem}\cite{DBLP:conf/stoc/AbrahamBN06}
	\label{theorem:cutprobability}
Suppose $\Pi_i$ is a random decomposition defined in Definition~\ref{defn:single_decomp}.
Then, it holds that
$$
\Pr[P \text{ is cut by } \Pi_i]=O(2^{-i}\cdot\DDim(M) \cdot \diam(P)).
$$
\end{theorem}

\begin{definition}[Hierarchical Decomposition] \label{defn:phd}
	Let $\left\{\Pi_i\right\}_{i}$ be a random decomposition
	induced as in Definition~\ref{defn:single_decomp}.
	At the top height $L$, the whole space
	is partitioned by $\Pi_{L}$ to form the height-$L$ cluster.  
	Inductively,
	each cluster at height $i+1$ is partitioned
	by $\Pi_i$ to form height-$i$ clusters, until height $1$ is reached. Any cluster
	has at most $O(1)^{\DDim(M)}$ child clusters.
\end{definition}

We call a set $P$ is cut at height $i$ \emph{iff}
the set $P$ is cut by some partition $\Pi_j$ with
$j \geq i$.
From Theorem~\ref{theorem:cutprobability}, we can see the following corollary.

\begin{corollary}
	\label{corollary:cutprob}
 The probability that $P$ is cut at height $i$
 is at most 
 $$\sum_{j \geq i} O(2^{-j}\cdot\DDim(M) \cdot \diam(P)) 
 = O(2^{-i}\cdot \DDim(M) \cdot \diam(P)).$$
\end{corollary}

\noindent\textbf{Random Net Tree $T$.}
Suppose $\Pi := \left\{ \Pi_i \right\}_i$ is a random hierarchical decomposition in Definition~\ref{defn:phd}. 
%
%
Then for each $x\in X$ and $i \leq L$, we define $\Par^{(i)}(x)$ to be the unique point $v\in N_i$ such that the cluster of $v^{(i)}$ contains $x$.
%
%
The key properties we need from $T$ are stated in the following lemma. They follow immediately from Definition~\ref{defn:phd}.

\begin{lemma}
	\label{lemma:decomp_net_tree}
	The random net tree $T$ satisfies the following.
	\begin{enumerate}
		\item $T$ is $2$-covering.
		\item For any $P\subseteq X$ and $i\leq L$,
		$\Pr[\forall v \in N_i: P\not\subseteq \Des(v^{(i)})] \leq O\left(2^{-i}\cdot \DDim(M) \cdot \diam(P)\right)$.
	\end{enumerate}
\end{lemma}

\subsection{Proof Overview of Theorem~\ref{theorem:weighted}}
\label{section:weighted_overview}
Recall that the goal is to bound $|\calH(0, \infty)|$.
Let us first focus on the case of a bounded interval $[a, b)$ such that $b \leq 2a$, studied in~Lemma~\ref{lemma:bounded_interval}. We make use of the smooth property of $\delta$ stated in Lemma~\ref{lemma:ball_equal}, so that we can ``hang'' the centers of balls on some net point of higher height, as noted in Remark~\ref{remark:hang}.
Then, we use the structure of the net tree $T$, and relate $|\calH(a, b)|$ to  the number of some designated net points that we hang on to. We show that the collection of net points, which are ``close to'' some point in $H$, must form a set with a bounded diameter.
Also, since the net points are of enough height, we can use the packing property in doubling metrics to bound the number of the relevant net points.

Observe that for an interval $[a, b)$ such that $\frac{b}{a}$ is bounded, Lemma~\ref{lemma:bounded_interval} provides a nice bound on $|\calH(a, b)|$.
However, since we consider $|\calH(0, \infty)|$, Lemma~\ref{lemma:bounded_interval} is not sufficient. 
To resolve this, we partition $[0, \infty)$ into at most $2m^4$ \emph{critical} intervals (Definition~\ref{definition:critical_interval}, Lemma~\ref{lemma:num_critical_interval}). Moreover, for each $r$, we take a set $R(r)$ as representatives of points (which are nets of proper height in the hierarchical net
$\left\{N_i^{(H)} \mid i\leq L\right\}$ with respect to $(H,d)$),\jian{we said when we speak about $N_i^{H}$, we need to be explicit. So we should do that! check all such places.} \lingxiao{It seems that the above is the only place we mention $N_i^{(H)}$. @Shaofeng, is it correct?}\shaofeng{I think so. I think being explicit means we use $N^{(H)}$ notation.}which we call $r$-representatives (recalling Definition~\ref{definition:basic_center}). %
Since the number of critical intervals is small, it suffices to bound $|\calH(a, b)|$ for a critical interval $[a, b)$.

The analysis of a critical interval $I := [a, b)$ is in Lemma~\ref{lemma:upper_bound_critical_interval}.
In the analysis, we further partition $[a, b)$ into three intervals, $[a, r_1)$, $[r_1, r_2)$ and $[r_2, b)$. Here, $r_1, r_2$ are picked not too far away from $a$ and $b$ respectively, and so we can apply Lemma~\ref{lemma:bounded_interval} to bound $|\calH(a, r_1)|$ and $|\calH(r_2, b)|$. 
The most difficult part is  to bound $|\calH(r_1, r_2)|$ (for carefully chosen $r_1, r_2$).
Note that applying Lemma~\ref{lemma:bounded_interval} does not provide a useful bound.
In the following, we give some intuitions why we can handle this situation.
The main idea is to show that any range $\bigcup_{i\in [l]} \Bdt_i(x, \frac{r}{w_i})$ ($r\in [r_1,r_2]$)
can only be of the form $\bigcup_{j\leq i\leq l} \Bd_i(u, \frac{a}{w_i})$ for some $u\in R(a)$ and $j\in [l+1]$.
%

The first observation is that by choosing $r_2=\frac{b}{2^9}$, any range $\bigcup_{i\in [l]} \Bdt_i(x, \frac{r}{w_i})$ must be a subset of some concentric ball $\bigcup_{i\in [L]} 
\Bd_i(u, \frac{a}{w_i})
$ centered at $u\in R(a)$.
%
%
%
%
We first show in Lemma~\ref{lemma:long_interval} that, for any $p,q\in R(a)$,
there must exist some $j$ such that 
$\Bd_i(p,\frac{a}{w_i})$ and $\Bd_i(q,\frac{a}{w_i})$ are disjoint for any $i\leq j-1$ and are equal for any $i\geq j$.
The intuition is if $\Bd_i(p,\frac{a}{w_i})$ and $\Bd_i(q,\frac{a}{w_i})$ are neither disjoint nor equal, the critical interval cannot be long.
%
%
Based on such structured intersections between concentric balls, we prove that if $a\leq r\leq \frac{b}{2^9}$, any range $\bigcup_{i\in [l]} \Bdt_i(x,\frac{r}{w_i})$ must be a subset of a concentric ball determined by some $u\in R(a)$, by Lemmas~\ref{lemma:large_r} and~\ref{lemma:subset}.
Thus, we need pick $r_2 = \frac{b}{2^9}$.

The remaining problem is that the range $\bigcup_{i\in [l]} \Bdt_i(x,\frac{r}{w_i})$
can still be an arbitrary subset of some concentric ball. 
To handle this problem, we show that $\bigcup_{i\in [l]} \Bdt_i(x,\frac{r}{w_i})$ can not be arbitrary if $r_1$ is large enough, due to the cross-free property of $\delta$ (Lemma~\ref{lemma:ball_laminar}) and the assumption of the gap-2 weight function.
To be more specific, we prove in Claim~\ref{claim:large_ball_equal} that any range $\bigcup_{i\in [l]} \Bdt_i(x, \frac{r}{w_i})$
can only be of the form
$\bigcup_{j\leq i\leq l} 
\Bd_i(u, \frac{a}{w_i})
$
for some $u\in R(a)$ and $j\in [l+1]$,
conditioning on the event that $\Bd_i(u, \frac{a}{w_i})$ is not cut at height $\Omega(\log \frac{\eps r}{ w_i})$.
%
Moreover, by Lemma~\ref{lemma:decomp_net_tree} and the choice of a sufficiently large $r_1$, such event happens with high probability, over the randomness of $\delta$.
Overall, we prove that with high (constant) probability, 
$|\calH(r_1, r_2)| \leq \left|\{ \bigcup_{j\leq i\leq l}
\Bdt_i(u, \frac{a}{w_i})
\mid j\in [l], u\in R(a) \}\cup \left\{\emptyset \right\}\right|$, which is clearly bounded by
$O(|R(a)|\cdot l+1)=O(m^2)$.


\subsection{Proof of Theorem~\ref{theorem:weighted}}
\label{section:weighted_proof}

In the remaining of this section, let $\lambda:= \frac{\eps\cdot (1-5c\eps)}{20 (1+4c\eps)} = \frac{\eps\cdot (1-10\eps)}{20 (1+8\eps)}$ be the one defined in Lemma~\ref{lemma:ball_laminar}.
We first bound the number of critical intervals by Lemma~\ref{lemma:num_critical_interval}. Then we bound $|\calH(a,b)|$ for each critical interval $I:=[a,b)$ by Lemma~\ref{lemma:upper_bound_critical_interval}. 
Theorem~\ref{theorem:weighted} is a corollary of these two lemmas.

\begin{lemma}
	\label{lemma:num_critical_interval}
	$[0,\infty)$ can be partitioned into at most $2m^4$ critical intervals.
\end{lemma}

\begin{proof}
	Observe that critical intervals may be constructed in the following greedy manner. We start from $a = 0$, and find the smallest $b$ such that $[a, b]$ violates either item 1 or item 2 of Definition~\ref{definition:critical_interval}. Once such a value $b$ is found, we define $[a, b)$ as a critical interval, and restart with $a := b$.
	
	Hence, it suffices to count how many violations can happen.
	Suppose $I := [a, b)$ is a critical interval and $I \cup \left\{b\right\}$ violates item 1, i.e., there exists $i\in [l]$ such that $R_i(a) \neq R_i(b)$.
	This implies $|R_i(b)| \leq |R_i(a)| - 1$ by item 2 of Lemma~\ref{lemma:basic_center_prop}. Then by item 1 of Lemma~\ref{lemma:basic_center_prop}, we conclude that
	such violation can happen at most $m$ times for any $i$. 
	In total, it can happen at most $m^2$ times for all $i$.
	Therefore, we divide $[0,\infty)$ into at most $m^2+1$ intervals which do not violate item 1.

	Now suppose a critical interval $I = [a, b)$ is a subset of a maximal interval $I' := [a', b')$ such that item $1$ is not violated in $I'$ but item 2 may be, and $I \cup \left\{b\right\}$ violates item 2.
	Then there exists $i\in [l]$ and $x \in R(a) = R(a')$, such that $
	\Bd_i(x, \frac{a}{w_i})
	\neq
	\Bd_i(x, \frac{b}{w_i})
	$.
	Observe that $\Bd_i(x, \frac{a}{w_i}) \subseteq \Bd_i(x, \frac{b}{w_i})$,
	and that $\left|
	\Bd_i(x, \frac{r}{w_i})
	\right| \leq |H_i|$ for any $x\in X$ and $r\geq 0$. Hence, such violation cannot happen more than $|H_i|$ times for any fixed $i\in [l]$ and $x\in R(a')$. 
	In total, since $\sum_{i \in [l]}{|H_i|} = m$ and $|R(a')|\leq m$, it cannot happen more than $m^2$ times for all $i\in [l]$ and $x \in R(a')$.
	Hence, the interval $I'$ can be divided into at most $m^2+1$ critical intervals.

	In conclusion, there are at most $(m^2+1) \cdot (m^2+1) \leq 2 m^4$ critical intervals.
\end{proof}

Next, we bound $|\calH(a, b)|$ for a bounded interval $[a, b)$ with $b \leq 2 a$ in Lemma~\ref{lemma:bounded_interval}.
First, we can see that each concentric ball
$\Bdt_i(x, \frac{r}{w_i}), i\in [l]$ 
can be hung to a net points $\Par^{(\zeta_i)}(x)$ in the hierarchical net $\left\{N_j\right\}_j$ 
using the smooth property (Lemma~\ref{lemma:ball_equal}) of the $\delta$ function.
%
Then we apply the packing property of the doubling metric to bound the number of all
possible chains of points $\{\Par^{(\zeta_i)}(x) \mid i\in [l] \}$ relevant to non-empty concentric balls, and we finally relate this to an upper bound of $|\calH(a, b)|$.
We note that Lemma~\ref{lemma:bounded_interval} holds for any $\eps$-smoothed distance function $\delta$ (thus no randomness is required here). Furthermore, it does not require the interval to be critical. 
In fact, we only need the property that the interval is bounded.
%

\begin{lemma}[bounded interval]
	\label{lemma:bounded_interval}
	Consider $I := [a, b)$, with $b \leq 2a$. It holds that 
	$$|\calH(a, b)| \leq O(\frac{1}{\lambda})^{\DDim(M)} \cdot m^2.
	$$
\end{lemma}

\begin{proof}
	For any $i\in [l]$, define integer $\zeta_i$ such that $2^{\zeta_i - 1} \leq \lambda \cdot \frac{a}{w_i} < 2^{\zeta_i}$.
	Since $w_1 \geq w_2 \geq \ldots \geq w_l$, we have $\zeta_1  \leq  \zeta_2 \leq \ldots \leq \zeta_l$.
	By Lemma~\ref{lemma:ball_equal}, for $i \in [l]$, $x\in X$ and $r\geq a$,
	$
	\Bdt_i(x, \frac{r}{w_i})
	=
	\Bdt_i(\Par^{(\zeta_i)}(x), \frac{r}{w_i})
	$.
	This intuitively means that we can ``hang'' the center $x$ to $\Par^{(\zeta_i)}(x)$ for each $i$, as noted in Remark~\ref{remark:hang}.
	Hence,
	\begin{equation}
	\label{eqn:hang}
	\bigcup_{i \in [l]}{
		\Bdt_i(x, \frac{r}{w_i})
	} = \bigcup_{i \in [l]}{
		\Bdt_i(\Par^{(\zeta_i)}(x), \frac{r}{w_i})
	}.
	\end{equation}
	Let $Z_i(x, r) :=
	\Bdt_i(\Par^{(\zeta_i)}(x), r)
	$.
	Let $E(r) := \{ x\in X \mid \exists i : Z_i(x, \frac{r}{w_i})\neq \emptyset \}$.
	For $x \in E(r)$ define $t(x, r)$ to be the smallest integer $i \in [l]$ such that $ Z_i(x, \frac{r}{w_i}) \neq \emptyset $ (noting that such $i$ must exist as $x \in E(r)$). So by Equation~\eqref{eqn:hang},
	we can see that
	\begin{equation}
	\bigcup_{i \in [l]}{
		\Bdt_i(x, \frac{r}{w_i})
	} = \bigcup_{t(x, r) \leq i \leq l}{Z_i(x, \frac{r}{w_i})}.
	\end{equation}
	So, we can rewrite $|\calH(a,b)|$ as follows:
	\begin{align}
	|\calH(a, b)|
	&= \Bigl|\Bigl\{ \bigcup_{i\in [l]}{
		\Bdt_i(x, \frac{r}{w_i})
	} \mid x\in X, r\in I \Bigr\}\Bigr| \nonumber \\
	&= 1 + \Bigl|\Bigl\{\,\, \bigcup_{t(x, r)\leq i \leq l}{Z_i(x, \frac{r}{w_i})} \mid  r\in I ,x\in E(r)\Bigr\}\Bigr| \nonumber \\
	& = 1 + |U|, \label{eqn:one_plus}
	\end{align}
	where the ``1'' is for the empty set, and $U := \left\{ \bigcup_{t(x, r) \leq i \leq l}{Z_i(x, \frac{r}{w_i})} \mid  r\in I ,x\in E(r) \right\}$.
	Hence, it suffices to bound $|U|$.
	
	Now define $V := \left\{ (x, r) \mid  r \in I,x\in E(r) \right\}$ which is the set of relevant pairs $(x,r)$ that we need to consider.
	Define
	\begin{align*}
	W := \left\{ (u, i) \mid \exists (x, r) \in V : u = \Par^{(\zeta_{t(x, r)})}(x), i = t(x, r) \right\},
	\end{align*}
	and we will count $|U|$ by relating it to $W$ via $V$.
	
	\begin{claim}
		\label{cl:u_leq_W}
		$|U|  \leq |W| \cdot m$.
	\end{claim}
	
	\begin{proof}
		Define $\nu_1 : V \rightarrow U$ such that $\nu_1(x, r) := \bigcup_{t(x, r) \leq i \leq l}{
			Z_i(x, \frac{r}{w_i})
		}$. 
		%
		Define $\nu_1(Q)=\left\{\nu_1(x,r)\mid (x,r)\in Q \right\}$ for any subset $Q\subseteq V$.
		Define $\nu_2 : V \rightarrow W$ such that $\nu_2(x, r) := ( \Par^{(\zeta_{t(x, r)})}(x), t(x, r) )$.
		Let $\nu_2^{-1} : W \rightarrow 2^V$ such that 
		\[
		\nu_2^{-1}(u, i) := \left\{ (x, r) \in V \mid u = \Par^{(\zeta_{t(x, r)})}(x) , i = t(x, r) \right\}.
		\]
		
		By definition, we observe that
		$U = \bigcup_{(u, i) \in W}\nu_1(\nu_2^{-1}(u, i))$.
		Hence
		\begin{equation}
		|U| \leq \sum_{(u, i) \in W}{|\nu_1(\nu_2^{-1}(u, i))|}. \label{eqn:relate_u}	
		\end{equation}

		\noindent\textbf{Analyzing $|\nu_1(\nu_2^{-1}(u, i))|$.}
		Fix $(u, i) \in W$, and let $Q := \nu_2^{-1}(u, i)$.
		Now, we show that $|\nu_1(Q)| \leq m$.
		By definition, $Q = \left\{ (x, r) \in V \mid x\in E(r), \Par^{(\zeta_{t(x, r)})}(x) = u,  t(x, r) = i \right\}$.
		Hence, for $(x, r) \in Q$ we have
		\begin{align*}
		\nu_1(x, r)
		= \bigcup_{i \leq j \leq l}{Z_j(x, \frac{r}{w_j})}
		= \bigcup_{i \leq j \leq l}{
			\Bdt_j(\Par^{(\zeta_j)}(u), \frac{r}{w_j})
		},
		\end{align*}
		where the last equality is by $\zeta_1 \leq \zeta_2 \leq \ldots \leq \zeta_l$ and so $\Par^{(\zeta_j)}(x) = \Par^{(\zeta_j)}(u)$ for all $j\geq i$.
		For fixed $u \in N_{\zeta_i}$, observe that $1\leq |\bigcup_{i\leq j \leq l}{
			\Bdt_j(\Par^{(\zeta_j)}(u), \frac{r}{w_j})
		}|\leq m$. 
		On the other hand,
		by the definition of critical intervals,
		we have the monotonicity property: 
		$\bigcup_{i\leq j \leq l}{
			\Bdt_j(\Par^{(\zeta_j)}(u), \frac{r'}{w_j})
		}\subseteq \bigcup_{i\leq j \leq l}{
			\Bdt_j(\Par^{(\zeta_j)}(u), \frac{r}{w_j})
		}$ for any $r'\leq r$.
		Hence, we can see that
		\begin{align*}
		|\nu_1(\nu_2^{-1}(u,i))| \leq 
		\left|\Bigl\{ \bigcup_{i\leq j \leq l}{
			\Bdt_j(\Par^{(\zeta_j)}(u), \frac{r}{w_j})
		} \mid (x,r)\in Q, r\geq 0 \Bigr\} \right| \leq m.
		\end{align*}
		We conclude that $|\nu_1(Q)|\leq m$.
		Therefore, by~(\ref{eqn:relate_u}), we complete the proof of the claim.
	\end{proof}
	
	\noindent\textbf{Bounding $|W|$.}
	It remains to bound $|W|$. Let $W_i := \left\{ (u, i) \in W \right\}$ for $i\in [l]$.
	Observe that $|W| = \sum_{i \in [l]}{|W_i|}$. Now fix some $i \in [l]$, and we are to bound $|W_i|$.
	
	Consider some $(u, i) \in W_i$. By definition, there exists $(x, r) \in V$, such that $u = \Par^{(\zeta_{t(x, r)})}$ and $i = t(x, r)$.
	By the definition of $t(x, r)$, this implies that $Z_i(x, \frac{r}{w_i}) \neq \emptyset$, which is equivalent to
	$
	\Bdt_i(u, \frac{r}{w_i})
	\neq \emptyset$.
	Since
	$
	\Bdt_i(u, \frac{r}{w_i})
	\subseteq
	\Bdt_i(u, \frac{b}{w_i})
	$, we conclude that for \emph{any} $(u, i) \in W_i$,
	$
	\Bdt_i(u, \frac{b}{w_i})
	\neq \emptyset$.
	
	For $y \in H_i$, define $P_y := \{ u \mid (u, i) \in W_i,  y \in
	\Bdt_i(u, \frac{b}{w_i})
	\}$. 
	Since $P_y \subseteq N_{\zeta_i}$, $P_y$ is a $2^{\zeta_i}$-packing. Moreover, $\diam(P_y) \leq O(1) \cdot \frac{b}{w_i}$.
	By the packing property (Fact~\ref{fact:packing}), we have
	$|P_y| \leq O\left(\frac{b}{2^{\zeta_i}\cdot w_i}\right)^{\DDim(M)}
	= O\left( \frac{b}{\lambda a} \right)^{\DDim(M)} =O\left(\frac{1}{\lambda}\right)^{\DDim(M)}$, where we use $b \leq 2a$.
	Since for any $(u, i) \in W_i$, $
	\Bdt_i(u, \frac{b}{w_i})
	\neq \emptyset$,
	we can see that $|\bigcup_{y \in H_i}{P_y}| \geq |W_i|$.
	We conclude that
	\begin{align*}
	|W_i| \leq \sum_{y \in H_i}{|P_y|}
	\leq |H_i|\cdot O\left(\frac{1}{\lambda}\right)^{\DDim(M)},
	\end{align*}
	Therefore, 
	\begin{equation}\label{eqn:W_upper_bound}
	|W| = \sum_{i \in [l]}{|W_i|}
	\leq \sum_{i \in [l]}{|H_i| \cdot O\left(\frac{1}{\lambda}\right)^{\DDim(M)} }
	= m \cdot O\left(\frac{1}{\lambda}\right)^{\DDim(M)}.
	\end{equation}
	
	\noindent\textbf{Concluding the Lemma.}
	By Equations~\eqref{eqn:one_plus}, \eqref{eqn:W_upper_bound}, and Claim~\ref{cl:u_leq_W}, we conclude that
	$|\calH(a, b)| \leq m^2 \cdot O\left(\frac{1}{\lambda}\right)^{\DDim(M)}$, as required.
\end{proof}

As we discussed before, the most difficult case is to handle very long critical intervals, where we cannot use Lemma~\ref{lemma:bounded_interval} to get a finite bound.
To resolve the issue, we observe that very long critical intervals must have nice structural properties, as described in Lemmas~\ref{lemma:long_interval}-\ref{lemma:subset}.
Lemma~\ref{lemma:long_interval} concerns the intersection between two concentric balls centered at any $p,q\in R(a)$. Roughly speaking, there exists an integer $j$ such that $\Bd_i(p,\frac{a}{w_i})$ and $\Bd_i(q,\frac{a}{w_i})$ do not intersect for any $i\leq j-1$, and are equal for any $i\geq j$.
More precisely, we have the following lemma. 

\begin{lemma}[intersection between concentric balls] 
	\label{lemma:long_interval}
	Suppose $I := [a, b)$ is a critical interval with $b>6 a$. 
	For any $p,q\in R(a)$, there exists an integer $j\in [l+1]$ such that 
	\begin{enumerate}
		\item If $j \leq l$ then $\Bd_j(p,\frac{a}{w_j})\cap \Bd_j(q,\frac{a}{w_j}) \neq \emptyset$.
		Moreover, for all integer $i$ such that $j\leq i\leq l$, $\Bd_i(p,\frac{a}{w_i})= \Bd_i(q,\frac{a}{w_i})$.
		\item For $1\leq i\leq j-1$, $\Bd_i(p,\frac{a}{w_i})\cap \Bd_i(q,\frac{a}{w_i}) = \emptyset$.
	\end{enumerate}
\end{lemma}

\begin{proof}
	Let $j' \in [l]$ be the smallest integer such that $\Bd_{j'}(p,\frac{a}{w_{j'}})\cap \Bd_{j'}(q,\frac{a}{w_{j'}}) \neq \emptyset$.
	If there is no such $j'$, then let $j:= l+1$ and let $j := j'$ otherwise.
	The Lemma holds trivially if $j=l+1$, and we consider the case that $j \leq l$.
	By the definition of $j$, we have $\Bd_j(p,\frac{a}{w_j})\cap \Bd_j(q,\frac{a}{w_j}) \neq \emptyset$. 
	This implies that there exists some $y\in H_{j}$ such that $y\in \Bd_j(p,\frac{a}{w_j})\cap \Bd_j(q,\frac{a}{w_j})$.
	Then we have $d(y, p) \leq \frac{a}{w_j} $ and $d(y, q) \leq  \frac{a}{w_j}$.
	Hence, by the triangle inequality,
	\begin{equation}
	\label{eqn:d_p_q}
		d(p, q)\leq  d(y,p)+d(y,q)\leq \frac{2a}{w_j}.
	\end{equation}
	By contradiction, assume that there exists an integer $j\leq i\leq l$ such that 
	$\Bd_i(p,\frac{a}{w_i})\neq \Bd_i(q,\frac{a}{w_i})$.
	W.l.o.g., assume $\Bd_i(p,\frac{a}{w_i})\setminus \Bd_i(q,\frac{a}{w_i})\neq \emptyset$ and $y'\in \Bd_i(p,\frac{a}{w_i})\setminus \Bd_i(q,\frac{a}{w_i})$.
	Since $I$ is a critical interval and $\frac{b}{2} \in I$, we have $\Bd_i(q,\frac{a}{w_i})
	=
	\Bd_i(q, \frac{b}{2w_i})
	$. Hence $y'\notin \Bd_i(q, \frac{b}{2w_i})$ which implies $d(y',q)>\frac{b}{2w_i}$.
	By the definition of $y'$, we have $d(y',p)\leq \frac{a}{w_i}$.
	Then by the triangle inequality,
	\begin{align*}
	d(p, q)\geq d(y',p)-d(y',q)> \frac{b}{2 w_i}-\frac{a}{w_i} 
	\stackrel{b>6 a}{>}\frac{2a}{w_i} \stackrel{j\leq i}{\geq} \frac{2a}{w_j},
	\end{align*}
	which contradicts with~(\ref{eqn:d_p_q}). 
	It completes the proof.
\end{proof}

Next, we make use of Lemma~\ref{lemma:long_interval} to investigate the intersection between a range $\bigcup_{i\in [l]} \Bdt_i(x,\frac{r}{w_i})$ $(x\in X, r\in I)$ and concentric balls with centers in $R(a)$, inside a long critical interval $I$.
More specifically, for any $p,q\in R(a)$, imagine that we increase $j$ from $1$ to $l$.
Initially, 
$
\Bdt_i(x, \frac{r}{w_i})
$
can only intersect with one of 
$
\Bd_i(p, \frac{r}{w_i})
$ or
$
\Bd_i(q, \frac{r}{w_i})
$;
once
$
\Bdt_i(x, \frac{r}{w_i})
$
intersects both of them,
the two intersections must be the same.
More precisely, we have the following lemma. 

\begin{lemma}[intersection between ranges and concentric balls]
	\label{lemma:large_r}
	Suppose $I := [a, b)$ is a critical interval with $b>2^9 a$.
	Then for any $x \in X$, $r \in [a,\frac{b}{2^9}]$ and $p, q \in R(a)$, there exists an integer $j\in [l+1]$ such that: 
	\begin{enumerate}
		\item Either $\bigcup_{1\leq i\leq j-1}
		\left(
		\Bdt_i(x, \frac{r}{w_i})
		\cap
		\Bd_i(p, \frac{r}{w_i})
		\right) = \emptyset$
		or 
		$\bigcup_{1\leq i\leq j-1} \left(
		\Bdt_i(x, \frac{r}{w_i})
		\cap 
		\Bd_i(q, \frac{r}{w_i})
		\right) = \emptyset$.
		\item For $j\leq i\leq l$,
		$
		\Bdt_i(x, \frac{r}{w_i})
		\cap 
		\Bd_i(p, \frac{r}{w_i})
		= 
		\Bdt_i(x, \frac{r}{w_i})
		\cap
		\Bd_i(q, \frac{r}{w_i})
		$.
	\end{enumerate}
\end{lemma}

\begin{proof}
	Let $j$ be the smallest integer such that
	$
	\Bdt_j(x, \frac{r}{w_j})
	\cap 
	\Bd_j(p, \frac{r}{w_j})
	\cap 
	\Bd_j(q, \frac{r}{w_j})
	\neq \emptyset$.
	If such $j$ does not exist, we let $j=l+1$.
	We show that items 1 and 2 hold for this $j$.

	\noindent \textbf{Item 1.} We first prove item 1.
	W.l.o.g., suppose for the contrary that there exists $1\leq i_1 \leq  i_2\leq j-1$ such that
	$
	\Bdt_{i_1}(x, \frac{r}{w_{i_i}})
	\cap 
	\Bd_{i_1}(p, \frac{r}{w_{i_1}})
	\neq \emptyset$ and
	$
	\Bdt_{i_2}(x, \frac{r}{w_{i_2}})
	\cap 
	\Bd_{i_2}(q, \frac{r}{w_{i_2}})
	\neq \emptyset$.
	Assume
	$y\in 
	\Bdt_{i_2}(x, \frac{r}{w_{i_2}})
	\cap 
	\Bd_{i_2}(q, \frac{r}{w_{i_2}})
	$. By the triangle inequality, 
	\[
	d(x, q) \leq d(x,y)+d(y,q)	\stackrel{\text{Lemma \ref{lemma:distortion}}}{\leq} (1+8\eps)\cdot \delta(x,y) + d(y,q)\leq \frac{3r}{w_{i_2}}.
	\]
	Similarly, we can prove $d(x, p) \leq \frac{3r}{w_{i_1}}$.
	Hence by the fact that $w_{i_1}\geq w_{i_2}$ and $r\leq \frac{b}{2^9}$,
	\begin{align}
	\label{ineq:distance_upper}
	 d(p, q)\leq d(x,p)+d(x,q)\leq \frac{3r}{w_{i_1}}+ \frac{3r}{w_{i_2}} \leq \frac{6r}{w_{i_2}} < \frac{b}{64 w_{i_2}}.
	\end{align}
	However, by the definition of $j$, we know that
	$
	\Bdt_{i_2}(x, \frac{r}{w_{i_2}})
	\cap 
	\Bd_{i_2}(p, \frac{r}{w_{i_2}})
	= \emptyset$.
	Since $I$ is a critical interval, we have
	$
	\Bd_{i_2}(p, \frac{r}{w_{i_2}})
	=
	\Bd_{i_2}(p, \frac{b}{2w_{i_2}})
	$.
	It implies that
	$
	\Bdt_{i_2}(x, \frac{r}{w_{i_2}})
	\cap
	\Bd_{i_2}(p, \frac{b}{2w_{i_2}})
	= \emptyset$.
	Since
	$y \in
	\Bdt_{i_2}(x, \frac{r}{w_{i_2}})
	$, we know
	that $y\notin
	\Bd_{i_2}(p, \frac{b}{2w_{i_2}})
	$.
	Then by the triangle inequality, 
	\[
	d(p,q)\geq d(y,p)-d(y,q)> \frac{b}{2w_{i_2}} - \frac{r}{w_{i_2}} \stackrel{r\leq b/2^9}{>} \frac{b}{4 w_{i_2}},
	\]
	which contradicts Inequality~\eqref{ineq:distance_upper}.
	
	\noindent \textbf{Item 2.} Next, we prove item 2.
	Since $I$ is a critical interval, we have
	$
	\Bd_j(p, \frac{r}{w_j})
	=
	\Bd_j(p, \frac{a}{w_j})
	$
	by item 2 of Definition~\ref{definition:critical_interval}.
	Then by Lemma~\ref{lemma:long_interval}, either
	$
	\Bd_j(p, \frac{r}{w_j})
	\cap
	\Bd_j(q, \frac{r}{w_j})
	= \emptyset$
	or $
	\Bd_j(p, \frac{r}{w_j})
	= 
	\Bd_j(q, \frac{r}{w_j})
	$.
	Since $
	\Bdt_j(x, \frac{r}{w_j})
	\cap 
	\Bd_j(p, \frac{r}{w_j})
	\cap 
	\Bd_j(q, \frac{r}{w_j})
	\neq \emptyset$, it must be the case that
	$
	\Bd_j(p, \frac{r}{w_j})
	= 
	\Bd_j(q, \frac{r}{w_j})
	\neq \emptyset$.
	Moreover, Lemma~\ref{lemma:long_interval} further implies that
	\[
	\Bd_i(p, \frac{r}{w_i})
	= 
	\Bd_i(p, \frac{a}{w_i})
	= 
	\Bd_i(q, \frac{a}{w_i})
	= 
	\Bd_i(q, \frac{r}{w_i})
	\]
	for any $j\leq i\leq l$.
	This implies item 2.
\end{proof}

The following lemma follows easily from  Lemma~\ref{lemma:large_r}.
Roughly speaking, it says that any range
$\bigcup_{i\in [l]}
\Bdt_i(x, \frac{r}{w_i})
$
($x\in X$) is contained in some concentric ball
$\bigcup_{i\in [l]}
\Bd_i(p, \frac{a}{w_i})
$ centered at $p\in R(a)$.
This will be useful shortly in Lemma~\ref{lemma:upper_bound_critical_interval}, 
in which we show 
$\Bdt_i(x, \frac{r}{w_i})= \bigcup_{j\leq i\leq l} \Bd_i(p,\frac{a}{w_i})$ for some $j$ and $p\in R(a)$
(hence, in a sense, we can ``hang" $x$ on an $a$-representative).

\begin{lemma}[any range is a subset of some concentric ball]
	\label{lemma:subset}
	Suppose $I := [a, b)$ is a critical interval with $b>2^9 a$.
	Then for any $x \in X$ and $r \in [a,\frac{b}{2^9}]$, there exists $p\in R(a)$ such that 
	$\bigcup_{i\in [l]} 
	\Bdt_i(x, \frac{r}{w_i})
	\subseteq \bigcup_{i\in [l]} 
	\Bd_i(p, \frac{a}{w_i})
	$.
\end{lemma}

\begin{proof}
	Let $p\in R(a)$ be the point such that the intersection number
	$$\left|\left(\bigcup_{i\in [l]} 
	\Bdt_i(x, \frac{r}{w_i})
	\right)\cap \left(\bigcup_{i\in [l]} 
	\Bd_i(p, \frac{a}{w_i})
	\right) \right|=\left|\bigcup_{i\in [l]} \left(
	\Bdt_i(x, \frac{r}{w_i})
	\cap 
	\Bd_i(p, \frac{r}{w_i})
	\right) \right|$$ 
	is maximized, where the equality is because $I$ is a critical interval.
	By contradiction, assume there exists $j\in [l]$ and $y\in H_j$ such that $y\in
	\Bdt_j(x, \frac{r}{w_j})
	$ but $y\notin 
	\Bd_j(p, \frac{r}{w_j})
	$. 
	By item 3 of Lemma~\ref{lemma:basic_center_prop}, there must exist $q\in R(r)=R(a)$ such that $y\in 
	\Bd_j(q, \frac{r}{w_j})
	$.
	It implies that $
	\Bdt_j(x, \frac{r}{w_j})
	\cap 
	\Bd_j(p, \frac{r}{w_j})
	\neq 
	\Bdt_j(x, \frac{r}{w_j})
	\cap 
	\Bd_j(q, \frac{r}{w_j})
	$ and $
	\Bdt_j(x, \frac{r}{w_j})
	\cap 
	\Bd_j(q, \frac{r}{w_j})
	\neq \emptyset$.
	By Lemma~\ref{lemma:large_r}, there must exist $j'\in [l+1]$ such that
	\begin{enumerate}
		\item Either $\bigcup_{1\leq i\leq j'-1} \left(
		\Bdt_i(x, \frac{r}{w_i})
		\cap 
		\Bd_i(p, \frac{r}{w_i})
		\right) = \emptyset$
		or
		$\bigcup_{1\leq i\leq j'-1} \left(
		\Bdt_i(x, \frac{r}{w_i})
		\cap 
		\Bd_i(q, \frac{r}{w_i})
		\right) = \emptyset$.
		\item For $j'\leq i\leq l$, $
		\Bdt_i(x, \frac{r}{w_i})
		\cap 
		\Bd_i(p, \frac{r}{w_i})
		= 
		\Bdt_i(x, \frac{r}{w_i})
		\cap 
		\Bd_i(q, \frac{r}{w_i})
		$.
	\end{enumerate}
	Since $
	\Bdt_j(x, \frac{r}{w_j})
	\cap 
	\Bd_j(p, \frac{p}{w_j})
	\neq 
	\Bdt_j(x, \frac{r}{w_j})
	\cap 
	\Bd_j(q, \frac{r}{w_j})
	$, we have $j'\geq j+1$. 
	On the other hand, since 
	$
	\Bdt_j(x, \frac{r}{w_j})
	\cap 
	\Bd_j(q, \frac{r}{w_j})
	\neq \emptyset$, 
	we have $\bigcup_{1\leq i\leq j'-1} \left(
	\Bdt_i(x, \frac{r}{w_i})
	\cap 
	\Bd_i(q, \frac{r}{w_i})
	\right) \neq \emptyset$.
	Hence,
	\begin{align*}
	\bigcup_{1\leq i\leq j'-1} \left(
	\Bdt_i(x, \frac{r}{w_i})
	\cap 
	\Bd_i(p, \frac{r}{w_i})
	\right) = \emptyset.
	\end{align*}
	Overall, we conclude that
	\begin{align*}
	\left|\bigcup_{i\in [l]} \left(
	\Bdt_i(x, \frac{r}{w_i})
	\cap 
	\Bd_i(q, \frac{r}{w_i})
	\right) \right|
	\geq \left|\bigcup_{i\in [l]} \left(
	\Bdt_i(x, \frac{r}{w_i})
	\cap 
	\Bd_i(p, \frac{r}{w_i})
	\right) \right|+1,
	\end{align*}
	which is a contradiction with the choice of $p$.
	It completes the proof.
\end{proof}

Now, we are ready to prove the main lemma which bounds $|\calH(I)|$ for any critical interval $I$.

\begin{lemma}[any critical interval]
	\label{lemma:upper_bound_critical_interval}
	Suppose $I=[a, b)$ is a critical interval. Then with probability at least $1 - \frac{\tau}{2 m^4}$,
	\begin{align*}
	|\calH(a, b)| \leq O\left(\frac{1}{\lambda}\right)^{\DDim(M)} \cdot \log{\frac{m}{\tau}} \cdot m^2.
	\end{align*}
\end{lemma}

\begin{proof}
	If $b\leq 2^9 a$, then using Lemma~\ref{lemma:bounded_interval}, we can see that
	\begin{align*}
	|\calH(a, b)| \leq O\left(\frac{1}{\lambda}\right)^{\DDim(M)} \cdot m^2.
	\end{align*}
	In the following, we consider the case that $b>2^9 a$ (so Lemma~\ref{lemma:subset} can be applied).
	Let
	\begin{align*}
	s := \lceil \log{a} + \log{\Theta(\DDim(M))} + 6\log{m} + \log{\frac{1}{\tau}}  \rceil,
	\end{align*}
	and define
	\begin{align*}
	r_1 := \frac{2^s}{\lambda}, \quad r_2 := \frac{b}{2^9} .
	\end{align*}
	If $r_1>r_2$, we define $\calH(r_1,r_2)=\emptyset$.
	Then we always have
	\begin{align*}
	|\calH(a, b)| \leq |\calH(a, r_1)| + |\calH(r_1, r_2)| + |\calH(r_2, b)|.
	\end{align*}
	By applying Lemma~\ref{lemma:bounded_interval}, we have
	\begin{align*}
	|\calH(a, r_1)| \leq O\left(\frac{1}{\lambda}\right)^{\DDim(M)} \cdot \log{\frac{m}{\tau}} \cdot m^2,
	\end{align*}
	and
	\begin{align*}
	|\calH(r_2, b)| \leq O\left(\frac{1}{\lambda}\right)^{\DDim(M)} \cdot m^2.
	\end{align*}
	
	\noindent\textbf{Upper Bound $|\calH(r_1, r_2)|$.}
	It remains to bound $|\calH(r_1, r_2)|$.
	Let $s_i$ be the integer such that $2^{s_i-1} \leq \frac{\lambda r_1}{w_i} < 2^{s_i }$.
	For any $i\in [l]$ and $u\in R(a)$, let event $\mathcal{E}^{(u)}_i$ be
	\begin{align*}
	\forall v \in N_{s_i} : \Bd_i(u, \frac{a}{w_i}) \not\subseteq \Des(v^{(s_i)}),
	\end{align*}
	Observe that for any $u \in R_i(a)$,
	\begin{align*}
	\diam(\Bd_i(u,\frac{a}{w_i}))
	\leq 2 \cdot \frac{a}{w_i}.
	\end{align*}
	By Definition~\ref{definition:basic_center} and~\ref{definition:critical_interval}, all critical intervals $I$ and all representatives in $R(a)$ are independent of the choice of $\delta$.
	Hence, $\Bd_i(u,\frac{a}{w_i})$ is also independent of the choice of $\delta$.
	Therefore, by Lemma~\ref{lemma:decomp_net_tree},
	\begin{align*}
	\Pr[\mathcal{E}^{(u)}_i]
	\leq O(\DDim(M)) \cdot \frac{2a}{w_i \cdot 2^{s_i}}
	< O(\DDim(M)) \cdot \frac{2a}{\lambda r_1} 
	\leq \frac{\tau}{2m^6}.
	\end{align*}
	
	\noindent\textbf{Conditioning on $\bigcap_{i \in [l], u \in R(a)}{\overline{ \mathcal{E}^{(u)}_i}}$.}
	We condition on the event that none of $\mathcal{E}^{(u)}_i$'s happens, and this event helps us to obtain an even stronger smooth property, which is stated in Claim~\ref{claim:large_ball_equal}.
	\begin{claim}
		\label{claim:large_ball_equal}
		For any $r \in [r_1, r_2)$ and any $x\in X$, at least one of the following holds:
		\begin{enumerate}
			\item $\bigcup_{i\in [l]}
			\Bdt_i(x, \frac{r}{w_i})
			= \emptyset$.
			\item There exists an integer $j\in [l] $ and $u\in R(a)$, such that 
			$\bigcup_{i\in [l]}
			\Bdt_i(x, \frac{r}{w_i})
			= \bigcup_{j\leq i\leq l} \Bd_i(u,\frac{a}{w_i})$.
		\end{enumerate}
	\end{claim}
	\begin{proof}
		We only need to consider the case that $\bigcup_{i\in [l]} B_i(x,\frac{r}{w_i}) \neq \emptyset$.
		Since $b>2^9 a$ and $r\leq r_2< \frac{b}{2^9}$, there must exist $u\in R(a)$ such that 
		$\bigcup_{i\in [l]} 
		\Bdt_i(x, \frac{r}{w_i})
		\subseteq 
		\bigcup_{i\in [l]} \Bd_i(u,\frac{a}{w_i})$, by Lemma~\ref{lemma:subset}.
		Suppose $j \in [l]$ is the smallest integer such that 
		$
		\Bdt_j(x, \frac{r}{w_j})
		\cap \Bd_j(u, \frac{a}{w_j})\neq \emptyset$ (observing that such $j$ must exist).
		We then show that $\bigcup_{i\in [l]} 
		\Bdt_i(x, \frac{r}{w_i})
		= \bigcup_{j\leq i\leq l} \Bd_i(u,\frac{a}{w_i})$, and so
		it suffices to show $\bigcup_{j\leq i\leq l} \Bd_i(u,\frac{a}{w_i})\subseteq \bigcup_{i\in [l]} 
		\Bdt_i(x, \frac{r}{w_i})
		$.
		We separate the argument into two cases: $i \geq j + 1$ and $i = j$.
		The two cases are handled differently, where we make use of the the gap of the weight function in the $i \geq j + 1$ case, and
		in the $i = j$ case we make use of the cross-free property of $\delta$.
		
		\vspace{0.1cm}
		\noindent\textbf{Showing $\bigcup_{j+1\leq i\leq l} \Bd_i(u,\frac{a}{w_i}) \subseteq \bigcup_{i\in [l]} 
			\Bdt_i(x, \frac{r}{w_i})
		$.} 	
		The key is to use the fact that $w$ is a gap-2 weight function.
		It suffices to prove for any $i$ such that $j+1\leq i\leq l$ and $x'\in \Bd_i(u,\frac{a}{w_i})$, we have 
		$x'\in 
		\Bdt_i(x, \frac{r}{w_i})
		$.
		We first bound the distance $d(x,u)$.
		Recall that 
		$
		\Bdt_j(x, \frac{r}{w_j})
		\cap \Bd_j(u,\frac{a}{w_j})\neq \emptyset$. 
		Pick $y$ such that $y\in 
		\Bdt_j(x, \frac{r}{w_j})
		\cap \Bd_j(u,\frac{a}{w_j})$. 
		Hence $\delta(x,y)\leq \frac{r}{w_j}$ and $d(u,y)\leq \frac{a}{w_j}$.
		Since $r\geq r_1\geq 100 a$ and $\eps\leq \frac{1}{100}$, we have
		\[
		d(x,u) \leq d(x,y)+d(u,y) \stackrel{\text{Lemma \ref{lemma:distortion}}}{\leq} (1+8\eps) \cdot \delta(x,y) + d(u,y)\leq \frac{(1+8\eps)r}{w_{j}}+\frac{a}{w_j}\leq \frac{1.1 r}{w_{j}}.
		\]
		%
		Now we are ready to show $x'\in 
		\Bdt_i(x, \frac{r}{w_i})
		$.
		Since $x'\in \Bd_i(u,\frac{a}{w_i})$, $d(x',u)\leq \frac{a}{w_i}$.
		Again since $r\geq r_1\geq 100 a$ and $\eps\leq \frac{1}{100}$,
		\begin{align*}
		& \delta(x,x')\stackrel{\text{Lemma \ref{lemma:distortion}}}{\leq} \frac{1}{1-8\eps} \cdot d(x,x') \leq \frac{1}{1-8\eps} \cdot (d(x,u)+d(x',u)) \leq \frac{1}{1-8\eps} \cdot (\frac{1.1 r}{w_{j}}+\frac{a}{w_i})\leq \frac{1.3 r}{w_j}+\frac{0.02 r}{ w_i}.
		\end{align*}
		Note that $w$ is a gap-2 weight function which implies $w_j\geq 2w_i$. 
		Thus,
		\[
		\delta(x,x')\leq \frac{1.3 r}{w_j}+\frac{0.02 r}{ w_i}\leq \frac{1.3 r}{2 w_i}+\frac{0.02 r}{ w_i}\leq \frac{r}{w_i}.
		\]
		It implies that $x'\in 
		\Bdt_i(x, \frac{r}{w_i})
		$.

		\vspace{0.1cm}
		\noindent\textbf{Showing $\Bd_j(u,\frac{a}{w_j}) \subseteq 
			\Bdt_j(x, \frac{r}{w_j})
		$.}
		The key is to use the smooth property of $\delta$.
		Recall that $2^{s_j-1} \leq \frac{\lambda r_1}{w_j} < 2^{s_j }$.
		Because $\mathcal{E}^{(u)}_j$ does not happen, there exists exactly one $y \in N_{s_j}$ such that $\Bd_j(u,\frac{a}{w_j}) \subseteq \Des(y^{(s_j)})$.
		%
		%
		Since $r_1\leq r$, we have $2^{s_j-1}\leq \frac{\lambda r}{w_j}$.
		Then by Lemma~\ref{lemma:ball_laminar}, 
		either $\Des(y^{(s_j)}) 
		\subseteq 
		\Bdt(x, \frac{r}{w_j})$ 
		or 
		$\Des(y^{(s_j)}) 
		\cap 
		\Bdt(x, \frac{r}{w_j}) = \emptyset$.
		Since $\Bd_j(u,\frac{a}{w_j}) \subseteq \Des(y^{(s_j)})$ (by the event $\overline{\mathcal{E}^{(u)}_j}$) and 
		$
		\Bdt_j(x, \frac{r}{w_j})
		\cap \Bd_j(u,\frac{a}{w_j})\neq \emptyset$, we have 
		$
		\Bdt_j(x, \frac{r}{w_j})
		\cap \Des(y^{(s_j)}) \neq \emptyset$. 
		Therefore, it has to be the case that $\Des(y^{(s_j)}) \subseteq \Bdt(x, \frac{r}{w_j})$.
		In conclusion, we have
		\begin{align*}
		\Bd_j(u,\frac{a}{w_j})
			\subseteq
			H_j\cap \Des(y^{(s_j)})  
			\subseteq 
			H_j\cap \Bdt(x, \frac{ r}{w_j})
			= 
			\Bdt_j(x, \frac{r}{w_j})
			.
		\end{align*}
		This finishes the proof of Claim~\ref{claim:large_ball_equal}.
	\end{proof}
	By Claim~\ref{claim:large_ball_equal}, we have
	\begin{align*}
	\calH(r_1, r_2)
	= \left\{ \bigcup_{i}{
		\Bdt_i(x, \frac{r}{w_i})
		\mid r \in [r_1, r_2), x\in X} \right\}
	\subseteq \left\{ \bigcup_{j\leq i\leq l}{
		\Bd_i(u, \frac{a}{w_i})
		} \mid j\in [l], u \in R(a) \right\}\cup \left\{\emptyset \right\}.
	\end{align*}
	Since $l\leq m$ and $|R(a)|\leq m$, we have $|\calH(r_1,r_2)|\leq m^2+1$.
	
	\noindent\textbf{Removing the Condition.}
	Using the union bound, we have
	\begin{align*}
	\Pr\left[\bigcap_{i \in [l], u \in R(a)}{\overline{ \mathcal{E}^{(u)}_i}}\right]
	\geq 1 - \sum_{i \in [l], u \in R(a)}{\Pr[\mathcal{E}^{(u)}_i]}
	\geq 1 - m^2 \cdot \frac{\tau}{2m^6}
	= 1 - \frac{\tau}{2m^4}.
	\end{align*}
	Therefore, with probability at least $1 - \frac{\tau}{2m^4}$, we have $|\calH(a, b)| \leq  O(\frac{1}{\lambda})^{\DDim(M)} \cdot \log{\frac{m}{\tau}} \cdot m^2 $. This concludes Lemma~\ref{lemma:upper_bound_critical_interval}.
\end{proof}

\noindent\textbf{Concluding Theorem~\ref{theorem:weighted}.}
By Lemma~\ref{lemma:num_critical_interval}, $[0,\infty)$ can be partitioned into $t\leq 2m^4$ critical intervals $[a_0=0,a_1),[a_1,a_2),\ldots,[a_{t-1},a_t=\infty)$. 
Then $|\ranges(\calF_H)|=|\calH(0,\infty)|\leq \sum_{i\in [t]} |\calH(a_{i-1},a_i)|$.
By Lemma~\ref{lemma:upper_bound_critical_interval}, for each $i\in [t]$, we have
\begin{align*}
\Pr\left[|\calH(a_{i-1},a_i)| \leq O\left(\frac{1}{\eps}\right)^{O(\DDim(M))} \cdot \log{\frac{|H|}{\tau}} \cdot  |H|^{2} \right] \geq 1 - \frac{\tau}{2m^4}.
\end{align*}
Hence we can finish the proof by applying the union bound.
\qed

\begin{remark}
	From the proof of Lemma~\ref{lemma:num_critical_interval}, 
	we can see that the property of the gap-2 weight function is only used to prove Claim~\ref{claim:large_ball_equal}.
	In fact, we can consider a gap-c weight function for any constant $c>1$. 
	The only difference is that we should choose $r_1$ to be sufficiently large and $\eps$ to be sufficiently small, e.g., $r_1\geq 100 a/(c-1)$ and $\eps\leq 1/100(c-1)$.
	By this modification, the bound for $|\range(\calF_H)|$ changes to $O\left(\frac{1}{\eps}\right)^{O(\DDim(M))}  \cdot  |H|^{6} \cdot \max\{ \log{\frac{|H|}{\tau}} , \log \frac{1}{c-1}\}$.
\end{remark}

Using a similar argument, we can generalize the above result to the case where the distance is taken to the power of $z$, as in the following corollary. The proof can be found in Appendix~\ref{section:weighted_z_ball}.

\begin{corollary}
	\label{corollary:weighted_z_ball}
	Suppose $M(X, d)$ is a metric space with a gap-$2$ weight function $w: X\rightarrow \R_{\geq 0}$.
	Let $z>0$, $0 < \epsilon \leq \frac{1}{100z}$ and $0 < \tau < 1$ be constant. There exists a random $\epsilon$-smoothed distance function $\delta$ (defined with respect to some random net tree), such that for $\calF := \left\{ w(x)\cdot \delta^z(x, \cdot) \mid x\in X \right\}$
	($\calF$ is defined with respect to the $z$-th power of the random smoothed distance function $\delta$),
	and any $H\subseteq X$,
	\begin{align*}
	\Pr_\delta\left[|\ranges(\calF_H)| \leq  O\left( \frac{1}{\epsilon} \right)^{O(\DDim(M))} \cdot \log{\frac{|H|}{\tau}} \cdot |H|^6  \right] \geq 1 - \tau,
	\end{align*}
	In addition, for $x, y \in X$, it holds that
	\begin{align*}
	(1 - O(\eps\cdot z))\cdot \delta^z(x, y)\leq d^z(x, y) \leq (1 + O(\eps \cdot z)) \cdot \delta^z(x, y).
	\end{align*}
	In other words, $\PDim_\tau(\calF) \leq O\left( \DDim(M)\cdot \log (1/\eps) + \log\log 1/\tau \right) $.
\end{corollary}

\section{Applications}
In this section, we provide three applications of our main result.
The major application is an efficient 
$\epsilon$-coreset construction algorithm for the $(k, z)$-clustering problem in doubling metrics.
The overall approach is to apply the Feldman-Langberg framework (Theorem~\ref{thm:fl}).
As noted in Section~\ref{section:intro}, one important building block is an $\alpha$-approximation for the weighted range space induced by the metric space.
This is done by combining the probabilistic dimension upper bound of the weighted range space (Corollary~\ref{corollary:weighted_z_ball}), and the $\alpha$-approximation lemma for the bounded probabilistic dimension (Lemma~\ref{lemma:restate_weak_app}).
Although the construction of the smoothed distance function in Corollary~\ref{corollary:weighted_z_ball} is quite involved, we only use it in the analysis. 
The algorithm is almost as simple as in the Euclidean case.
In particular, the core of the algorithm is a weighted sampling of the points in the \emph{original} metric.
In the analysis, we consider an \emph{auxiliary} range space resulted from Corollary~\ref{corollary:weighted_z_ball}, and we relate the sample on the original point set to a sample on the auxiliary range space. We show that the sample is a good approximation for the auxiliary space with high probability, and we translate it into a good coreset in the original space.
We elaborate the details in Section~\ref{sec:coreset}.

In Section~\ref{section:robust_coreset}, we introduce the construction of the robust coreset and its application to property testing. The construction of the robust coreset is simply a \emph{uniform} sample of points from the metric space. The key proof for the correctness is Lemma \ref{rctech}, that presents a simple (yet previously unknown) relationship between $\alpha$-approximation and robust coreset.
Having this lemma, we can then follow a similar argument as in the coreset construction (Section~\ref{sec:coreset}) to get a robust coreset in doubling metrics.
Finally, we discuss an application of the robust coreset to the property testing.

Another application is the construction of the centroid set and its application to accelerate the local search algorithms for the $(k, z)$-clustering problem in doubling metrics.
The centroid set is essentially an extension of a coreset, such that a $(1+\eps)$-approximate solution to the clustering objective is \emph{included} in the centroid set.
The centroid set was first considered by~\cite{DBLP:journals/dcg/Matousek00} and was applied to a constant approximation 
for the geometric $k$-means clustering problem in Euclidean space.
In a high level, our construction of the centroid set is similar with that in~\cite{DBLP:journals/dcg/Matousek00}.
But our construction does not rely on the specific properties in Euclidean spaces and the $k$-means objective.
We obtain a small sized centroid set for the $(k, z)$-clustering problem with arbitrary $k$ and $z$,
and for any doubling metric.
Recently, Friggstad et al.~\cite{DBLP:conf/focs/FriggstadRS16} showed that the local search algorithm 
actually gives a PTAS for the $(k, z)$-clustering problem in doubling metrics. 
For the special case of $k$-means in Euclidean spaces, they used the centroid set in~\cite{DBLP:journals/dcg/Matousek00} to improve the running time. However, a centroid set for doubling metrics was not known and hence the running time was not improved for more general doubling metrics.
Using our new result, we obtain a similar speedup comparable to theirs in Euclidean spaces.
The construction of the centroid set as well as its application is discussed in Section~\ref{sec:centroid}.

\subsection{Coreset Construction in Doubling Metrics}
\label{sec:coreset}
In this section, we present the construction of the coreset in doubling metrics.
Theorem~\ref{thm:coreset} is the formal statement of Theorem~\ref{theorem:main}.
\begin{theorem}
	\label{thm:coreset}
	We are given a doubling metric $M(X,d)$ with $X$ being a set of $n$ discrete points. Let real numbers $0<\e,\tau<1/100$, $z > 0$, and integer $k \geq 1$. There exists an algorithm running in $\poly(n)$ time,
	that constructs a weighted subset $S\subseteq X$ of size
	\[
	\Gamma := O\left( \frac{2^{O(z\log z)}k^3}{\e^{2}}\left(\DDim(M)\cdot \log(z/\e)+ \log k + \log \log (1/\tau)\right)+\frac{2^{O(z\log z)}k^2 \log (1/\tau)}{\e^{2}}\right),
	\]
	such that $S$ is an $\eps$-coreset for the $(k,z)$-clustering problem with probability at least $1-\tau$.
\end{theorem}

The algorithm for constructing the coreset is described in Algorithm~\ref{alg:coreset}, and this algorithm is almost the same as that for the Euclidean spaces in~\cite{FL11}.
The algorithm computes a weight for each point in the metric, and then independently sample a number of points following the distribution proportional to the weights.
As in~\cite{FL11}, the weight for a point in the metric is defined as an upper bound for its \emph{sensitivity}, which is defined in Definition~\ref{def:totalsen}.


Actually, in the \FL framework (Theorem~\ref{thm:fl}), apart from the $\alpha$-approximation lemma, it also requires that the sum of the sensitivity is bounded.
To bound the sensitivity, we use the result in~\cite{varadarajan2012sensitivity}, and show in Theorem~\ref{thm:totalsen} that the total sensitivity is bounded by a function of $k$ and $z$ and is independent of the doubling dimension.

\subsubsection{Sensitivity}
\begin{definition} [sensitivity for $(k,z)$-clustering]
\label{def:totalsen}
Given a metric space $M(X,d)$ for the $(k,z)$-clustering problem, the sensitivity of $x\in X$ is
\begin{align*}
	\sigma_X(x):= \inf\left\{\beta\geq 0 \mid d^z(x,C)\leq \beta \sum_{y\in X}d^z(y,C), \forall C\in [X]^k\right\}.
\end{align*}
The total sensitivity of $X$ is defined by $\sum_{x\in X}\sigma_{X}(x)$.
\end{definition}

In Theorem~\ref{thm:totalsen}, we analyze the total sensitivity for the $(k, z)$-clustering problem. 
As observed in previous work~\cite{langberg2010universal,varadarajan2012sensitivity},
the total sensitivity can actually be bounded by 
a constant (depending on $k$ and $z$) even for a general metric, and a constant factor approximation to the total
sensitivity can be computed efficiently.


\begin{theorem}
\label{thm:totalsen}
Given a metric space $M(X,d)$ for the $(k,z)$-clustering problem, there exists an algorithm that computes an upper bound $\pi_x$ of $2\sigma_X(x)$ for any $x\in X$, such that
\begin{align*}
	\sum_{x\in X}\pi_x =O(2^{O(z\log z)}k), 
\end{align*}
with probability at least $1-\tau$.
Moreover, the algorithm runs in $\poly(n)$ time.
\end{theorem}
The proof of Theorem~\ref{thm:totalsen} can be found in Appendix~\ref{section:proof_sensitivity}.

\subsubsection{Proof of Theorem \ref{thm:coreset}}
The $\eps$-coreset construction is given in
Algorithm~\ref{alg:coreset}.

\begin{algorithm}
\caption{\textrm{Coreset(X)}}
\label{alg:coreset}
For each $x\in X$, compute $\pi_x$ by Theorem \ref{thm:totalsen}. Let $\zeta$ be the integer such that $2^{\zeta-1}\leq n\pi_x < 2^{\zeta}$.
Let $\theta_x=2^{\zeta}$. \\
Pick a non-uniform random sample $S$ of $\Gamma:=O\left( \frac{2^{O(z\log z)}k^3}{\e^{2}}\left(\DDim(M)\cdot \log(z/\e)+ \log k+ \log \log (1/\tau)\right)+\frac{2^{O(z\log z)}k^2 \log (1/\tau)}{\e^{2}}\right)$ points from $X$, where for each sample $u_i\in S$, the probability that $u_i=x$ is $\theta_x/\sum_{x\in X}\theta_x$ $(x\in X)$.  
\\
For each sample $s_i\in S$, define the weight $w_i := \sum_{x\in X}\theta_x/\Gamma \theta_{s_i}$.
\end{algorithm}

Note that our algorithm does not depend on the construction of the smoothed distance function but rather runs on the original metric space $M$.
By Theorem \ref{thm:totalsen}, we can directly prove that the construction time in Theorem \ref{thm:coreset}. Thus, we only need to prove that the collection $S$ is an $\eps$-coreset with probability at least $1-\tau$.
The main tool is the following theorem (Theorem~\ref{thm:fl}) which is a restatement of~\cite[Theorem 4.1]{FL11}.

\begin{theorem}[\cite{FL11}]
\label{thm:fl}
Let $\Psi = \left\{\psi_x\mid x\in X\right\}$ be a set of $n$ functions indexed by $X$ where $\psi_x: [X]^k\rightarrow [0,\infty)$. Let $0<\e<1/8$.
For any function $\psi_x\in \Psi$, let $\theta_x$ be an integer such that
\begin{equation}
\label{ieq:sen}
\theta_x\geq n\cdot \max_{C\in [X]^k}\frac{\psi_x(C)}{\sum_{\psi_y\in \Psi}\psi_y(C)}.
\end{equation}
For each $\psi_x\in \Psi$, let $g_x: [X]^k\rightarrow \R_{\geq 0}$ be defined as $g_x(C)=\psi_x(C)/ \theta_x$.
Let $\calG_x$ consist of $ \theta_x$ copies of $g_x$ and $\calG=\bigcup_{x\in X} \calG_x$. 
Let $\calD$ be a subset of $\calG$ satisfying that $\calD$ is an $(\e n/\sum_{x\in X}\theta_x)$-approximation of the range space $(\calG,\ranges(\calG))$.
%
Then $\calS=\left\{g_x\cdot |\calG|/|\calD| \mid g_x\in \calD\right\}$ satisfies that for every $C\in [X]^k$,
$$
\left| \cost(\Psi,C)-\cost(\calS,C) \right| \leq \e \cdot \cost(\Psi,C),
$$
where $\cost(\calA,C)=\sum_{\psi_x\in \calA} \psi_x(C)$.
\end{theorem}
Our proof strategy is to apply Theorem~\ref{thm:fl} on an auxiliary set of functions.
The auxiliary set is defined based on the random smoothed distance functions as in Corollary~\ref{corollary:weighted_z_ball}, and hence has bounded probabilistic shattering dimension.
\eat{
\begin{corollary}[restatement of Corollary~\ref{corollary:weighted_z_ball}]
	\label{corollary:restate_weighted_z_ball}
	Suppose $M([X]^k, d)$ is a metric space, and $w : [X]^k \rightarrow \mathbb{R}_{\geq 0}$.
	Let $0 < \epsilon \leq \frac{1}{16}$, $0 < \tau < 1$ and $z > 0$ be constant. There exists an $\eps$-smoothed distance function $\delta$ satisfying that for any $x, y \in [X]^k$,
	\begin{align*}
	(1 - O(\eps\cdot z))\cdot \delta^z(x, y)\leq d^z(x, y) \leq (1 + O(\eps \cdot z)) \cdot \delta^z(x, y),
	\end{align*}
	such that for any $\calG \subseteq \{ w(x) \cdot \delta^z(x, \cdot)  \mid x \in [X]^k\}$,
	\begin{align*}
	|\ranges(\calG)| \leq  O\left( \frac{1}{\epsilon} \right)^{O(\DDim(M))}\cdot |\calG|^6.
	\end{align*}
\end{corollary}
}

\noindent\textbf{Auxiliary Functions.}
%
Let $\delta: X\times X\rightarrow \R_{\geq 0}$ be an $(\eps/100z)$-smoothed distance function resultant from Corollary~\ref{corollary:weighted_z_ball}.
%
We construct a function $\psi_x:[X]^k \rightarrow \R_{\geq 0}$ for any $x\in X$ where $\psi_x(C)=\min_{y\in C}\delta^z(x,y)$.
Let $\Psi$ be the collection of all $\psi_x$.
%
%
We have the following lemma which shows that $\theta_x$ (computed in Algorithm~\ref{alg:coreset}) is a valid upper bound for the sensitivity of $\psi_x\in \Psi$, as $\theta_x$ in Theorem \ref{thm:fl}.
The proof can be found in Appendix~\ref{sec:lemma_sen}.

\begin{lemma}[sensitivity of auxiliary functions]
\label{lm:sen}
For any $x\in X$, $\theta_x$ is an integer satisfying that $\theta_x \geq n\cdot \max_{C\in \calO}\frac{\psi_x(C)}{\sum_{\psi_y\in \Psi}\psi_y(C)}$.
Moreover, $\sum_{x\in X} \theta_x = O(2^{O(z \log z)} kn)$.
%
\end{lemma}



\noindent\textbf{Auxiliary Range Space.}
We then construct the following functions based on $\Psi$: let
$g_x: \calO \rightarrow \R_{\geq 0}$ be defined as $g_x(C):=\min_{y\in C}\delta^z(x,y)/\theta_x$ for $x\in X$.
%
Note that $g_x(C) = \psi_x(C)/\theta_x$ satisfies the condition in Theorem \ref{thm:fl}.
Let $\calG$ consist of $\theta_x$ copies of $g_x$ for each $x\in X$ as in Theorem \ref{thm:fl}.
%
%
%
%
%
We show in Lemma~\ref{lm:apprange} that the output of our algorithm can be interpreted as a good approximation for the auxiliary range space $(\calG, \ranges(\calG))$.
This is an important lemma in order to apply Theorem~\ref{thm:fl}.
%
The approach of proving Lemma~\ref{lm:apprange} is to apply Lemma~\ref{lemma:restate_weak_app}.

\eat{
For the sake of presentation, the lemma is restated as follows in Lemma~\ref{lemma:restate_weak_app}.

\begin{lemma}[restatement of Lemma~\ref{lm:balltoapp}]
	\label{lemma:restate_weak_app}
	Let $[X]^k$ be a ground set, and $V$ be a basis set.
	%
	%
	Suppose $g$ maps $V$ to a collection of functions $[X]^k \rightarrow \mathbb{R}_{\geq 0}$, and define $\mathcal{F} := g(V)$.
	Suppose there exists a function $T: \R \times \R \rightarrow \R$ such that for any $H\subseteq V$,
	\begin{align*}
	\Pr_{g}[|\ranges(g(H))| \leq T(|H|, \gamma)] \geq 1 - \gamma.
	\end{align*}
	Let $\calS$ be a collection of $m$ uniformly independent samples from $\mathcal{\Psi}$.
	Then with probability at least $1-\tau$, $\calS$ is an $\alpha$-approximation of the range space $(\cal\Psi, \ranges)$, where the randomness is taken over $\cal\Psi$ and $\calS$, and
	\begin{align*}
	\alpha := \sqrt{\frac{32\left(\log{T(2m, \frac{\tau}{4})} + \log{\frac{8}{\tau}}\right)}{m}}.
	\end{align*}
\end{lemma}
}

\begin{lemma}
\label{lm:apprange}
Assume that the output of Algorithm \ref{alg:coreset} is $S$.
%
%
Then with probability at least $1-\tau$, $\calG_S=\left\{g_x\mid x\in S \right\}$ is an $(\e n/\sum_{x\in X}\theta_x)$-approximation of the range space $(\calG, \ranges(\calG))$.
\end{lemma}

\begin{proof}
Let $V$ be the multiset $\left\{x\in X \mid g_x\in \calG\right\}$, i.e., the index set of $\calG$.
Note that the doubling dimension of $V$ is still $\DDim(M)$.
By the sampling process of $S$ in Algorithm~\ref{alg:coreset}, $\calS$ can be viewed as a uniform sample from $V$ and $\calG_S$ can be viewed as a uniform sample from $\calG$.
%
%
We will apply Lemma~\ref{lemma:restate_weak_app}. The first step is to show that for $T:\mathbb{N} \times \mathbb{R}_{\geq 0}$ such that
\begin{align*}
T(m, \gamma):= O\left(\frac{z}{\eps}\right)^{ O(k\cdot\DDim(M))}\cdot \log^k \frac{m}{\gamma}\cdot m^{6k},
\end{align*}
$\calG$ satisfies
for any $H\subseteq V$ and $\gamma> 0$,
\[
\Pr[|\ranges(\calG_H)| \leq T(|H|,\gamma)] \geq 1 - \gamma.
\]
Fix a subset $H\subseteq V$. We will apply Corollary~\ref{corollary:weighted_z_ball} to bound $|\ranges(\calG_H)|$. 
%
%
However, the ground set of $\calG$ is $[X]^k$, while the ground set considered in Corollary~\ref{corollary:weighted_z_ball} is $X$. Hence, we will define another range space with ground set $X$ so that Corollary~\ref{corollary:weighted_z_ball} may be used, and then relate $|\ranges(\calG_H)|$ to that range space.

Define another collection of functions $\calF$ with index set $X$ and ground set $X$ as follows. For any $x\in X$, let
$f_x: X \rightarrow \R_{\geq 0}$ be defined as $f_x(y)=\delta^z(x,y)/\theta_x$.
Let $\calF$ consist of all $f_x$ for $x\in X$.
%
%
By the construction of $\theta_x$, we know that $1/\theta_x$ is a gap-$2$ weight function.
Moreover, the error parameter $\eps/100z <1/100z$.
Hence, we can apply Corollary~\ref{corollary:weighted_z_ball} to the range space $(\calF_H,\ranges(\calF_H))$.
By Corollary~\ref{corollary:weighted_z_ball}, we have for $\gamma>0$
	\begin{align}
		\Pr\left[|\ranges(\calF_H)|
		\leq O\left(\frac{z}{\eps}\right)^{O(\DDim(M))} \cdot \log \frac{|H|}{\gamma}\cdot |H|^6\label{eqn:use_cor_z}\right]\geq 1-\gamma.
	\end{align}

%
Next, we prove the following simple 
Claim~\ref{claim:min_equal} which relates the range space $(\calG_H, \ranges(\calG_H))$ with ground set $[X]^k$ to the range space $(\calF_H, \ranges(\calF_H))$ with ground set $X$.
%
%
A very similar claim is also shown in~\cite[Lemma 6.5]{FL11}. We postpone its proof to Appendix~\ref{section:min_equal}.

\begin{claim}
	\label{claim:min_equal}	$|\ranges(\calG_H)| \leq |\ranges(\calF_H)|^k$.
\end{claim}
\ignore{
\begin{proof}
	By the definition of $V$, each range $\range(\calG_H, C,r)$ of $(\calG,\ranges(\calG))$ corresponds to a unique point set $\left\{y\in X: g_y\in \range(\calG_H, C,r)\right\}$.
	Then it suffices to show that for any $C\in [X]^k$ and $r\geq 0$,
	$$\left\{y\in X: g_y\in \range(\calG_H, C,r)\right\}=\cup_{x\in C} \left\{y\in X: f_y\in \range(\calF_H, x,r)\right\}.$$
	For any $C\in \calO$ and $r\geq 0$, if $g_y\in \range(\calG_H, C,r)$, we have
	\[
	g_y(C) = \min_{x\in C} \delta^z(x,y)/\theta_y \leq r.
	\]
	Let $x^*=\arg \min_{x\in C} \delta^z(x,y)$.
	We have $f_y(x^*)=\delta^z(x^*,y)/\theta_y\leq r$ which implies that $f_y\in \range(\calF_H, x^*,r)\subseteq \cup_{x\in C}\range(\calF_H, x,r)$.
	Therefore, we have $$\left\{y\in X: g_y\in \range(\calG_H, C,r)\right\}\subseteq \cup_{x\in C} \left\{y\in X: f_y\in \range(\calF_H, x,r)\right\}.$$
	It remains to prove $\cup_{x\in C} \left\{y\in X: f_y\in \range(\calF_H, x,r)\right\}\subseteq \left\{y\in X: g_y\in \range(\calG_H, C,r)\right\}$.
	If $f_y\in  \cup_{x\in C}\range(\calF_H, x,r)$, there must exist some $x^*\in C$ such that $f_y\in \range(\calF_H, x^*,r)$.
	It implies that
	\[
	g_y(C)=\min_{x\in C}\delta^z(x,y)/\theta_y \leq \delta^z(x^*,y)/\theta_y = f_y(x^*)\leq r.
	\]
	Hence $g_y\in \range(\calG_H, C,r)$.
	Thus, we have $$\cup_{x\in C} \left\{y\in X: f_y\in \range(\calF_H, x,r)\right\}\subseteq \left\{y\in X: g_y\in \range(\calG_H, C,r)\right\},$$ and this completes the proof of Claim~\ref{claim:min_equal}.
\end{proof}
}
By Claim~\ref{claim:min_equal}, with probability at least $1-\gamma$,
\begin{equation}
\label{eq:number}
	|\ranges(\calG_H)|
	\leq |\ranges(\calF_H)|^k
	\leq  O\left(\frac{z}{\eps}\right)^{ O(k\cdot\DDim(M))}\cdot \log^k \frac{|H|}{\gamma}\cdot |H|^{6k}  = T(|H|,\gamma).
\end{equation}
Hence, we have for any $H\subseteq V$ and $\gamma> 0$,
\[
	\Pr[|\ranges(\calG_H)| \leq T(|H|,\gamma)] \geq 1 - \gamma.
\]
Now we are ready to apply Lemma~\ref{lemma:restate_weak_app}.
Note that $\sum_{x\in X}\theta_x=O( 2^{O(z\log z)} k n)$ by Lemma~\ref{lm:sen}.
Plugging in the values of $\Gamma$ and $T(2\Gamma, \tau/4)$ to Lemma~\ref{lemma:restate_weak_app}, we can verify that $\calG_S$ is an $(\e n/\sum_{x\in X}\theta_x)$-approximation of the range space $(\calG, \ranges(\calG))$ with probability at least $1-\tau$.
This completes the proof of Lemma~\ref{lm:apprange}.
\end{proof}

Now, we are ready to prove Theorem \ref{thm:coreset}.

\begin{proof}[Proof of Theorem \ref{thm:coreset}]
The running time follows from Theorem \ref{thm:totalsen} and the sampling process. It remains to prove the correctness.
Denote $\mathcal{E}$ as the event that $\calD := \calG_S$ is an $(\e n/\sum_{x\in X}\theta_x)$-approximation of the range space $(\calG, \ranges(\calG))$ induced by $\delta$.
By Lemma \ref{lm:apprange}, we have $\Pr[\mathcal{E}]\geq 1-\tau$.
%
In the following, we assume that $\mathcal{E}$ happens. 
Let $\calS=\left\{g_x\cdot |\calG|/|\calD| \mid g_x\in \calD\right\}$.
Combining with Theorem \ref{thm:fl} and Lemma \ref{lm:sen}, we conclude that for every $C\in \calO$,
\begin{equation}
\label{eq:thm1}
\left| \cost(\calG,C)-\cost(\calS,C) \right| \leq \eps\cdot \cost(\calG,C),
\end{equation}
where $\cost(\calA,C)=\sum_{g_x\in \calA} g_x(C)$ for a given set $\calA\subseteq \calG$.

On the other hand, since $\delta$ is an $O(\eps/z)$-smoothed distance function, we have
\begin{equation}
\label{eq:thm2}
\cost(\calG,C) = \sum_{g_x\in \calG} g_x(C) = \sum_{x\in X} \theta_x \cdot \frac{\min_{y\in C}\delta^z(x,y)}{\theta_x} \in (1\pm \eps) \sum_{x\in X} d^z(x,C)= (1\pm \eps)\kdist_z(X,C),
\end{equation}
and by the fact that $\calS=\left\{g_x\cdot |G|/|\calD| \mid x\in S\right\}$ and the definition of $w_i$ in Algorithm~\ref{alg:coreset},
\begin{equation}
\label{eq:thm3}
\cost(\calS,C) = \sum_{s_i\in S} \frac{|\calG|}{|\calD|} \cdot g_{s_i}(C) = \sum_{s_i\in S} \frac{\sum_{x\in X}\theta_x}{\Gamma} \cdot \frac{\min_{y\in C}\delta^z(s_i,y)}{\theta_{s_i}} \in (1\pm \eps) \sum_{s_i\in S} w_i \cdot d^z(s_i, C).
\end{equation}
Combining Inequalities \eqref{eq:thm1}-\eqref{eq:thm3}, we have
\[
\sum_{s_i\in S} w_i\cdot d^z(s_i,C) \in (1\pm \e) \cost(\calS,C) \in (1\pm \e)^2 \cost(\calG,C) \in (1\pm \e)^3 \kdist_z(X,C) = (1\pm O(\e)) \kdist_z(X,C).
\]
Therefore, Algorithm \ref{alg:coreset} outputs an $O(\e)$-coreset $S$ of $X$ with probability at least $1-\tau$.
\end{proof}

\eat{

Since $G$ consists of $\theta_x$ copies of $g_x$ for each $x\in X$, we have the following
\begin{equation}
\label{eq:2}
|\range(C,r)| = \sum_{x\in R_\delta(C,r)} \theta_x, |G| = \sum_{x\in X}\theta_x.
\end{equation}
By the definition of $\calD$, we have
\begin{equation}
\label{eq:3}
|\calD\cap \range(C,r)| = |S\cap R_\delta(C,r)|, |G| = \sum_{x\in X}\theta_x,
\end{equation}
Combining with \eqref{eq:1}-\eqref{eq:3}, we complete the proof.
}

\subsection{Robust Coreset and Property Testing}
\label{section:robust_coreset}
In this section, we consider robust coresets for the $(k,z)$-clustering problem with outliers (see Definition~\ref{def:robust_coreset}).
We generalize and improve the prior result~\cite{FL11} for Euclidean space, and prove the existence of robust coresets with smaller size in doubling metrics. The following is the main theorem of this section.
%

%
%
%
%
%


\begin{theorem}
	\label{thm:robust_coreset}
	Let $M(X,d)$ be a doubling metric space (a $d$-dimensional Euclidean space resp.).
	Suppose $S$ is a uniform independent sample of $\Gamma$ ($\Gamma'$ resp.) points from $X$, where
	$$
	\Gamma := O\bigg(\frac{k}{\alpha^2}(\mathrm{ddim}(M)\cdot \log (z/\eps) +\log k+\log \log (1/\tau))+\frac{\log (1/\tau)}{\alpha^2}\bigg)
	$$
    and
    $$
    \Gamma' := O\bigg(\frac{1}{\alpha^2}(kd \log k+\log (1/\tau))\bigg).
    $$
 Then with probability at least $1-\tau$, $\calS$ is an $(\alpha,\eps)$-robust coreset ($(\alpha,0)$-robust coreset resp.)
 for the $(k,z)$-clustering problem with outliers.
\end{theorem}

\subsubsection{Proof of Theorem~\ref{thm:robust_coreset}}

Similar to Section~\ref{sec:coreset}, our main idea is to construct an auxiliary range space with bounded probabilistic dimension, and then obtain an $\eps$-approximation.
We show that an $\eps$-approximation for the range space already induces a robust coreset.
Again, we consider the functional representation of the problem as follows:

\begin{definition}[Robust Coreset for a Set of Functions]
	\label{rcoreset}
Assume $0<\alpha, \eps <\frac{1}{4}$. Let $\calG$ be a finite set of functions $[X]^k \rightarrow \R_{\geq 0}$.
For any $0< \gamma< 1$, $C\in [X]^k$ and $\calS\subseteq \calG$, let
\[
\calS^{-\gamma}(C):= \min_{\calS'\subseteq \calS: |\calS'|= \lceil(1-\gamma)|\calS|\rceil} \sum_{g\in \calS'} g(C),
\]
which is the sum of the smallest $\lceil (1-\gamma) |\calS| \rceil$ values $g(C)$.
%
Then a subset $\calS\subseteq \calG$ is called an $(\alpha,\eps)$-robust coreset of $\calG$ if for any $\alpha< \gamma< 1-\alpha$ and $C\in [X]^k$,

\begin{eqnarray} \label{rcdef}
(1-\eps)\cdot \frac{\calG^{-(\gamma+\alpha)}(C)}{|\calG|}\leq \frac{\calS^{-\gamma}(C)}{|\calS|} \leq (1+\eps)\cdot \frac{\calG^{-(\gamma-\alpha)}(C)}{|\calG|}.
\end{eqnarray}

\end{definition}


\begin{remark}
	\label{remark:robust}
To reduce the problem of constructing a robust coreset for clustering to the problem for functions, for $x\in X$, let $g_x(\cdot)$ be a function from $[X]^k$ to $\mathbb{R}_{\geq 0}$ such that $g_x(C)=d^z(x,C)$. 
Let $\calG:=\left\{ g_x\mid x\in X \right\}$.
\jian{specify $g_x$. the argument is $C$ or a point $y$}\xuan{specify it.}\shaofeng{The def of $\calG$ looks strange. I suggest to define $g_x$ before $\calG$.}\xuan{done}
%
%

We note that our definition is slightly different from that in~\cite[Definition 8.1]{FL11} \footnote{In fact, our definition is more general. It is unclear whether their result applies to our definition.}.
In particular, in Euclidean spaces, one can check that an $(\eps\gamma/4,0)$-robust coreset is a $(\gamma,\eps)$-coreset in \cite[Definition 8.1]{FL11}.
\end{remark}

Next, we prove the following simple connection between $\alpha$-approximation of $(\calG, \ranges(\calG))$ and robust coreset of $\calG$ in Lemma~\ref{rctech}.
This lemma improves \cite[Theorem 8.3]{FL11} in which they show that an $(\eps^2\gamma/63)$-approximation is a $(\gamma,\eps)$-coreset\shaofeng{53? 63?}\xuan{63!}\footnote{Consider the $(\gamma,\eps)$-coreset in \cite[Definition 8.1]{FL11}. Since an $(\eps\gamma/4,0)$-robust coreset is a $(\gamma,\eps)$-coreset, our Lemma \ref{rctech} implies that an $(\eps\gamma/8)$-approximation is a $(\gamma,\eps)$-coreset.}.\jian{state Theorem 8.3 FL11 in our notation. how do we improve on theirs}\xuan{Add a footnote.}
First we need the following simple formulas. For $x\in \mathbb{R}$, let $(x)_+$ denote $\max\{0, x\}$.

\begin{claim} For any $\gamma\in (\alpha,1-\alpha)$ and $C\in [X]^k$, the following equations hold:

\begin{eqnarray} \label{rceq1}
 \frac{\calS^{-\gamma}(C)}{|\calS|}=\int_{0}^{\infty} \bigg(\frac{\lceil (1-\gamma)|\calS| \rceil}{|\calS|}-\frac{|\calS \cap \range(\calG, C, r)|}{|\calS|}\bigg)_+ dr,
 \end{eqnarray}
\begin{eqnarray} \label{rceq2}
\frac{\calG^{-(\gamma+\alpha)}(C)}{|\calG|}=\int_{0}^{\infty} \bigg(\frac{\lceil (1-\gamma- \alpha)|\calG| \rceil}{|\calG|}-\frac{|\range(\calG, C, r)|}{|\calG|}\bigg)_+ dr,
\end{eqnarray}
\begin{eqnarray} \label{rceq3}
\frac{\calG^{-(\gamma-\alpha)}(C)}{|\calG|}=\int_{0}^{\infty} \bigg(\frac{\lceil (1-\gamma+ \alpha)|\calG| \rceil}{|\calG|}-\frac{|\range(\calG, C, r)|}{|\calG|}\bigg)_+ dr.
\end{eqnarray}
\end{claim}


\begin{proof}
	We only prove the first one.
	The other two Equations~\eqref{rceq2} and~\eqref{rceq3}
	can be proved in the same manner.
	Let $\calD$ be the collection of functions $g\in \calS$ with the smallest $\lceil (1-\gamma) |\calS| \rceil$ values $g(C)$.
	Using integration, we know that
	$$
	\frac{\calS^{-\gamma}(C)}{|\calS|}=\int_{0}^{\infty} \frac{\big{|}\{ g(C)> r\mid g\in \calD\}\big{|}}{|\calS|}  dr.
	$$
	By definition, we have
	\begin{align*}
	\frac{\big{|}\{ g(C)> r\mid g\in \calD\}\big{|}}{|\calS|}
	&=\frac{|\calD\setminus \range(\calG, C, r)|}{|\calS|}= \frac{\big(|\calD|-|\calS \cap \range(\calG, C, r)|\big)_+}{|\calS|}\\
	&= \bigg(\frac{\lceil (1-\gamma)|\calS| \rceil}{|\calS|}-\frac{|\calS \cap \range(\calG, C, r)|}{|\calS|}\bigg)_+,
	\end{align*}
	which proves Equation~\eqref{rceq1}.
\end{proof}

\begin{lemma} \label{rctech}
	If $\calS$ is an $\frac{\alpha}{2}$-approximation of $(\calG,\ranges(\calG))$ such that $|\calS|,|\calG|\geq 2/\alpha$, then $\calS$ is an $(\alpha,0)$-robust coreset of $\calG$.
\end{lemma}

\begin{proof}
	Let $\calS\subseteq \calG$ be an $\alpha$-approximation of $(\calG,\ranges(\calG))$.
	Now, we show that $\calS$ is also an $(\alpha,0)$-robust coreset of $\calG$.	
	%
	Since $\calS$ is an $\frac{\alpha}{2}$-approximation of $\calG$, for any $C\in [X]^k$ and $r\geq 0$,
	\begin{align}
	\label{eq:approximation}
	\left|\frac{|\range(\calG, C, r)|}{|\calG|}-\frac{|\calS\cap \range(\calG, C, r)|}{|\calS|}\right|
	\leq \frac{\alpha}{2}.
	\end{align}
	So we have that
	\begin{align*}
			\bigg(\frac{\lceil (1-\gamma- \alpha)|\calG| \rceil}{|\calG|}-\frac{|\range(\calG, C, r)|}{|\calG|}\bigg)_+
			\leq & \,\,  \bigg(\frac{\lceil (1-\gamma- \alpha)|\calG| \rceil}{|\calG|}+\frac{\alpha}{2}-\frac{|\calS \cap \range(\calG, C, r)|}{|\calS|}\bigg)_+  \\
			\leq & \,\, \bigg(\frac{(1-\gamma- \alpha)|\calG| +1}{|\calG|}+\frac{\alpha}{2}-\frac{|\calS \cap \range(\calG, C, r)|}{|\calS|}\bigg)_+ & \\
			=&\bigg(\frac{(1-\gamma)|\calS|}{|\calS|}-\frac{\alpha}{2}+\frac{1}{|\calG|}-\frac{|\calS \cap \range(\calG, C, r)|}{|\calS|}\bigg)_+ & \\
			\leq  & \,\, \bigg(\frac{\lceil(1-\gamma)|\calS|\rceil}{|\calS|}-\frac{|\calS \cap \range(\calG, C, r)|}{|\calS|}\bigg)_+
	\end{align*}
	The first inequality holds due to
	\text{Inequality~\eqref{eq:approximation}}
	and the last follows because
	$|\calG|\geq 2/\alpha$.
	
	Together with (\ref{rceq1}) and (\ref{rceq2}), we have that
	\begin{eqnarray*}
		\frac{\calG^{-(\gamma+\alpha)}(C)}{|\calG|}\leq \frac{\calS^{-\gamma}(C)}{|\calS|}.
	\end{eqnarray*}
	Similarly, by (\ref{rceq1}), (\ref{rceq3}) and~\eqref{eq:approximation}, we can also show that
	\begin{eqnarray*}
		\frac{\calS^{-\gamma}(C)}{|\calS|} \leq  \frac{\calG^{-(\gamma-\alpha)}(C)}{|\calG|},
	\end{eqnarray*}
	which completes the proof.
\end{proof}

In the $d$-dimensional Euclidean space, one can utilize a $(\gamma \eps/8)$-approximation to construct a $(\gamma,\eps)$-coreset of \cite[Definition 8.1]{FL11}. Using the improved Lemma \ref{rctech}, we can improve the robust coreset size in \cite[Definition 8.1]{FL11} from $O(kd \log k\cdot \gamma^{-2}\eps^{-4})$
\footnote{The size stated in~\cite{FL11} is $O(kd \gamma^{-2}\eps^{-4})$. We defer interesting readers to \cite[Section 5]{bachem2017scalable} to see why an additional $\log k$ factor is required.}
to $O(kd \log k\cdot \gamma^{-2}\eps^{-2})$.

With the help of Lemma~\ref{rctech}, Theorem~\ref{thm:robust_coreset} follows by a similar argument as in the coreset construction (Section~\ref{sec:coreset}).
We present the details in Appendix~\ref{section:proof_robust}.

\subsubsection{Application to Property Testing} \label{PT}

In this section, we show some applications of robust coreset to property testing.
We start with the following definition that captures the notion of bi-criteria algorithms.

\begin{definition}
	\label{def:bicriteria}
Let $M(X,d)$ be a metric space.
Let $\lambda \geq 1$, $0<\alpha<1/4$ and $\alpha< \gamma <1-\alpha$.
We say $A$ is a $(\lambda, \gamma, \alpha)$-approximation algorithm for the $(k,z)$-clustering problem with outliers, if $A$ returns a number $\Lambda$ such that $\min_{C\in [X]^k} \kdist_z^{-(\gamma+\alpha)}(X, C)\leq \Lambda\leq \lambda\cdot\min_{C\in [X]^k}  \kdist_z^{-(\gamma-\alpha)}(X, C)$.
\end{definition}

Now, we state our result for testing $(k,z)$-clustering.
The testing problem was first proposed by 
Alon et al.~\cite{alon2003testing}
for the $k$-center problem.

\begin{theorem}[Testing of $(k,z)$-clustering]
Let $M(X,d)$ be a doubling metric space ($d$-dimensional Euclidean space resp.).
Let $\lambda \geq 1$, $0<\alpha<1/4$ and $\alpha< \gamma <1-\alpha$.
Suppose there is a $(\lambda, \gamma, \alpha)$-approximation algorithm for the $(k,z)$-clustering problem with outliers, which runs in time $T(|X|,\lambda,\gamma, \alpha)$.
Then for any $\Delta>0$ and $0<\eps< 1/4$, there is an algorithm satisfying the following:
\begin{enumerate}
\item If $\min_{C\in [X]^k} \kdist_z^{-(\gamma-\alpha)}(X, C)\leq \Delta$, it accepts with probability $1-\tau$;
\item If $\min_{C\in [X]^k} \kdist_z^{-(\gamma+\alpha)}(X, C)\geq \lambda(1+\eps)\cdot \Delta$, it rejects with probability $1-\tau$;
\end{enumerate}
and the running time is $T(\Gamma,\gamma,\lambda, \frac{\alpha}{2})+\Gamma^2$,
where
$$
\Gamma:=O\bigg(\frac{k}{\alpha^2}(\mathrm{ddim}(M)\cdot \log (z/\eps) +\log k+\log \log (1/\tau))+\frac{\log (1/\tau)}{\alpha^2}\bigg)
$$
for doubling metrics and
$$
\Gamma:= O\bigg(\frac{1}{\alpha^2}(kd \log k+\log (1/\tau))\bigg)
$$
for $d$-dimensional Euclidean space.
\end{theorem}

\begin{proof}
Consider the following algorithm:
\begin{enumerate}
\item Take a uniformly independent sample $S$ of size $\Gamma$ from $X$.
\item Run the $(\lambda, \gamma, \frac{\alpha}{2})$-approximation algorithm on $S$.
Suppose the output is $\Gamma$.
\item Accept if $\Gamma\leq \frac{(1+\eps/4)\lambda \Gamma}{|X|}\cdot \Delta$, and reject otherwise.
\end{enumerate}

By Theorem \ref{thm:robust_coreset}, with probability at least $1-\tau$, $S$ is an $(\frac{\alpha}{2},\frac{\eps}{4})$-robust coreset for $X$ \footnote{Recall that in Euclidean space, $S$ is actually an $(\frac{\alpha}{2},0)$-robust coreset, but the weaker guarantee is sufficient here.}.
In the following, we condition on the event that $S$ is an $(\frac{\alpha}{2},\frac{\eps}{4})$-robust coreset for $X$.
Hence, for any $C\in [X]^k$ and $\alpha<\gamma<1-\alpha$, we have
\begin{eqnarray}
\label{eq:robust_coreset}
(1-\eps/4)\cdot \frac{\kdist_z^{-(\gamma+\frac{\alpha}{2})}(X, C)}{|X|} \leq \frac{\kdist_z^{-\gamma}(S, C)}{|S|} \leq (1+\eps/4)\cdot \frac{\kdist_z^{-(\gamma-\frac{\alpha}{2})}(X, C)}{|X|}.
\end{eqnarray}

Recall that $\Lambda$ is the output of the $(\lambda, \frac{\alpha}{2})$-approximation algorithm.
Then by Definition~\ref{def:bicriteria} and Inequality~\eqref{eq:robust_coreset}, we have
\begin{align}
\label{eq:upper}
\Lambda< \lambda\cdot \min_{C\in [X]^k} \kdist_z^{-(\gamma-\frac{\alpha}{2})}(S, C)\leq \frac{(1+\eps/4)\lambda \Gamma}{|X|}\cdot \min_{C\in [X]^k} \kdist_z^{-(\gamma-\alpha)}(X, C),
\end{align}
and
\begin{align}
\label{eq:lower}
\Lambda \geq \min_{C\in [X]^k} \kdist_z^{-\gamma}(S, C)\geq \frac{(1-\eps/4)\Gamma}{|X|}\cdot \min_{C\in [X]^k} \kdist_z^{-(\gamma+\alpha)}(X, C)
\end{align}

If $\min_{C\in [X]^k} \kdist_z^{-(\gamma-\alpha)}(X, C)\leq \Delta$, we have $\Gamma \leq \frac{(1+\eps/4)\lambda \Gamma}{|X|}\cdot \Delta$ by Inequality~\eqref{eq:upper}.
In this case, our algorithm accepts.
On the other hand, if $\min_{C\in [X]^k} \kdist_z^{-(\gamma+\alpha)}(X, C)\geq \lambda(1+\eps) \cdot \Delta$, we have
\[
\Gamma \stackrel{\text{Ineq. \eqref{eq:lower}}}\geq \frac{(1-\eps/4)\Gamma}{|X|}\cdot \lambda(1+\eps) \cdot \Delta> \frac{(1+\eps/4)\lambda \Gamma}{|X|}\cdot \Delta.
\]
In this case, our algorithm rejects.
It completes the proof.
\end{proof}

\begin{remark}
	The $(\lambda, \gamma, \alpha)$-approximation algorithm for the $(k,z)$-clustering problem with outliers is used as a subroutine in our testing algorithm.  
	If we use exhaustive search, we obtain a $(1,\gamma,0)$-approximation algorithm with running time exponential in $|S|$ for $(k,z)$-clustering with outliers. 
	If we use the approximation algorithm for $k$-median
	by Charikar et al. \cite{charikar2001algorithms}, we have a $(4(1+\lambda^{-1}),\gamma,\lambda\gamma)$-approximation algorithm with running time polynomial in $|S|$ for the $(k,1)$-clustering problem with outliers.
\end{remark}

\subsection{Centroid Set and Fast Local Search Algorithm}
\label{sec:centroid}
In this section, we discuss the construction of centroid sets in doubling metrics.
By this construction, we can improve the running time of the PTAS for $(k,z)$-clustering in~\cite{DBLP:conf/focs/FriggstadRS16}.
We give the definition of centroid set as follows.

\begin{definition}
	\label{definition:centroid_set}
	Let $k \geq 1$ be an integer and $\eps, z > 0$.
	Let $M(X, d)$ be a metric space. Given a weighted point set $S\subseteq X$ with weight function $w: S\rightarrow \R_{\geq 0}$, an $(\eps, k, z)$-centroid set is a subset $H$ of points such that
	\begin{enumerate}
		\item  $S\subseteq H \subseteq X$.
		\item  there exists a $k$-point set $C \subseteq H$ such that, 
		\begin{align*}
		\sum_{x \in S}{w(x) \cdot d^z(x, C)} \leq (1 + \eps)\cdot \min_{C'\in [X]^k} \sum_{x \in S}{w(x) \cdot d^z(x, C')}.
		\end{align*}
	\end{enumerate}
\end{definition}

In other words, $H$ extends $S$ in the sense that a $(1+\eps)$-approximate solution to the weighted $(k, z)$-clustering instance $S$ is contained in $H$.
Then if $S$ is an $\eps$-coreset of $X$, we have a natural corollary that the centroid set $H$ must contain a $(1+2\eps)$-approximate solution for the $(k,z)$-clustering problem on $X$. 

The idea of centroid set was first introduced in~\cite{DBLP:journals/dcg/Matousek00},
for obtaining a constant approximation to the geometric $k$-means problem in bounded dimensional Euclidean spaces.
%
%
However, their construction cannot be readily applied to our setting, because it relies heavily on certain properties 
that are only available in Euclidean spaces and the $k$-means objective.

We present the first efficient construction of small sized centroid set in doubling metrics.
Although the size of our centroid set is slightly larger than that in Euclidean case (\cite{DBLP:journals/dcg/Matousek00}), it is still independent of $|X|$.

\noindent\textbf{Construction of Centroid Set.}
We start with an overview of the idea.
Assume the optimal clustering for $S$ is $C^\star = \left\{o_1,\ldots, o_k \right\}$.
Let $P_1, P_2, \ldots, P_k$ be the clustering of $S$ with respect to $C^\star$, where $o_i$ is the center of $P_i$.
An intuitive idea is that for each $o_i$, the centroid set $H$ should contain some point $x_i\in X$ such that $x_i$ and $o_i$ are ``close'' enough. 
For the sake of presentation, assume we are to approximate $o_1$ and $z = 1$.
Let $y := \arg\min_{y' \in P_1}{d(o_1, y')}$ be the closest point in $P_1$ to $o_1$.
Then a point $o\in X$ such that $d(o_1,o)\leq \eps\cdot d(o_1,y)$ is a good estimation of $o_1$, i.e., $\sum_{x\in P_1} w(x)\cdot d (x,o)\leq (1+\eps)\cdot \sum_{x\in P_1} w(x)\cdot d (x,o_1)$.
%
%
Hence, if we know $y$ and the distance $d(o_1, y)$, we can include an $\eps \cdot d(o_1, y)$-net around $y$ in $H$, such that $o_1$ must be covered by at least one net point within distance at most $\eps \cdot d(o_1, y)$.
Since $o_i$'s are unknown in the first place, we may need add a large enough net at every distance scale. Since there may be $\log{\diam(X)}$ distance scales, a naive implementation would have $|H|$ depending on $|X|$.
To resolve this issue, we consider the \emph{invariant intervals}.

\noindent\textbf{Invariant Intervals.}
As in Section~\ref{section:prelim}, rescale the metric space such that the minimum intra point distance is $1$, and 
let $L := \lceil \log{\diam(X)} \rceil$.
Let $\{ N_i \mid i \leq L \}$ denote a hierarchical net for the metric space.
Define $T$ as a simple net tree with respect to $\{ N_i \}_{i}$ (as in Section~\ref{section:unweighted}).
We partition $\{0, 1, 2, \ldots, L\}$ into invariant intervals (defined as follows.).

\begin{definition} (invariant interval)
	\label{def:invariant_interval}
An invariant interval is a maximal interval $[a, b)$, such that
for all integer $i \in [a, b)$ and $u\in S$, $\Des(\Par^{(i)} (u)) \cap S = \Des(\Par^{(i+1)} (u)) \cap S$.
\end{definition}
Let $\mathcal{I} := \{I_i\}_i$ be the collection of the invariant intervals listed in the increasing order.
For each invariant interval $[a,b)$, we consider two small sub-intervals $[a,a+3]$ and $[b - 6 - \lceil \log{\frac{1}{\eps}}\rceil, b-1]$.
For $j$ in these two sub-intervals, we enumerate all net points $u\in N_j$ with $\Des(u^{(j)}) \cap S \neq \emptyset$.
We include in $H$ all net points of height $\lfloor \log(\eps \cdot 2^j) \rfloor$ inside the ball $B(u, 5\cdot 2^j)$. By the packing property (Fact~\ref{fact:packing}), there are only $O(\frac{1}{\eps})^{\DDim(M)}$ such net points for each $u$ and $j$.
%

We argue that $H$ is an $(O(z\cdot \eps),k,z)$-centroid set by the above construction.
Again consider $o_1\in C^\star$ and $y = \arg\min_{y' \in P_1}{d(o_1, y')}$.
The key observation is that if $2^{a+4}\leq d(o_1,y)< 2^{b - 6 - \lceil \log{\frac{1}{\eps}}\rceil}$, then $y$ itself is already a good estimation of $o_1$.
Observe that $y$ is included in $H$, because $y\in P_1\subseteq S$ and we always include $S$ in $H$.
This is the reason why we do not need to enumerate $j$ among the sub-interval $[a+4, b-7- \lceil \log{\frac{1}{\eps}}\rceil]$.
For the remaining case that $2^a\leq d(o_1,y)< 2^{a+4}$ or $2^{b - 6 - \lceil \log{\frac{1}{\eps}}\rceil}\leq d(o_1,y)< 2^{b}$, by construction, there must exist a net point $o\in H$ such that $d(o,o_1)\leq \eps\cdot d(o_1,y)$. 
Then $o$ is already a good estimation of $o_1$.
%

\begin{algorithm}
	\caption{\textrm{Centroid-Set(S, w)}}
	\label{alg:centroid_set}
	Initially, let $H := S$. \\
	Construct a hierarchical net $\left\{N_i \mid i \leq L\right\}$ on $X$ and a simple net tree $T$ with respect to it. \\
	Partition $\{0, 1, \ldots, L\}$ into invariant intervals $\mathcal{I} = \{I_i\}_i$. \\
	For each $i\in [|\mathcal{I}|]$, assume $I_i := [a, b)$.
	 For all $j$ such that either $a\leq j \leq \max\{a + 3, b\}$ or $\min\{a, b - 6 - \lceil \log{\frac{1}{\eps}}\rceil\}\leq j\leq b-1$,
	and all $u \in N_j$ such that $\Des(u^{(j)}) \cap S \neq \emptyset$, do the following.
	\begin{itemize}
		\item[1)] Define $j'$ to be the integer such that $2^{j'} \leq \eps \cdot 2^j < 2^{j'+1}$.
		\item[2)] 
		Define $C^{(u)}_j := N_{j'} \cap \Bd(u, 5 \cdot 2^j)$.
		\item[3)] 
		Update $H := H \cup C^{(u)}_j$.
	\end{itemize}
\end{algorithm}

The main result is stated in the following theorem. 

	\begin{theorem}[centroid set]
		\label{theorem:app_local_search}
		Let $k \geq 1$ be an integer, $z > 0$ and $0<\eps<\frac{1}{z}$.
		Given a ground set $X$ and a weighted point set $S\subseteq X$ with weight function $w: S\rightarrow \R_{\geq 0}$, there is an algorithm running in $\poly(|X|)$ time, that finds an $(O(z\cdot \eps), k, z)$-centroid set of size at most $ O(\frac{1}{\eps})^{O(\DDim(M))}\cdot |S|^2$.
	\end{theorem}

	\begin{proof}
		The construction of the centroid set is provided in	Algorithm~\ref{alg:centroid_set}.
		We need to prove the algorithm satisfies the desired properties.
		Obviously Algorithm~\ref{alg:centroid_set} runs in polynomial time.
		So, we only need to analyze the size and the correctness.

	\noindent\textbf{Size Analysis.}
	We first bound the number of invariance intervals.
	For each $i\leq L$, define $\Delta_i: = \left\{\Des(\Par^{(i)} (u)) \cap S\mid u\in S\right\}$. 
	By definition, we have $1\leq |\Delta_i|\leq |S|$.
	Note that for any $p,q\in S$, either $\Des(\Par^{(i)} (p))=\Des(\Par^{(i)} (q))$ or $\Des(\Par^{(i)} (p))\cap \Des(\Par^{(i)} (q))=\emptyset$.
	Then $\Delta_i$ must form a partition of $S$.
	On the other hand, for any $u\in S$, $\Des(\Par^{(i+1)} (u)) \cap S$ is a union of some non-empty point sets $\Des(\Par^{(i)} (v)) \cap S$.
	It implies that $|\Delta_{i+1}|\leq |\Delta_i|$.
	For any invariant interval $[a,b)$, we have $\Delta_a = \ldots = \Delta_{b-1}\neq \Delta_b$ by definition.
	Therefore, $|\Delta_b|\leq |\Delta_{b-1}|-1$, i.e., $|\Delta_i|$ must decrease at least 1 when $i$ goes through an invariant interval.
	Thus, there are at most $|S|$ invariant intervals.

	By the definition of $C_{j}^{(u)}$, we have $|C_j^{(u)}| \leq O(\frac{1}{\eps})^{\DDim(M)}$ because of the packing property (Fact~\ref{fact:packing}).
	Moreover, for each invariant interval $[a,b)$, we enumerate at most $O(\log{\frac{1}{\eps}})$ integers $j\in [a,b)$. 
	For each height $j$, there are at most $|S|$ points $u\in N_j$ such that $\Des(u^{(j)}) \cap S \neq \emptyset$. Therefore, we add at most $ \log\frac{1}{\eps}\cdot |S| \cdot O(\frac{1}{\eps})^{O(\DDim(M))}$ net points to $H$ inside $[a,b)$.
	Since there are at most $|S|$ invariant intervals, we have $|H| \leq O(\frac{1}{\eps})^{O(\DDim(M))} \cdot |S|^2$, as desired.
	
\eat{	\noindent\textbf{Running Time Analysis.}
	The term $|S| \cdot \log{\frac{1}{\eps}}\cdot n$ in the running time also follows from the above analysis, and the term $n^2$ comes from the construction of the hierarchical net $\{N_i\}_i$.
}

	\noindent\textbf{Correctness Analysis.}
	Suppose $C^\star:=(o_1, o_2, \ldots, o_k) \in [X]^k$ is the centers for the optimal clustering of $S$. 
	Let $P_1, P_2, \ldots, P_k$ be the clustering of $S$ with respect to $C^\star$, where $o_i$ is the center of $P_i$. 
	That is, $P_i := \{ x \in S \mid o_i = \arg\min_{o \in C^\star}{d(x,o)} \}$ (we break ties arbitrarily).
	In Claim~\ref{claim:alternative_center}, we show that for each $i\in [k]$, $o_i$ has a good approximation in $H$.

	\begin{claim}
		\label{claim:alternative_center}
		For each $i \in [k]$, there exists $o_i' \in H$ such that
		$\sum_{x \in P_i}{d^z(x, o_i')\cdot w(x)} \leq (1+\eps)^z\cdot \sum_{x \in P_i}{d^z(x, o_i)\cdot w(x)}$.
	\end{claim}
	\begin{proof}
		Fix some $i$.
		Let $y := \arg\min_{y' \in P_i}{d(o_i, y')}$, so $y\in S$.
		Define $j$ to be the integer such that $2^j \leq d(y, o_i) < 2^{j+1}$, and $j'$ to be the integer such that $2^{j'} \leq \eps \cdot 2^j < 2^{j'+1}$.
		Let $o$ be the nearest point in $N_{j'}$ to $o_i$.
		Then since $N_{j'}$ is a $2^{j'}$-covering, we have $d(o_i, o) \leq 2^{j'}\leq \eps\cdot 2^j \leq \eps\cdot d(y, o_i)$. Note that
		\begin{align*}
			\sum_{x \in P_i}{w(x) \cdot d^z(x, o)}
			&\leq \sum_{x \in P_i}{w(x)\cdot (d(x, o_i) + d(o_i, o))^z} \\
			&\leq \sum_{x \in P_i}{w(x) \cdot (d(x, o_i) + \eps\cdot d(y, o_i))^z} \\
			&\leq (1+\eps)^z\cdot  \sum_{x \in P_i}{w(x) \cdot d^z(x, o_i)}.
		\end{align*}
		Hence, if $o \in H$, we pick $o_i' := o$. By the above argument, it completes the proof.

	Thus, we only need to consider the case that
	$o \notin H$. Let $I := [a, b)$ be the invariant interval that contains $j$.
	Let $u := \Par^{(j)}(y)$. 
	Since $T$ is a simple net tree, we have $d(u,y)\leq 2^{j+1}$ by Fact~\ref{fact:simple}.
	Moreover, $d(o,o_i)\leq \eps \cdot d(y,o_i)< 2^{j}$. By the triangle inequality, we have
	\[
	d(u, o) \leq d(u,y)+ d(y, o_i) + d(o, o_i) \leq 5 \cdot 2^{j}.
	\]
	Thus, if $a\leq j \leq a + 3$ or $b - 6 - \lceil \log{\frac{1}{\eps}}\rceil\leq j\leq b-1$, then $o \in C_j^{(u)}\subseteq H$. 
	However, recall that $o \notin H$. 
	So we must have $a + 4 \leq j \leq b - 7 - \lceil \log{\frac{1}{\eps}} \rceil$. 
	%
	Also, if $b\leq a+11$, then no such $j$ exists.
	Thus, we consider the case that $b\geq a+12$ in the following.
	%
	%
	
	Let $U := \{ u \in N_j \mid \Des(u^{(j)}) \cap P_i \neq \emptyset \}$ be the set of height-$j$ net points $u$ such that the subtree $\Des(u^{(j)})$ has nonempty intersection with $P_i$, and we write $U = \{u_1, u_2, \ldots, u_l\}$.
	Define $Q_t := \Des(u^{(j)}_t) \cap P_i$, for $t\in [l]$. W.l.o.g., assume $y \in Q_1$.
	%
	%
	We show that picking $o_i' := y$ is sufficient, that is,
	\begin{align*}
		\sum_{x \in P_i}{w(x)\cdot d^z(x, y)} \leq (1+\eps)^z \cdot \sum_{x \in P_i}{w(x) \cdot d^z(x, o_i)}.
	\end{align*}
	\noindent\textbf{Analyzing $Q_1$.}
	By the definition of the invariant interval,
	there exists $v \in N_a$ such that $\Des(u^{(j)}_{1}) \cap S = \Des(v^{(a)}) \cap S$. Hence, $\diam(Q_1) \leq  \diam(\Des(v^{(a)})) \leq 8 \cdot 2^a$. (Note that this diameter bound also holds for $Q_t$ for $t \in [l]$.)
	By the definition of $y$, $d(y, o_i) \geq 2^j \geq 8 \cdot 2^a $.
	Therefore,
	\begin{align*}
		\sum_{x \in Q_1}{w(x) \cdot d^z(x, y)}
		&\leq \sum_{x \in Q_1}{ w(x) \cdot (8 \cdot 2^a)^z } 
		\leq \sum_{x \in Q_1}{w(x) \cdot d^z(y, o_i)} \\
		&\leq (1+ \eps)^z\cdot \sum_{x \in Q_1}{w(x) \cdot d^z(y, o_i)}  \\
		&\leq (1+\eps)^z \cdot \sum_{x \in Q_1}{w(x) \cdot d^z(x, o_i)}. \quad (\text{Definition of $y$})
	\end{align*}
	
	\noindent\textbf{Analyzing $Q_2,Q_3, \ldots, Q_l$.}
	%
	By the definition of invariant interval, we must have $d(Q_{t_1}, Q_{t_2}) > 2^{b-4}$ for any $1 \leq t_1 < t_2 \leq l$.
	%
	%
	To see this, since $[a,b)$ is an invariant interval,
	there exits $v_1,v_2\in N_{b-1}$ such that 
	$Q_{t_1} = \Des(v^{(b-1)}_{1}) \cap P_i$ and $Q_{t_2} = \Des(v^{(b-1)}_{2}) \cap P_i$.
	Since $Q_{t_1}\cap Q_{t_2}=\emptyset$, we have $v_1\neq v_2$.
	It implies that $d(v_1,v_2)\geq 2^{b-1}$ since $N_{b-1}$ is a $2^{b-1}$-packing.
	On the other hand, by the same argument for $Q_1$, we have $\diam(Q_{t_1})\leq 8\cdot 2^a$ and $\diam(Q_{t_2})\leq 8\cdot 2^a$.
	By the assumption that $b\geq a+12$,
	\[
	d(Q_{t_1}, Q_{t_2}) \geq d(v_1,v_2)-\diam(Q_{t_1})-\diam(Q_{t_2}) \geq 2^{b-1}-16\cdot 2^a \geq 2^{b-4}.
	\]
	
	By the triangle inequality, for $2\leq t \leq l$,
	$d(Q_t, o_i)
	\geq d(y, Q_t) - d(y, o_i)
	\geq d(Q_1, Q_t) - d(y, o_i)
	\geq 2^{b-4} - 2^j \geq  2^{b-5}$.
	Therefore, for $2\leq t \leq l$,
	\begin{align*}
		\sum_{x \in Q_t}{w(x) \cdot d^z(x, y)}
		&\leq \sum_{x\in Q_t}{w(x) \cdot ( d(x, o_i) + d(o_i, y) )^z} \quad (\text{triangle ineq.})\\
		&\leq \sum_{x \in Q_t}{w(x) \cdot ( d(x, o_i) + 2^{j+1} )^z} \\
		&\leq \sum_{x \in Q_t}{w(x) \cdot ( d(x, o_i) + \eps\cdot 2^{b-5}  )^z}  \quad (j\leq b-7-\lceil \log{\frac{1}{\eps}} \rceil)  \\
		&\leq (1+\eps)^z\cdot \sum_{x \in Q_t}{w(x) \cdot d^z(x, o_i)}. \quad (d(x,o_i)\geq d(Q_t, o_i)\geq 2^{b-5}).
	\end{align*}
	In conclusion, we have that
	\begin{align*}
		\sum_{x \in P_i}{w(x) \cdot d^z(x, y)}
		&= \sum_{t \in [l]}{ \sum_{x \in Q_t}{w(x) \cdot d^z(x, y)} } \\
		&\leq (1+\eps)^z\cdot\sum_{t \in [l]}{  \sum_{x\in Q_t}{w(x) \cdot d^z(x, o_i)} } \\
		&= (1+\eps)^z\cdot\sum_{x \in P_i}{w(x) \cdot d^z(x, o_i)}.
	\end{align*}
		This finishes the proof of Claim~\ref{claim:alternative_center}.
	\end{proof}

	Let $C' := \{ o_1', o_2', \ldots, o_k' \}$ be the centers asserted in Claim~\ref{claim:alternative_center}. Then we have
	\begin{align*}
		\sum_{x \in S}{w(x)\cdot d^z(x, C')}
		&\leq \sum_{i \in [k]}{\sum_{x \in P_i}{ w(x) \cdot  d^z(x, o_i') }} \\
		&\leq (1+\eps) \cdot \sum_{i \in [k]}{ \sum_{x \in P_i}{w(x)\cdot d^z(x, o_i)} } \\
		&= (1+\eps)\cdot \sum_{x \in S}{w(x)\cdot d^z(x, C)}.
	\end{align*}
	By the fact that $(1+\eps)^z \leq 1 + O(z\cdot \eps)$, we finish the proof of Theorem~\ref{theorem:app_local_search}.
	\end{proof}

\subsubsection{Accelerating Local Search Algorithms}
Recently, Friggstad et al.~\cite{DBLP:conf/focs/FriggstadRS16}
analyzed the local search algorithm for the $(k, z)$-clustering problem in a doubling metric $M$.
The local search algorithm considered in~\cite{DBLP:conf/focs/FriggstadRS16} works as follows:
in each iteration, the algorithm tests if it is possible to swap at most $\rho$ centers to improve the objective; if it is possible, the algorithm swaps arbitrary $\rho$ centers that improve the objective, and terminates otherwise.\jian{be a bit more precise about the algorithm. add one sentence about the condition. e.g., as long as the cost decreases}\shaofeng{addressed.}
They showed that for $\rho := \rho(\eps, \DDim(M), z)$, the number of iterations 
is polynomial in the input size, and the output is a $(1+\eps)$-approximation for the $(k, z)$-clustering problem. The running time for each iteration is roughly $n^{\rho}$.

For the special case of bounded dimensional Euclidean spaces, 
Friggstad et al. showed that by utilizing the centroid set~\cite{DBLP:journals/dcg/Matousek00}, one can improve the running time to roughly $(k/\eps)^{O(\rho)}$
per iteration.
Since succinct centroid sets for doubling metrics were not known, they did not achieve an improved running time for the more general doubling metrics.

We show how to improve the running time per iteration in doubling metrics.
As noted in~\cite{DBLP:conf/focs/FriggstadRS16}, the local search algorithm actually works on weighted instance.
In the improved algorithm, we first construct an $\eps$-coreset using Theorem~\ref{thm:coreset}, denoted as $S$ with weight $w$.
Then by applying Theorem~\ref{theorem:app_local_search} on $(S,w)$ with error parameter $O(\frac{\eps}{z})$, we get an $(O(\eps),k,z)$-centroid set $H$. 
Extend $w': H \rightarrow \mathbb{R}_{\geq 0}$ such that $w'(x) = w(x)$ for $x \in S$ and $w(x) = 0$ otherwise.
Then, we apply the local search algorithm on the weighted instance $(H, w')$.
In this way, we achieve a per iteration running time of $O(2^{O(z\log z)} \cdot \frac{k}{\eps}\cdot \log\frac{1}{\tau})^{O(\rho)}$, by a $\poly(|X|)$ preprocessing time. This is comparable to the result for Euclidean spaces in~\cite{DBLP:conf/focs/FriggstadRS16}.

We summarize our conclusion in the following corollary.
\begin{corollary}
	\label{corollary:fast_local_search}
	Let $M(X, d)$ be a (finite) metric space, and consider the $(k, z)$-clustering problem in $M$.
	The local search algorithm for the $(k, z)$-clustering problem that swaps $\rho := \DDim(M)^{O(\DDim(M))} \cdot (\frac{2^z}{\eps})^{O(2^z)\cdot \DDim(M)\cdot \eps^{-1}}$ centers (as defined in~\cite{DBLP:conf/focs/FriggstadRS16}) in each iteration, gives a $(1+\eps)$-approximate solution after polynomial (in the input size) number of iterations. 
	Furthermore, with a $\poly(|X|)$-time preprocessing procedure that succeeds with probability at least $1- \tau$, the local search algorithm runs in $(2^{O(z\log z)}\cdot \frac{k}{\eps}\cdot \log{\frac{1}{\tau}})^{O(\rho)}$ time per iteration.
\end{corollary}

\section{Acknowledgment}
We are grateful to Robert Krauthgamer for fruitful discussions and insightful comments, particularly for pointing to us the snowflake embedding literature, and the example in Proposition~\ref{prop:expander}.

\bibliographystyle{alpha}
\bibliography{main}

\appendix

\section{Proof of Lemma~\ref{lemma:restate_weak_app}}
\label{sec:approximation}
By classical results (e.g.,~\cite{DBLP:journals/jcss/LiLS01}), a small sized $\alpha$-approximation can be constructed efficiently by random sampling
for range space with bounded shattering dimension.
But such results do not suffice for us since we only have a bound for the probabilistic dimension
(Theorem~\ref{theorem:weighted}).

Our proof idea is similar to the analysis of $\alpha$-net\footnote{$\calS$ is an $\alpha$-net for range space $(\calF, \ranges(\calF))$, if for any $\calR\in \ranges(\calF)$, $\frac{|\calR|}{|\calF|} > \alpha$ implies $\calR \cap \calS \neq \emptyset$.} for range spaces of bounded VC dimension (see~\cite{kearns1994introduction}),
where the key technique is the double sampling.
In their proof, to show an independent sample $\calS$ of size $m$ is an $\alpha$-approximation with constant probability,
they alternatively first draw $2m$ independent samples $\calP$ and then draw $m$ independent samples from $\calP$ to form $\calS$.
Then they showed that it suffices to prove $\calS$ is an $\alpha$-net for functions projected on $\calP$.
As an important step that uses the bounded dimension, they applied Sauer's Lemma~\cite{sauer1972density} on $\calP$ to show that the number of distinct projected functions is bounded. Finally, this makes it possible to use a union bound on the projected functions, so that even a small sample can result in an $\eps$-net with high probability.

In our proof, we need to use the weaker guarantee of dimension in replacement of the Sauer's lemma. In addition,
since we have two very different source of randomness, namely $\calS$ and $\calF$, we need separate the two in the double sampling argument.
We use the fact that $\calF$ is a random indexed function set with index set $V$, and do the double sampling on $V$. In particular, suppose $P$ is a size-$2m$ sample from $V$ and $S$ is a size-$m$ sample from $P$.
Then we define $\calP := \calF_P$ and $\calS := \calF_S$.
In this way, we separated the randomness of $\calS$ and $\calF$.

We follow mostly the same double sampling argument.
In the final step, we additionally introduce a conditional probability argument to use the randomness of $\calF$, where the condition guarantees that $|\ranges(\calP)|$ is bounded, which replaces the use of Sauer's lemma in the original proof. 

\begin{lemma}[restatement of Lemma~\ref{lemma:restate_weak_app}]
	\label{lm:balltoapp}
Suppose $\calF$ is a random indexed function set with index set $V$. In addition, suppose $T : \mathbb{N} \times \mathbb{R}_{\geq 0}$ satisfies for any $H \subseteq V$ and $0<\gamma<1$,
\begin{align*}
\Pr[|\ranges(\calF_H)| \leq T(|H|, \gamma) ] \geq 1 - \gamma.
\end{align*}
%
%
Let $\calS$ be a collection of $m$ uniformly independent samples from $\calF$. 
Then with probability at least $1-\tau$, $\calS$ is an $\alpha$-approximation of the range space $(\calF, \ranges(\calF))$, where the randomness is taken over $\calS, \calF$ and
\begin{align*}
\alpha := \sqrt{\frac{48\left(\log(T(2m, \frac{\tau}{4})) + \log{\frac{8}{\tau}}\right)}{m}}.
\end{align*}
\end{lemma}

\begin{proof}
	
	Let $P$ be an independent sample from $V$ of size $2m$, and let $S$ be a size-$m$ uniform sample from $P$. Then define $\calS := \calF_S$ which is the final set we require in the lemma.
	It is easy to see $\calS$ is an independent sample of size $m$.
	Also, define $\calP := \calF_P$. Note that $P$ and $S$ are independent of the choice of $\calF$, but $\calP$ and $\calS$ have the randomness of $\calF$.
	
	Let $A$ be the event that there exists $\calR\in \ranges(\calF)$ such that $\left| \frac{|\calR|}{|\calF|} - \frac{|\calS\cap \calR|}{|\calS|} \right| > \alpha$.
	It suffices to prove $\Pr[A] \leq \tau$.
	Let $B$ be the event that there exists $\calR\in \ranges(\calF)$ such that
	$\left| \frac{|\calR|}{|\calF|} - \frac{|\calS\cap \calR|}{|\calS|} \right| > \alpha $ and
	$\left| \frac{|\calP \cap \calR|}{|\calP|} - \frac{|\calS\cap \calR|}{|\calS|} \right| > \frac{\alpha}{4}$.
	
	\begin{claim}
		\label{claim:compare_A_B}
		$\Pr[A] \leq 2 \Pr[B]$.
	\end{claim}
	\begin{proof}
		It suffices to show that $\Pr[\overline{B} \mid A] \leq \frac{1}{2}$, since it would imply
		\begin{align*}
			\Pr[B]
			= \Pr[A \cap B]
			= \Pr[B \mid A] \cdot \Pr[A]
			\geq \frac{1}{2}\Pr[A].
		\end{align*}
		
		Since
		$\Pr[\overline{B} \mid A] = \E_{\calS,\calF}[\Pr[\overline{B} \mid A, \calS, \calF]]$,
		it suffices to show $\Pr[\overline{B} \mid A, \calS, \calF] \leq \frac{1}{2}$
		for any $A$, $\calS$ and $\calF$.
		Since $A$ happens with given $\calF, \calS$, there exists $\calR_0$ (defined with respect to $\calF, \calS$) such that $ \left| \frac{|\calR_0|}{|\calF|} - \frac{|\calS\cap \calR_0|}{|\calS|} \right| > \alpha$. Fix any such $\calR_0$.
		Then the event $(\overline{B}\mid A, \calS, \calF)$ implies $\left| \frac{|\calP\cap \calR_0|}{|\calP|} - \frac{|\calS\cap \calR_0|}{|\calS|} \right| \leq \frac{\alpha}{4}$.
		Let $\calZ := \calP\backslash \calS$.
		The event $(\overline{B}\mid A,\calS, \calF)$ implies
		\begin{align*}
			\left|\frac{|\calZ\cap \calR_0|}{|\calZ|} - \frac{|\calR_0|}{|\calF|}\right|
			&= \left| \frac{|\calZ\cap \calR_0|}{m} - \frac{|\calS\cap \calR_0|}{|\calS|} + \frac{|\calS\cap \calR_0|}{|\calS|} - \frac{|\calR_0|}{|\calF|} \right| \\
			&\geq \left|\frac{|\calR_0|}{|\calF|} - \frac{|\calS\cap \calR_0|}{|\calS|} \right|
				- \left|\frac{|\calZ\cap \calR_0|}{m} - \frac{|\calS\cap \calR_0|}{m}\right| \\
			&= \left|\frac{|\calR_0|}{|\calF|} - \frac{|\calS\cap \calR_0|}{|\calS|} \right| - 2\cdot \left|\frac{|\calZ \cap \calR_0|}{2m} - \frac{|\calS\cap \calR_0|}{2m}\right| \\
			&= \left|\frac{|\calR_0|}{|\calF|} - \frac{|\calS\cap \calR_0|}{|\calS|} \right| - 2\cdot \left|\frac{|\calP\cap \calR_0|}{|\calP|} - \frac{|\calS\cap \calR_0|}{|\calS|}\right| \\
			&> \alpha - \frac{\alpha}{2} \\
			&= \frac{\alpha}{2}.
		\end{align*}
		
		Therefore,
		\begin{align*}
			\Pr[\overline{B} \mid A, \calS, \calF]
			\leq \Pr\left[\left|\frac{|\calZ \cap \calR_0|}{m} - \frac{|\calR_0|}{|\calF|}\right| > \frac{\alpha}{2} \mid A, \calS, \calF \right]
			\leq 2\exp\left(- 2m \cdot \frac{\alpha^2}{4} \right)
			\leq \frac{1}{2},
		\end{align*}
		where the second last inequality is by Hoeffding's inequality, using the randomness of $P$:
		because $\E[|\calZ \cap \calR_0| \mid A, \calS, \calF] = \frac{|\calR_0|}{|\calF|}\cdot m$ and $|\calZ \cap \calR_0|$ may be viewed as the sum of $m$ independent $\left\{0, 1\right\}$ random variables such that each variable takes value $1$ with probability $\frac{|\calR_0|}{|\calF|}$.
		
		This finishes the proof of Claim~\ref{claim:compare_A_B}.
	\end{proof}

	Recall that our goal is to prove that $\Pr[A]\leq \tau$. 
	By Claim~\ref{claim:compare_A_B}, it suffices to prove $\Pr[B] \leq \frac{\tau}{2}$.

	For $\calR\in \ranges(\calF)$, define $B_\calR$ to be the event that $\left| \frac{|\calR|}{|\calF|} - \frac{|\calS\cap \calR|}{|\calS|} \right| > \alpha $ and
	$\left| \frac{|\calP \cap \calR|}{|\calP|} - \frac{|\calS\cap \calR|}{|\calS|} \right| > \frac{\alpha}{4}$, and let $H_\calR$ be the event that $\left| \frac{|\calP \cap \calR|}{|\calP|} - \frac{|\calS\cap \calR|}{|\calS|} \right| > \frac{\alpha}{4}$.
	Observe that $B_\calR$ implies $H_\calR$. So $B = \bigcup_{\calR \in \ranges(\calF)}{B_\calR}$ implies $\bigcup_{\calR \in \ranges(\calF)}{H_\calR}$. Hence
	$\Pr[B] \leq \Pr\left[\bigcup_{\calR\in \ranges(\calF)}{H_\calR}\right]$.
	
	Observe that
	\begin{align*}
		\Pr[B]
		\leq \Pr\left[\bigcup_{\calR \in \ranges(\calF)}{H_\calR}\right]
		= \E_{P}\left[ \Pr\left[\bigcup_{\calR \in \ranges(\calF)}{H_\calR} \mid P\right] \right].
	\end{align*}
	So it suffices to show that $  \Pr\left[\bigcup_{\calR \in \ranges(\calF)}{H_\calR} \mid P\right] \leq \frac{\tau}{2}$ for any fixed $P$.
	
	Fix some $P$. Since $\calS\subseteq \calP$, we have $\bigcup_{\calR\in \ranges(\calF)}{H_\calR} = \bigcup_{\calR \in \ranges(\calP)}{H_\calR}$. Hence,
	\begin{align*}
		\Pr\left[\bigcup_{\calR\in \ranges(\calF)}{H_\calR} \mid P\right]
		= \Pr\left[\bigcup_{\calR \in \ranges(\calP)}{H_\calR} \mid P\right].
	\end{align*}
	
	\noindent\textbf{Analyzing $\Pr\left[\bigcup_{\calR \in \ranges(\calP)}{H_\calR} \mid P\right]$.}
	Given $P$, let $\mathcal{E}$ be the event that $|\ranges(\calP)| \leq T(2m, \frac{\tau}{4})$.
	We emphasize that $\mathcal{E}$ is defined with respect to a given $P$, and is 
	still a random event (the randomness comes from $\calF$).
	By our assumption in the lemma, $\Pr[\mathcal{E} \mid P] \geq 1- \frac{\tau}{4}$ for any $P$.
	Then, we can see that
	\begin{align}
		\Pr\left[\bigcup_{\calR\in \ranges(\calP)}{H_\calR} \mid P\right]
		&= \Pr\left[\bigcup_{\calR\in \ranges(\calP)}{H_\calR } \cap \mathcal{E} \mid P\right] +
		\Pr\left[\bigcup_{\calR \in \ranges(\calP)}{H_\calR} \cap \overline{\mathcal{E}} \mid P\right] \nonumber \\
		&\leq \Pr\left[\bigcup_{\calR \in \ranges(\calP)}{H_\calR} \cap \mathcal{E} \mid P\right] + \frac{\tau}{4} \nonumber \\
		&\leq \Pr\left[\bigcup_{\calR\in \ranges(\calP)}{H_\calR} \mid P, \mathcal{E}\right] + \frac{\tau}{4} \nonumber \\
		& = \E_{\calF}\left[ \Pr\left[\bigcup_{\calR\in \ranges(\calP)}{H_\calR} \mid P, \mathcal{E}, \calF\right] \right] + \frac{\tau}{4}. \label{eqn:conditional_prob} 
	\end{align}
	Recall that $\calP = \calF_P$.
	Define 
	$\Sigma:= \left\{ (P,\calF,\calR) \mid P\subseteq V, |P|=2m, |\ranges(\calP)| \leq T(2m, \frac{\tau}{4}), \calR\in \ranges(\calP) \right\}$ which is the set of relevant tuples $(P,\calF,\calR)$ that we need to consider, and
	let
	\begin{align*}
		P_{\max} := \max_{(P,\calF,\calR)\in \Sigma}{\Pr\left[H_\calR \mid P, \mathcal{E}, \calF\right]}. 
	\end{align*}
	Then by union bound,
	\begin{align*}
		\Pr\left[\bigcup_{\calR\in \ranges(\calP)}{H_\calR} \mid P, \mathcal{E}, \calF\right] 
		\leq T(2m, \frac{\tau}{4}) \cdot P_{\max}.
	\end{align*}
	By Inequality~\ref{eqn:conditional_prob},
	%
	it suffices to bound $P_{\max}$.
	Recall that $S$ is formed by drawing $m$ independent samples from $P$ (without replacement).
	In fact, one can show those indicator random variables (indicating whether an element in $S$ or not) 
	{\em negatively correlated}, (see e.g.,~\cite{DBLP:journals/siamcomp/PanconesiS97,DBLP:journals/jacm/GandhiKPS06}). Therefore, we are in place to apply the generalized Chernoff bound that works on negatively correlated random variables~\cite[Theorem 3.1]{DBLP:journals/jacm/GandhiKPS06}, and we conclude for any $(P,\calF,\calR)\in \Sigma$,
	\begin{align*}
		\Pr\left[H_\calR \mid P, \mathcal{E}, \calF\right]
		= \Pr\left[\left| \frac{|\calP\cap \calR|}{|\calP|} - \frac{|\calS \cap \calR|}{|\calS|} \right| > \frac{\alpha}{4}\right]
		\leq 2\exp\left(-\frac{\alpha^2 m}{48}\right),
	\end{align*}
	and this implies $P_{\max} \leq 2\exp\left(-\frac{\alpha^2 m}{48}\right)$.

	\noindent\textbf{Concluding Lemma~\ref{lm:balltoapp}.}
	Combining Inequality~\ref{eqn:conditional_prob}, we conclude that
	\begin{align*}
	\label{eq:totalbound}
		\Pr\left[\bigcup_{\calR \in \ranges(\calP)}{H_\calR} \mid P\right] 
		\leq T(2m, \frac{\tau}{4}) \cdot 2\exp\left(-\frac{\alpha^2 m}{48}\right)+\frac{\tau}{4}.
	\end{align*}
	Plugging in the value of $\alpha$, we complete the proof of Lemma~\ref{lm:balltoapp}.
\end{proof}

\section{Proof Sketch of Corollary~\ref{corollary:weighted_z_ball}}
\label{section:weighted_z_ball}

\jian{this section can go into appendix.}\lingxiao{Since we use the corollary in this section, I think it is better to give an explanation of Corollary~\ref{corollary:weighted_z_ball} in the main body.}
Observe that only the $z=1$ case is addressed in Theorem~\ref{theorem:weighted}. We claim that using the same technique, we can obtain Corollary~\ref{corollary:weighted_z_ball} for more general $z$. We highlight only the necessary changes to the proof.

\begin{corollary}[restatement of Corollary~\ref{corollary:weighted_z_ball}]
	\label{corollary:restate_weighted_z_ball}
Suppose $M(X, d)$ is a metric space with a gap-$2$ weight function $w: X\rightarrow \R_{\geq 0}$.
Let $z>0$, $0 < \epsilon \leq \frac{1}{100z}$ and $0 < \tau < 1$ be constant. There exists a random $\epsilon$-smoothed distance function $\delta$ (defined with respect to some random net tree), such that for $\calF := \left\{ w(x)\cdot \delta^z(x, \cdot) \mid x\in X \right\}$
($\calF$ is defined with respect to the $z$-th power of the random smoothed distance function $\delta$),
and any $H\subseteq X$,
\begin{align*}
\Pr_\delta\left[|\ranges(\calF_H)| \leq  O\left( \frac{1}{\epsilon} \right)^{O(\DDim(M))} \cdot \log{\frac{|H|}{\tau}} \cdot |H|^6  \right] \geq 1 - \tau,
\end{align*}
In addition, for $x, y \in X$, it holds that
\begin{align*}
(1 - O(\eps\cdot z))\cdot \delta^z(x, y)\leq d^z(x, y) \leq (1 + O(\eps \cdot z)) \cdot \delta^z(x, y).
\end{align*}
In other words, $\PDim_\tau(\calF) \leq O\left( \DDim(M)\cdot \log (1/\eps) + \log\log 1/\tau \right) $.
\end{corollary}

\begin{proof} 
	The distortion part follows immediately from Lemma~\ref{lemma:distortion}, and it remains to lower bound the probability that $\ranges(\calF_H)$ is small.
	Since the argument is very similar to that of Theorem~\ref{theorem:weighted}, we will go through the argument and only highlight the most significant changes.
	The constructions of $\delta$, $H_i$, $\Bdt_i$, $\Bd_i$ and $\calH(a,b)$ are unchanged.
	Basically, the largest change is $\frac{r}{w_i} \rightarrow (\frac{r}{w_i})^{1/z}$ for any $r\geq 0$ and $i\in [l]$.
	We take $r$-representatives and critical intervals as an example.

	\noindent\textbf{$r$-Representatives.}
	In the definition of $r$-representatives, it is still $R_i(r) := N_{\zeta_i}^{(H)}$
	but $\zeta_i$ is changed to be the integer such that $2^{\zeta_i} \leq (\frac{r}{w_i})^{1/z} < 2^{\zeta_i+1}$.
	Lemma~\ref{lemma:basic_center_prop} still holds except item 3, which should be changed to ``there exists $y\in R(r)$ such that $x\in B^d_i(y, (\frac{r}{w_i})^{1/z})$''.
	
	\noindent\textbf{Critical Intervals.}
	The second item changes to 
	$
	\Bd_i\left(x, (\frac{r_1}{w_i})^{1/z}\right)
	= 
	\Bd_i\left(x, (\frac{r_2}{w_i})^{1/z}\right)
	$. 
	Lemma~\ref{lemma:num_critical_interval} follows from basically the same argument.
	
	\noindent\textbf{Lemmas. }
	Among the statements and the proofs of all lemmas, we should replace $\frac{r}{w_i}$ by $(\frac{r}{w_i})^{1/z}$ everywhere.
	Next, we state the other differences.

	For Lemma~\ref{lemma:bounded_interval}, the statement changes to an interval $I:=[a,b)$ with $b\leq 2^z a$. 
	
	In Lemma~\ref{lemma:long_interval}, the condition should be $b> 6^z a$ instead of $b> 6a$.
	
	In Lemmas~\ref{lemma:large_r} and~\ref{lemma:subset}, the statement changes a little: $r \in [a,\frac{b}{2^9}]$ should be changed to $r \in [a,\frac{b}{2^{O(z)}}]$, and the Lemma follows from a similar argument as in the original proof.
	
	The statement of Lemma~\ref{lemma:upper_bound_critical_interval} does not change.
	In the proof, the definition of $s$ changes to
	\begin{align*}
		s := \left\lceil \log{a} + O(z)\cdot \left( \log{\DDim(M)} + \log{m} + \log{\frac{1}{\tau}} \right) \right\rceil.
	\end{align*}
	Moreover, $r_1$ should be changed to $\frac{2^s}{\lambda^{z}}$ and $r_2$ should be changed to $r_2 := \frac{b}{2^{O(z)}}$.
	Since the statement of Lemma~\ref{lemma:bounded_interval} changes to $b \leq 2^z a$, the bound for $|\calH(a, r_1)|$ and $|\calH(r_2, b)|$ actually doesn't change.
	%
	%
	
\end{proof}

\section{Proofs for Smoothed Distance Functions}
\label{section:smooth_proof}

\begin{lemma}[restatement of Lemma~\ref{lemma:distortion}]
	\label{lemma:restate_distortion}
	If $T$ is $c$-covering, then for any $x,y \in X$ and any $\eps>0$,
	\begin{align*}
	(1-4c\cdot \eps)\cdot \delta(x,y) \leq d(x,y) \leq (1+4c\cdot \eps)\cdot \delta (x,y).
	\end{align*}
\end{lemma}
\begin{proof}
	Let $j = h(x,y)$. By definition,
	\begin{equation*}
	\delta(x,y) = d(\Par^{(j)}(x), \Par^{(j)}(y))\geq \frac{2^j}{\eps}.
	\end{equation*}
	By Fact~\ref{fact:des_dis}, $d(x, \Par^{(j)}(x)) \leq c\cdot 2^{j+1}$,
	and $d(y, \Par^{(j)}(y)) \leq c \cdot 2^{j+1}$.
	By the triangle inequality,
	\begin{align*}
	d(x,y)
	&\leq d(x, \Par^{(j)}(x)) + d(\Par^{(j)}(x), \Par^{(j)}(y)) + d(\Par^{(j)}(y), y) \\ &\leq (1 + 4c\cdot \eps)\cdot \delta(x, y).
	\end{align*}
	Similarly, one can see that
	\begin{align*}
	d(x, y)
	&\geq d(\Par^{(j)}(x), \Par^{(j)}(y)) - d(x, \Par^{(j)}(x)) - d(y, \Par^{(j)}(y)) \\
	&\geq (1 - 4c \cdot \eps) \cdot \delta(x,y).
	\end{align*}
	This finishes the proof of the lemma.
\end{proof}

\begin{lemma}[restatement of Lemma~\ref{lemma:delta_equ}]
	\label{lemma:restate_delta_equ}
	For any $x, y \in X$, assume that $j = h(x, y)$, $u = \Par^{(j)}(x)$ and $v = \Par^{(j)}(y)$. Then for any $x'\in \Des(u^{(j)})$ and $y'\in \Des(v^{(j)})$, we have $ \delta(x, y) = \delta(x', y') = d(u, v) $.
\end{lemma}
\begin{proof}
	By definition, we immediately have $\delta(x, y) = d(u, v)$.
	Observe that $\Par^{(j)}(x') = \Par^{(j)}(x)$, and $\Par^{(j)}(y) = \Par^{(j)}(y')$.
	By Definition \ref{definition:eps_smooth}, we know $d(\Par^{(j)}(x'),\Par^{(j)}(y'))=d(u,v) \geq \frac{2^j}{\eps}$. Thus, we have $h(x', y') \geq j$.
	On the other hand, for $j' > j = h(x, y)$, we have $d(\Par^{(j')}(x'),\Par^{(j')}(y'))=d(\Par^{(j')}(x),\Par^{(j')}(y)) < \frac{2^{j'}}{\eps}$, where the last inequality is by the definition of $h(x,y)$. Thus, we have $h(x', y') \leq j$. Therefore, $h(x, y) = h(x', y')$ which implies that $\delta(x', y') = d(u, v)$.
\end{proof}

\begin{lemma}[restatement of Lemma~\ref{lemma:ball_equal}]
	\label{lemma:restate_ball_equal}
	Suppose $\{N_i \mid i \leq L\}$ is a hierarchical net and $T$ is a $c$-covering net tree with respect to $\{N_i\}_i$.
	Consider $0<\eps\leq \frac{1}{8c}$ and $r>0$. 
	Let $\lambda := \frac{\eps\cdot (1-5c\eps)}{20 (1+4c\eps)}$.
	Define $j$ to be the integer satisfying that $2^{j-1} \leq \lambda \cdot r $.
	Then for any $x,x'\in X$, if $\Par^{(j)}(x) = \Par^{(j)}(x')$, we have $\Bdt(x, r) = \Bdt(x', r)$.
\end{lemma}
	\begin{proof}
		Observe that it suffices to prove the case that $2^{j-1} \leq \lambda \cdot r< 2^{j}$. 
		To see this, assume that we have proved this case. Consider some $r' > r$ and an
		integer $j'>j$ such that $2^{j'-1} \leq \lambda \cdot r' < 2^{j'}$. Then we have $\Par^{(j')}(x) = \Par^{(j')}(x')$ because $\Par^{(j)}(x) = \Par^{(j)}(x')$. 
		By the assumption, we have $\Bdt(x, r') = \Bdt(x', r')$.
		Thus, we focus on the case that $2^{j-1} \leq \lambda \cdot r< 2^{j}$ in the following.
		
		It suffices to show that $v \in \Bdt(x, r)$ if and only if $v \in \Bdt(x', r)$.
		We only need to show the direction that if $v \in \Bdt(x, r)$ then $v \in \Bdt(x', r)$ since the ``only if'' direction is symmetric.
		Define $t := \frac{1 - 5c\eps}{1 + 4c\eps} = \frac{20 \lambda}{\eps}$.
		We consider the following two cases.
		\begin{itemize}
			\item $\delta(x, v) \leq t\cdot r \leq r$.
			We prove that $\delta(x', v) \leq r$.
			By Fact~\ref{fact:des_dis} and the fact that $\Par^{(j)}(x)=\Par^{(j)}(x')$, we have $d(x, x') \leq d(x,\Par^{(j)}(x))+d(x',\Par^{(j)}(x')) \leq c\cdot 2^{j+2}$.
			Therefore by Lemma \ref{lemma:restate_distortion}, we have
			\begin{eqnarray*}
			\begin{split}
			\delta(x', v)
			\leq & \frac{1}{(1-4c\cdot \eps)}\cdot d(x', v) & (\text{Lemma \ref{lemma:restate_distortion}}) \\
			\leq & \frac{1}{(1-4c\cdot \eps)} \cdot ( d(x, x') + d(x, v) ) & (\text{triangle ineq.})\\
			\leq & \frac{1}{(1-4c\cdot \eps)} \cdot (c\cdot 2^{j+2} + (1+4c\cdot \eps)\cdot \delta(x,v)) & (d(x,x')\leq c\cdot 2^{j+2},~\text{Lemma \ref{lemma:restate_distortion}})\\
			\leq & \frac{1}{(1-4c\cdot \eps)} \cdot (8c\lambda + (1+4c\cdot\eps)\cdot t) \cdot r & (2^{j-1}\leq \lambda \cdot r,~\delta(x',v)\leq t\cdot r)\\  
			<& r. & (\text{Definitions of $\lambda$ and $t$})
			\end{split}
			\end{eqnarray*}
			\item $\delta(x, v) > t \cdot r$. We prove that $\delta(x, v) = \delta(x', v)$.
			By Fact~\ref{fact:des_dis},
			$d(x, x') \leq c\cdot 2^{j+2}$.
			Observe that by Lemma \ref{lemma:restate_distortion},
			\begin{eqnarray*}
			\begin{split}
			\delta(x', v)
			&\geq \frac{1}{1+4c\cdot \eps}\cdot d(x', v) & (\text{Lemma \ref{lemma:restate_distortion}})\\
			& \geq \frac{1}{1+4c\cdot \eps} \cdot (d(v, x) - d(x, x')) & 	(\text{triangle ineq.})\\
			&\geq \frac{1}{1+4c\cdot \eps}\cdot ( (1-4c\cdot \eps)\cdot \delta(x,v) - c\cdot 2^{j+2} ) & (\text{Lemma \ref{lemma:restate_distortion}},~d(x,x')\leq c\cdot 2^{j+2})\\
			&\geq \frac{1}{1+4c\cdot \eps}\cdot ((1-4c\cdot \eps)t - 8c\lambda) \cdot r. & (\delta(x',v)> t\cdot r,~2^{j-1}\leq \lambda \cdot r)
			\end{split}
			\end{eqnarray*}
			Therefore, we can see that
			\begin{eqnarray*}
			\begin{split}
			d(x', v)
			&\geq (1-4c\eps) \cdot \delta(x', v) & (\text{Lemma \ref{lemma:restate_distortion}}) \\
			& \geq (1-4c\eps) \cdot \frac{(1-4c\eps)t - 8c\lambda}{1+4c\eps}\cdot r \\
			&\geq (1-4c\eps) \cdot \frac{(1-4c\eps)t - 8c\lambda}{1+4c\eps} \cdot \frac{2^{j-1}}{\lambda} & (~2^{j-1}\leq \lambda \cdot r) \\
			&\geq (1-4c\eps) \cdot \frac{(20/\eps-88c)\cdot \lambda}{1+4c\eps} \cdot \frac{2^{j-1}}{\lambda} & (t= \frac{20\lambda}{\eps})\\
			&\geq \frac{2^{j+3}}{\eps}. & (0< \eps \leq \frac{1}{8c})
			\end{split}
			\end{eqnarray*}
			Hence, $d(\Par^{(j)}(x'), \Par^{(j)}(v)) \geq d(x', v) - d(\Par^{(j)}(x'), x') - d(\Par^{(j)}(v), v) \geq \frac{2^{j+3}}{\eps}- c\cdot 2^{j+2} \geq \frac{2^j}{\eps}$ since $\eps \leq \frac{1}{8c}$.
			This implies that $h(x', v) \geq j$.
			Hence, $\Par^{(j')}(x) = \Par^{(j')}(x')$ for $j' := h(x', v)$, since $\Par^{(j)}(x) = \Par^{(j)}(x')$ and $j'\geq j$.
			Thus by Lemma~\ref{lemma:restate_delta_equ}, $\delta(x, v) = \delta(x', v)$.
		\end{itemize}
	\end{proof}

\begin{lemma}[restatement of Lemma~\ref{lemma:ball_laminar}]
	\label{lemma:restate_ball_laminar}
	Suppose $\{N_i \mid i \leq L\}$ is a hierarchical net and $T$ is a $c$-covering net tree with respect to $\{N_i\}_i$.
	Consider $0<\eps\leq \frac{1}{8c}$ and $r>0$. 
	Let $\lambda := \frac{\eps\cdot (1-5c\eps)}{20 (1+4c\eps)}$.
	Suppose $j$ is an integer such that $2^{j-1} \leq \lambda \cdot r$.
	Then for any $x \in X$ and $v \in N_{j}$, either $\Des(v^{(j)}) \subseteq \Bdt(x, r)$ or $\Des(v^{(j)}) \cap \Bdt(x, r) = \emptyset$.
\end{lemma}
\begin{proof}
	We use Lemma~\ref{lemma:restate_ball_equal} to prove this lemma.
	Fix $x\in X$ and $v\in N_j$.
	Suppose $\Des(v^{(j)}) \cap \Bdt(x, r) \neq \emptyset$. 
	We want to prove $\Des(v^{(j)}) \subseteq \Bdt(x, r)$ in this case.
	Let $x' \in \Des(v^{(j)}) \cap \Bdt(x, r)$.
	Then we have $\delta(x, x')\leq r$.
	This implies that $x \in \Bdt(x', r)$.
	It suffices to show that for any $y \in \Des(v^{(j)})$, $\delta(y, x) \leq r$.
	By Lemma~\ref{lemma:restate_ball_equal}, for any $y \in \Des(v^{(j)})$, $\Bdt(y, r) = \Bdt(x', r)$. Hence, $x\in \Bdt(y, r)$ by the fact that $x\in \Bdt(x', r)$. This implies $\delta(y, x) \leq r$, which means $y \in \Bdt(x, r)$.
	This completes the proof.
\end{proof}

\section{Proof of Theorem~\ref{thm:totalsen}}
\label{section:proof_sensitivity}
\begin{theorem}[restatement of Theorem~\ref{thm:totalsen}]
	\label{thm:restate_totalsen}
	Given a metric space $M(X,d)$ for the $(k,z)$-clustering problem, there exists an algorithm that computes an upper bound $\pi_x$ of $2\sigma_X(x)$ for any $x\in X$, such that
	\begin{align*}
	\sum_{x\in X}\pi_x =O(2^{O(z\log z)} k), \quad \quad \forall z>0,
	\end{align*}
	with probability at least $1-\tau$.
	Moreover, the running time is $\poly(n)$.
\end{theorem}

Theorem~\ref{thm:restate_totalsen} follows immediately from the following two lemmas.

\begin{lemma}[restatement of Theorems 7 and 9 in \cite{varadarajan2012sensitivity}]
	\label{lemma:vx12}
	Given a $k$ point set $B\subseteq X$ such that $\kdist_z(X,B)\leq c\cdot \min_{C\in [X]^k} \kdist_z(X,C)$ for some $c\geq 1$, we can compute an upper bound $\pi_x$ of $2\sigma_X(x)$ for each $x\in X$ satisfying $\sum_{x\in X}\pi_x =O(c\cdot 2^{2z}k)$.
	Moreover, the computation time is $O(nk)$.
	\footnote{Note that Theorems 7 and 9 in \cite{varadarajan2012sensitivity} only consider the Eucludean space. However, the proofs of them directly work for any metric space.}
\end{lemma}

\begin{lemma}[\cite{gupta2008simpler}]
	We can compute a set $B\subseteq X$ of $k$ points in $\poly(n)$ time such that,
	$$
	\kdist_z(X,B)\leq 2^{O(z\log z)} \min_{C\in [X]^k} \kdist_z(X,C).
	$$
\end{lemma}

\section{Proof of Lemma~\ref{lm:sen}}
\label{sec:lemma_sen}

\begin{lemma}[restatement of Lemma~\ref{lm:sen}]
	For any $x\in X$, $\theta_x$ is an integer satisfying that $\theta_x \geq n\cdot \max_{C\in \calO}\frac{\psi_x(C)}{\sum_{\psi_y\in \Psi}\psi_y(C)}$.
	Moreover, $\sum_{x\in X} \theta_x = O(2^{O(z \log z)} kn)$.\shaofeng{what is $F$???}\lingxiao{Modify the typos.}
\end{lemma}

\begin{proof}
	By definition, $\theta_x  $ is an integer larger than $n\pi_x$.\jian{typo??}\lingxiao{It is a typo. Fix it.}
	By the definition of $\pi_x$, we have
	\begin{equation}
	\label{eq:sen1}
	\theta_x> n\pi_x \geq 2n\cdot \max_{C\in \calO}\frac{d^z(x,C)}{\sum_{y\in X}d^z(y,C)}.
	\end{equation}
	By Corollary~\ref{corollary:weighted_z_ball} and the definition of $\psi_x$, we have $d^z(x,C)\in (1\pm \eps) \psi_x(C)$ since $d^z(x,y)\in (1\pm \eps) \delta^z(x,y)$ for any $y\in X$.
	Hence
	\begin{equation}
	\label{eq:sen2}
	\max_{C\in \calO}\frac{d^z(x,C)}{\sum_{C\in \calO}d^z(y,C)} \geq \max_{C\in \calO}\frac{(1-\eps) \psi_x(C)}{(1+\eps)\sum_{\psi_y\in \Psi}\psi_y(C)}\geq \frac{1}{2} \max_{C\in \calO}\frac{ \psi_x(C)}{\sum_{\psi_y\in \Psi}\psi_y(C)}
	\end{equation}
	Combining with Inequalities \eqref{eq:sen1} and \eqref{eq:sen2}, we prove the first part.
	By Theorem~\ref{thm:totalsen}, we have $\sum_{x\in X} \pi_x = O(2^{O(z\log z)}k)$.
	On the other hand, $\theta_x\leq 2 n \pi_x$ by definition.
	Thus, we have $$\sum_{x\in X} \theta_x \leq 2 n \cdot \sum_{x\in X} \pi_x = O(2^{O(z\log z)}kn).$$
	This completes the proof of the lemma.
\end{proof}

\section{Proof of Claim~\ref{claim:min_equal}}
\label{section:min_equal}
\begin{claim}[restatement of Claim~\ref{claim:min_equal}]
	\label{claim:restate_min_equal}	$|\ranges(\calG_H)| \leq |\ranges(\calF_H)|^k$.
\end{claim}
\begin{proof}
	The proof is similar to that of~\cite[Lemma 6.5]{FL11}.\jian{mention this proof is similar to the proof of xxx in FL11 again.}\shaofeng{addressed.}
	By the definition of $V$, each range $\range(\calG_H, C,r)$ of $(\calG,\ranges(\calG))$ corresponds to a unique point set $\left\{y\in X\mid g_y\in \range(\calG_H, C,r)\right\}$.
	Then it suffices to show that for any $C\in [X]^k$ and $r\geq 0$,
	$$\left\{y\in X\mid g_y\in \range(\calG_H, C,r)\right\}=\bigcup_{x\in C} \left\{y\in X\mid f_y\in \range(\calF_H, x,r)\right\}.$$
	For any $C\in \calO$ and $r\geq 0$, if $g_y\in \range(\calG_H, C,r)$, we have
	\[
	g_y(C) = \min_{x\in C} \delta^z(x,y)/\theta_y \leq r.
	\]
	Let $x^*=\arg \min_{x\in C} \delta^z(x,y)$.
	We have $f_y(x^*)=\delta^z(x^*,y)/\theta_y\leq r$ which implies that $f_y\in \range(\calF_H, x^*,r)\subseteq \bigcup_{x\in C}\range(\calF_H, x,r)$.
	Therefore, we have $$\left\{y\in X\mid g_y\in \range(\calG_H, C,r)\right\}\subseteq \bigcup_{x\in C} \left\{y\in X\mid f_y\in \range(\calF_H, x,r)\right\}.$$
	It remains to prove $\bigcup_{x\in C} \left\{y\in X\mid f_y\in \range(\calF_H, x,r)\right\}\subseteq \left\{y\in X\mid g_y\in \range(\calG_H, C,r)\right\}$.
	If $f_y\in  \bigcup_{x\in C}\range(\calF_H, x,r)$, there must exist some $x^*\in C$ such that $f_y\in \range(\calF_H, x^*,r)$.
	It implies that
	\[
	g_y(C)=\min_{x\in C}\delta^z(x,y)/\theta_y \leq \delta^z(x^*,y)/\theta_y = f_y(x^*)\leq r.
	\]
	Hence $g_y\in \range(\calG_H, C,r)$.
	Thus, we have $$\bigcup_{x\in C} \left\{y\in X\mid f_y\in \range(\calF_H, x,r)\right\}\subseteq \left\{y\in X\mid g_y\in \range(\calG_H, C,r)\right\},$$ and this completes the proof.
\end{proof}

\section{Missing Proofs for Theorem~\ref{thm:robust_coreset}}
\label{section:proof_robust}
Consider a doubling metric space $M=(X,d)$.
As noted in Theorem~\ref{theorem:doubling_high_ball_dim}, if we let $\calG=\{d^z(x,\cdot) \mid x\in X\}$ as in Remark~\ref{remark:robust}, then the range space $(\calG, \ranges(\calG))$ may not have bounded dimension, which makes it hard to achieve a succinct $\frac{\alpha}{2}$-approximation.
Hence, we use the same idea in Section~\ref{sec:coreset}, i.e., to construct a random $(\eps/100z)$-smoothed distance function $\delta$ resultant from Corollary~\ref{corollary:weighted_z_ball}. 
Then for each $x\in X$, let $g_x(\cdot)$ be a function from $[X]^k$ to $\mathbb{R}_{\geq 0}$ such that $g_x(C)=\delta^z(x,C)$. 
Let $\calG:=\left\{ g_x\mid x\in X \right\}$.
Then by the same argument as in Lemma~\ref{lm:apprange}, we have the following lemma.

\begin{lemma}
	\label{lemma:approximation}
	Let $S$ be a uniformly independent sample of
	\begin{align*}
	\Gamma := O\bigg(\frac{k}{\alpha^2}(\mathrm{ddim}(M)\cdot \log (z/\eps) +\log k+\log \log (1/\tau))+\frac{\log (1/\tau)}{\alpha^2}\bigg)
	\end{align*}
	points from $X$.
	Then with probability at least $1-\tau$, $\calG_S=\left\{g_x \mid x\in S \right\}$ is an $\frac{\alpha}{2}$-approximation of the range space $(\calG, \ranges(\calG))$.
\end{lemma}

\begin{proof}
	
	The proof is almost identical to that in Lemma~\ref{lm:apprange}.
	For any $H\subseteq X$, recall that $\calG_H=\left\{g_x\mid x\in H \right\}\subseteq \calG$.
	We still want to apply Lemma~\ref{lemma:restate_weak_app}.
	By the same argument as in Lemma~\ref{lm:apprange}, we can show that for $T:\mathbb{N} \times \mathbb{R}_{\geq 0}$ such that
	\begin{align*}
	T(m, \gamma):= O\left(\frac{z}{\eps}\right)^{ O(k\cdot\DDim(M))}\cdot \log^k \frac{m}{\gamma}\cdot m^{6k},
	\end{align*}
	$\calG$ satisfies
	for any $H\subseteq V$ and $\gamma> 0$,
	\[
	\Pr[|\ranges(\calG_H)| \leq T(|H|,\gamma)] \geq 1 - \gamma.
	\]

	Now we are ready to apply Lemma~\ref{lemma:restate_weak_app}.
	Plugging in the values of $\Gamma$ and $T(2\Gamma, \tau/4)$ to Lemma~\ref{lemma:restate_weak_app}, we can verify that $\calG_S$ is an $\frac{\alpha}{2}$-approximation of the range space $(\calG, \ranges(\calG))$ with probability at least $1-\tau$.
	This completes the proof.
\end{proof}

Now we are ready to prove the main theorem.
\begin{theorem}[restatement of Theorem~\ref{thm:robust_coreset}]
	\label{thm:restate_robust_coreset}
	Let $M(X,d)$ be a doubling metric space (a $d$-dimensional Euclidean space resp.).
	Suppose $S$ is a uniform independent sample of $\Gamma$ ($\Gamma'$ resp.) points from $X$, where
	$$
	\Gamma := O\bigg(\frac{k}{\alpha^2}(\mathrm{ddim}(M)\cdot \log (z/\eps) +\log k+\log \log (1/\tau))+\frac{\log (1/\tau)}{\alpha^2}\bigg)
	$$
	and
	$$
	\Gamma' := O\bigg(\frac{1}{\alpha^2}(kd \log k+\log (1/\tau))\bigg).
	$$
	Then with probability at least $1-\tau$, $\calS$ is an $(\alpha,\eps)$-robust coreset ($(\alpha,0)$-robust coreset resp.)
	for the $(k,z)$-clustering problem with outliers.
\end{theorem}

\begin{proof} 
	%
	For the Euclidean space $\R^d$, by \cite{DBLP:journals/jcss/LiLS01}, we can construct an $\frac{\alpha}{2}$-approximation of $\calG$ defined as in Remark~\ref{remark:robust}, by taking $O(\frac{kd \log k}{\alpha^2})$ uniform samples from $X$.
	Then by Lemma \ref{rctech}, we complete the proof for the Euclidean space.
	
	For doubling metrics, by Lemma~\ref{lemma:approximation}, $\calG_S$ is an $\frac{\alpha}{2}$-approximation of $\calG$ with probability at least $1-\tau$.
	Then by Lemma~\ref{rctech}, $\calG_S$ is also an $(\alpha,0)$-robust coreset of $\calG$ with probability at least $1-\tau$.
	In the following, we condition on the event that $\calG_S$ is an $(\alpha,0)$-robust coreset of $\calG$.

	Now we fix a number $\gamma$ such that $\alpha<\gamma< 1-\alpha$ and a subset $C\in [X]^k$.
	Since $\calG_S$ is an $(\alpha,0)$-robust coreset of $\calG$, we have
	\[
	\frac{\calG^{-(\gamma+\alpha)}(C)}{|\calG|}\leq \frac{\calG_S^{-\gamma}(C)}{|\calS|} \leq  \frac{\calG^{-(\gamma-\alpha)}(C)}{|\calG|}.
	\]
	
	On the other hand, we have $d^z(x,y)\in (1\pm \eps/10)\cdot \delta^z(x,y)$ for any $x,y\in X$, by the definition of $\delta$.
	Then by the same argument as in the proof of Theorem~\ref{thm:coreset},
	\[
	\frac{\calG^{-(\gamma+\alpha)}(C)}{|\calG|}\in (1\pm \eps/10)\cdot \frac{\kdist_z^{-(\gamma+\alpha)}(X, C)}{|X|},
	\]
	\[
	\frac{\calG^{-(\gamma-\alpha)}(C)}{|\calG|} \in (1\pm \eps/10)\cdot \frac{\kdist_z^{-(\gamma-\alpha)}(X, C)}{|X|},
	\]
	\[
	\frac{\calG_S^{-\gamma}(C)}{|\calS|} \in (1\pm \eps/10)\cdot \frac{\kdist_z^{-\gamma}(S, C)}{|S|}.
	\]

	By the above inequalities, we conclude that
	\[
	(1-\eps)\cdot \frac{\kdist_z^{-(\gamma+\alpha)}(X, C)}{|X|} \leq \frac{\kdist_z^{-\gamma}(S, C)}{|S|} \leq (1+\eps)\cdot \frac{\kdist_z^{-(\gamma-\alpha)}(X, C)}{|X|},
	\]
	which completes the proof.
\end{proof}

\section{Distortion Lower Bound for Embedding Snowflake Doubling Metrics into $\ell_2$}
\begin{proposition}
	\label{prop:expander}
	Suppose $0 < z < 1$.
	There exists a metric space $M(X, d)$ such that any embedding of $(X, d^{z})$ into $\ell_2$ has distortion at least $\Omega((\DDim(M)^z)$.
\end{proposition}
\begin{proof}
	Let $k$ be a sufficiently large integer.
	Let $G$ be a constant degree (that is independent of $k$) expander graph with $2^k$ vertices. Let $M(X, d)$ be the shortest path metric of $G$, so $\DDim(M) \leq k$.
	It was shown in~\cite[Proposition 4.2]{DBLP:journals/combinatorica/LinialLR95} that, for any $n$ vertices constant degree expander graph, every embedding of its shortest path metric into $\ell_p$ has distortion at least $\Omega(\log{n})$, for any fixed $1 \leq p \leq 2$. Therefore, every embedding of $M$ into $\ell_2$ has distortion at least $\Omega(k)$.
	
	Observe that $\frac{d(x, y)}{d^z(x, y)} \leq \diam(X)^{1-z}$. This implies $(X, d^z)$ may be embedded into $(X, d)$ with distortion at most $\diam(X)^{1-z}$.\jian{grammar}\shaofeng{fixed.}
	Since it is well known that the diameter of an $n$ vertices constant degree expander graph is $O(\log{n})$, we have $\diam(X) \leq O(k)$.
	So there exists an embedding of $(X, d^z)$ into $(X, d)$ with distortion at most
	$O(k^{1-z})$.
	
	Therefore, if for the contrary that there exists an embedding of $(X, d^z)$ into $\ell_2$ with distortion at most $O(\DDim(M)^z) \leq O(k^z)$, then this would imply an embedding of $(X, d)$ into $\ell_2$ with distortion at most $O(k)$. This leads to a contradiction.\jian{grammar}\shaofeng{fixed.}
\end{proof}


\end{document}